\pdfoutput=1
\documentclass[11pt, a4paper, oneside, reqno]{elsarticle}
\usepackage{booktabs}
\usepackage[english]{babel}
\usepackage[T1]{fontenc}
\usepackage{lmodern}
\usepackage[usenames,dvipsnames]{xcolor}
\usepackage[colorlinks=true,linkcolor=NavyBlue,urlcolor=RoyalBlue,citecolor=PineGreen,
hypertexnames=false]{hyperref}
\usepackage{amsthm}
\usepackage{amsmath}
\usepackage{amssymb}
\usepackage{amsfonts}
\usepackage{mathtools}
\usepackage{stmaryrd}
\usepackage{bm}
\usepackage{mathbbol}
\usepackage{tikz}
\usepackage{tikz-cd}
\usepackage{subfigure}
\usepackage{verbatim}
\usepackage[margin=0.7in]{geometry}
\usepackage[font=small]{caption}
\usepackage{enumerate}
\usepackage{adjustbox}
\usepackage{svg}
\usepackage{mathpazo}

\newcommand{\lmss}[1]{\textrm{\normalfont{{\fontfamily{lmss}\selectfont #1}}}}

\newcommand{\sbt}{\,\begin{tikzpicture}[baseline=(X.base)]%
		\node[draw, fill,black,circle, inner sep=1pt] at (0,0.1) {};
		\node[circle,inner sep=0pt,outer sep=0pt] (X){$\ $};
	\end{tikzpicture}%
	\,
}

\newcommand{\freeunitary}{{U^{\langle n \rangle}}}
\newcommand{\mfreeunitary}{{\lmss{U}^{\langle n \rangle}}}

\newcommand{\munitaryfd}{\lmss{U}_{N}^{\mathbb{K}}}

\newcommand{\udualgroup}{\mathcal{O}\langle n \rangle}

\newcommand{\runitaryalg}{\mathcal{R}\mathcal{O}\langle n \rangle}

\newcommand{\runitaryfdN}{U^{\mathbb{K}}_{d_{N}}}

\newcommand{\mdnoise}{\lmss{W}}

\def\Alg{{\rm Alg}}
\def\biMod{{\rm biMod}}

\def\rCoMod{{\rm rCoMod}}
\def\lCoMod{{\rm lCoMod}}
\def\coModAlg{{\rm coModAlg}}
\def\freem{\mathrel{\dot\sqcup}}
\def\free{\mathrel{\sqcup}}

\def\Hom{{\rm Hom}}
\def\id{{\rm id}}
\def\Prob{{\rm Prob}}

\def\path{{\sf P}}

\def\seqlasso{{\overset{\circ}{\mathcal{P}}(\lmss{lassos})}}

\def\Obj{{\mathrm{Obj}}}
\def\Vect{{\mathrm{Vect}}}

\def\C{{\mathbb C}}
\def\R{{\mathbb R}}
\def\H{{\sf H}}
\def\F{{\mathcal{F}}}
\def\U{{\mathbb{U}}}
\def\K{{\mathbb{K}}}

\newtheorem{theorem}{Theorem}
\newtheorem{remarque}{Remark}
\newtheorem*{notation}{Notation}
\newtheorem{example}{Example}

\newtheorem{assumptions}{Assumptions}
\newtheorem{construction}{Construction}
\theoremstyle{plain}
\newtheorem{proposition}[theorem]{Proposition}
\newtheorem{corollaire}{Corollary}
\newtheorem*{theorem*}{Theorem}

\newtheorem{lemma}[theorem]{Lemma}
\theoremstyle{definition}
\newtheorem{definition}[theorem]{Definition}

\author{Nicolas Gilliers\fnref{email}\fnref{onleave}}
\address{Sorbonne Universit\'e, Sorbonne Paris Cit\'e, CNRS,\\ Laboratoire de Probabilit\'es Statistique et Mod\'elisation, LPSM, F-75005 Paris, France}
\fntext[email]{{nicolas.gilliers@gmail.com}}
\fntext[onleave]{Now (2022) at Institut de Mathématiques de Toulouse (IMT), University Paul Sabatier, Toulouse, France}
\title{Quantum Holonomy Fields}

\tikzset{commutative diagrams/.cd,
arrow style=tikz,
diagrams={>={Computer Modern Rightarrow[length=5pt,width=3pt]}},
}

\begin{document}
\begin{abstract}
	We investigate lattice and continuous quantum gauge theories on the Euclidean plane with a structure group that is replaced by a $H$-algebra. $H$-algebras are non-commutative analogues of groups and contain the class of Voiculescu's dual groups. We are interested in non-commutative analogues of random gauge fields, which we describe through the random Holonomy that they induce. We propose a general definition of a Holonomy Field with symmetries displaying the structure of a $H$-algebra and construct such a field starting from a quantum L\'evy process on a $H$-algebra in the category of probability spaces. We call them Quantum Holonomy Fields. We also consider the more abstract case of an $H$-algebra in a given algebraic category. This yields the notion of a Categorical Holonomy Field. As an application, we define higher dimensional generalizations of the so-called Master Field on the plane.
\end{abstract}
\maketitle

\tableofcontents

\section{Introduction}

\subsection{Background}
The present work intends to define gauge theories on certain non-commutative analogues of the space of functions on a (topological, Lie) group. The key notions are the ones of \emph{Random Planar Holonomy Fields} introduced by T.Lévy in \cite{levy2008two} and \emph{$H$-algebras} introduced by J. Zhang in \cite{zhang}. To the risk of being negative, for the knowledgeable reader, this does not include Drinfeld or Woronowicz (compact) quantum groups\footnote{the coproduct of an $H$-algebra takes values in the free product of $H$ with itself instead of the tensor product.} \cite{woronowicz1998compact}, see also \cite{meusburger2021hopf, buffenoir2002hamiltonian} for gauge theories over (Ribbon) Hopf algebras.

In the sequel, we drop the adjective \emph{planar}. In a nutschell, a \emph{Random Holonomy Field} is a group-valued random process indexed by \emph{rectifiable loops} (loops with finite length, well-approximated by piecewise linear loops, in a very vague sense) on the plane, instead of time points as is usual for stochastic processes. One may think of them as Lévy processes indexed by loops since their \emph{increments} have \emph{stationary distribution} and are \emph{classically independent}. Of course, each of these terms -- increments, stationarity and distributions should be properly defined. For example, the stationarity of their distributions relates to invariance by area-preserving homeomorphisms acting on loops. This will be discussed at length in the forthcoming work.
Let us mention that a Brownian Planar Random Holonomy Field is a \emph{quantized Yang-Mills} field on the Euclidian plane, we will come back to this point further below in this introduction.

The theory of \emph{Quantum Lévy processes} builds on a non-commutative equivalent to the notion of classical Lévy processes. It has been largely developed in the work of Schürmann \cite{schurmann} and Franz \cite{franz2004theory}. By analogy, classical probabilities are generative of developments for the non-commutative theory. In particular, basic results on classical Lévy processes extend partially to quantum Lévy processes, such as the Lévy-Khinchine classification of their generators. The present work originates from analogy: if Random Holonomy Fields are loop-indexed Lévy processes, what is a loop-indexed Quantum Lévy process, which will be called a \emph{Quantum Holonomy Field}? This work partially answers this question with a proposal for a definition of a Quantum Holonomy Field indexed by piecewise geodesic loops on the plane. To the risk of being repetitive, in this work \emph{loops are drawn on the plane and are piecewise \emph{affine}: they are made up of a finite number of segments concatenated together}. This is \emph{not} the appropriate index set; it should be completed to complete metric space, the space of all rectifiable loops, for a distance introduced in \cite{levy2008two}. The extension of our general constructions to rectifiable loops will be discussed elsewhere.

For this work, we very much strongly rely on \cite{cebron2017generalized} wherein the authors introduce \emph{generalized Master Fields}, which will become a subclass of the Quantum Holonomy Fields defined in this work. In our terminology, these correspond to \emph{$\mathcal{O}\langle 1 \rangle$ \emph{free} Quantum Holonomy Field} or \emph{$\mathcal{F}(\mathbb{U}(N,\mathbb{C}))$ \emph{classical} Quantum Holonomy Field}. The words $\emph{free}$ and $\emph{classical}$ refer to types of (non-commutative) independences; \emph{Voiculescu's freeness} \cite{voiculescu1985symmetries} and the usual independence from classical probability theory. We have used the notation $\mathcal{O}\langle 1 \rangle$ for the first dual Voiculescu group. We propose a far-reaching generalization of the notion of the generalized Master field. In fact, as a first take, we will replace $\mathcal{O}\langle 1 \rangle$ with any $H$-algebras. We will be even more general, and work with $H$-algebras in any algebraic categories. We will also go beyond classical and free independence and allow most of the notions of non-commutative independence. Our main source of examples (see Section \ref{examples}) will be the class of $\mathcal{O}\langle n \rangle$ $H$-algebras, or \emph{dual Voiculescu groups} and free Lévy processes on these dual groups (as the initial data to define a $\mathcal{O}\langle n \rangle$ (free) Quantum Holonomy Fields).

Let us first start with some general remarks about Yang-Mills theory and introduce fundamentals about \emph{Random Holonomy Fields} \cite{levy2008two}. In 1954 Yang and Mills initiated in \cite{yang1954conservation} a theory at the roots of modern particle physics, \emph{known nowadays as the Yang-Mills theory}. This theory is the culminating point of a process started decades beforehand: physicists do not longer consider symmetries as a means to reduce computations, but rather as being responsible for many of the physical interactions. Contributions to the development of the Yang-Mills theory come from both the community of physicists and mathematicians, see \cite{gopakumar1995mastering,gross1988maxwell, driver1991two,singer1995master,sengupta1997gauge}.

An output of the Yang-Mills theory is a Lagrangian, which is a dynamic -- a particle -- the input being, essentially, a group of symmetries $G$ (called at times the \emph{gauge group}) and a \emph{space-time} $\Sigma$ on which matter particles evolve. The $\mathbb{U}(N)$--Yang-Mills theory, the theory with $G=\mathbb{U}(N,\mathbb{C})$ has now a long history, see the survey \cite{hooft200550}. The \emph{quantum Euclidian Yang-Mills theory} \cite{driver1991two} describes informally a \emph{connection} chosen at random over a \emph{principal bundle} with a compact structure group $G$. In two dimensions, with $\Sigma=\mathbb{R}^2$ or a compact surface with a Riemannian metric, the theory has been extensively studied, see \cite{levy} for a historical account. 
The Gordian knot of quantizing is a rigorous construction of a probability measure on an infinite dimensional space. Ideally, we would, on the space of all connections, construct a Gaussian measure. This is challenging since, in infinite dimensions, there are no translations invariant Lebesgues measure (here translation means action by the gauge group). Quantizing Yang-Mills theory on a four-dimensional Lorentzian space, even on a Euclidian three-dimensional space-time, is still an open question, see however \cite{chandra2022stochastic} for recent progress about the three-dimensional theory. The point-of-view adopted in \cite{driver1991two,levy2008two} and of the current article focuses on making sense and studying the \emph{expectation of Random holonomies}, also called \emph{Wilson loops} and defining Random Holonomy Fields.
By using this approach, in two dimensions, the quantization of Yang-Mills theory given a (compact) surface $\Sigma$ and a compact Lie group $G$ of symmetries, has been rigorously addressed by T. L\'evy in his thesis\cite{levy2008two}. By using gauge fixing, it is even possible to use ``Gaussian calculus'', this opens the possibilities of rigorous calculations in the discrete and continuous limit. This is the point-of-view developed by L. Gross, C. King and A. Sengupta in \cite{GROSS198965} and B. Driver \cite{driver1989ym2}. More in detail, a connection is most accurately ``described'' through its Holonomy or \emph{parallel transport} along any closed loop drawn on the plane. This is an element of the structure group $G$ attached to each loop on the plane. Under the (so far) non-existing Yang-Mills measure $\mu_{\mathrm{YM}}$, a random connection yields a random Holonomy (a random matrix if $G$ is a matrix group) around any closed loop. 
Hence, informally, the quantized Yang-Mills theory yields a Random Holonomy Field. The distribution of these random holonomies, (the joint distributions of the random matrices associated with any finite sequence of loops on the plane when $G=\mathbb{U}(N,\mathbb{C})$) is \emph{invariant by conjugation by any element of the structure group $G$}. It is also \emph{invariant by the action of any area-preserving morphism}. For example, the Holonomy of a \emph{simple loop} (it has no self-intersections) only depends on the area enclosed by the loop. Finally, holonomies around loops enclosed in disjoint discs are \emph{classically independent}.


In \cite{t1993planar} G. ’t Hooft observed that quantities of interest in the $\mathbb{U}(N,\mathbb{C})$ become simpler in the limit $N$ tends
to infinity.
The high-dimensional limit of the $\mathbb{U}(N,\mathbb{C})$ was
extensively studied by physicists, see for example \cite{kazakov1981wilson}, and the idea emerged that
there should be a universal deterministic large $N$ limit to a broad class of matrix models. This limit was named the \emph{Master Field}.
In our terminology, this Quantum Holonomy Field is associated, in a sense that will be made precise in the present work, to the $H$-algebra $\mathcal{O}\langle 1 \rangle$ and to the \emph{free unitary Brownian motion} \cite{biane1998free}, a free Lévy process on $\mathcal{O}\langle 1 \rangle$. 

Finally, there exists a correspondence between Random Holonomy Fields classical, Lie group valued, Lévy processes, whose distribution satisfies some invariance. This correspondence will be expounded in Section \ref{sec:defiZhangHF}. It extends to the non-commutative setting: the distribution of a Quantum Holonomy Field is determined by a Quantum Lévy process associated with it. It is also possible to follow the reverse path: any Quantum Lévy process with non-commutative distribution satisfying some invariance yields a Quantum Holonomy Field. 
 Our contributions are discussed in Section \ref{sec:contributions} where the reader will find a detailed outline of each Section of the present work. 
\subsection{Contributions}
\label{sec:contributions}
\begin{enumerate}
\item The present work introduces the notion of Quantum Holonomy Field, more general than the notion of generalized Master Field introduced in \cite{cebron2017generalized}. A Quantum Holonomy Field in particular depends on the choice of a $H$-algebra (which plays the role of the gauge group, the group of symmetries) and of a Quantum Lévy process over the same $H$-algebra whose distribution displays braid-invariance. The generalized Master Fields defined in \cite{cebron2017generalized} correspond to the special case of a $H$-algebra equal to $\mathcal{O}\langle 1 \rangle$ or $\mathcal{F}(\mathbb{U}(N,\mathbb{K}))$. We give a rather abstract definition of a \emph{categorical Holonomy field} to allow even more non-commutative independences to enter in the notion of Quantum Holonomy Field than the Muraki's five \cite{muraki2002classification}, in particular multi-states independences, e.g conditional freeness \cite{bozejko1996convolution}.

\item We construct in particular \emph{higher dimensional free Master Fields}, based on the series $\mathcal{O}\langle n \rangle$ of dual Voiculescu group $\mathcal{O}\langle n \rangle$. Concretely, the distribution of this field is the limit in high-dimensions of observables of the Quantum Holonomy Field over $\mathcal{F}(\mathbb{U}(N, \mathbb{C}))$ associated with a Brownian diffusion on $\mathbb{U}(N,\mathbb{C})$ (it is called the Yang-Mills field further in the text), the commutative Hopf algebra of function of the group $\mathbb{U}(N,\mathbb{C})$ invariant with respect to a sub-group of $\mathbb{U}(N,\mathbb{C})$ isomorphic to a lower dimensional unitary group. We propose Makeenko-Migdal equations for this field in Section \ref{sec:perspectives} but we do not prove them.

\item The abstract setting we develop for constructing Quantum Holonomy Fields allows us to define and provide an example of a Quantum Holonomy Field with symmetry group an \emph{amalgamated $H$-algebra}: a $H$-algebra in a category of bimodules algebras over a fixed algebra. An example of such a field is given by considering the expectation of observables of the Yang-Mills field over $\mathbb{U}(N)$ invariant with respect to \emph{a product of lower dimensional unitary groups of possibly different dimensions}.

\end{enumerate}

Once again, we emphasize that our contribution is purely \emph{algebraic}; the fields defined and constructed are indexed by piecewise linear loops on the plane. In \cite{cebron2017generalized}, the constructed fields are indexed by rectifiable loops on the plane.

\subsection{Organization of the paper}
To help the reader find their way through the paper, we have chosen to summarize the content of each section and to add pointers to the main definitions and results of the present work. The first two sections, besides the introduction, should be familiar to anyone comfortable with the categorical formulation of non-commutative independence, following the work of Schürmann, Ben-Gorbhal and U. Franz \cite{franz2004theory} among others. We added Annexes wherein facts about algebraic categories are recalled to make the paper self-contained and missing proofs of propositions, be they original to the present work or not, are provided.
\subsubsection{Outline Section \ref{sec:zhangalgebras}}
In Section \ref{sec:zhangalgebras} we define what $H$-algebras are. They should be distinguished from Woronowicz quantum groups and \emph{easy quantum groups}\footnote{https://www.math.uni-sb.de/ag/speicher/weber/QuickIntroductionEasyQGHomepage.pdf} \emph{à la} Banica and Speicher \cite{banica2009liberation}, though they possess the same set of structural morphisms to the exception that they are objects in an \emph{algebraic category}(see the Annexes). They are an answer to a problem of representation of a functor from the category of (possibly non-commutative) algebras to the category of groups. This was first pointed out by J. Zhang \cite{zhang}. This representation property is central to the present work.
We provide many examples of $H$-algebras.
In this section, we develop a graphical calculus for $H$-algebras akin to the graphical calculus for Hopf algebras (even in commutative Hopf algebras); this helps to improve the readability of the proofs of the propositions in this section.
Secondly, we study the category of comodules over a $H$-algebra in an abstract algebraic category. The main results are stated in Propositions \ref{lemma_alge_alg_alg_star} and \ref{lem:zhangzhang}: we prove that comodules over a $H$-algebra form an algebraic category and that a $H$-algebra always comes equipped with the structure of a comodule-algebra over itself, the {\em conjugation co-action} which is an abstract generalization to $H$-algebras of the dual conjugation action of a group on its space of functions.


\subsubsection{Outline Section \ref{noncommasterfields} }
This section is the core of the present work. Based on the work \cite{cebron2017generalized} wherein, we recall, classical and free Master Fields are defined, we give definitions of a Quantum Holonomy Field and of Categorical Holonomy Field based on a $H$-algebra and given a certain notion of non-commutative independence. We distinguish two cases; the chosen $H$-algebra is in a category of (operator-valued) probability spaces, this is the case addressed in Definition \ref{definition_master_field}(we use the terminology Quantum Holonomy Field in this case) or the chosen $H$-algebra is an object in an abstract algebraic category, see Definition \ref{definition_master_field_categorical} (we use the terminology Categorical Holonomy Field in this second case). In this respect, Definition \ref{definition_master_field} strictly comprises Definition \ref{definition_master_field_categorical}.
In the remaining part of this section, we address the problem of constructing a Categorical Holonomy Field (given a notion of independence) based on a $H$-algebra starting from a direct system of objects (akin to a probability space but located in an abstract category) indexed by \emph{lassos} drawn on the plane and pin-point the required invariances of this family that make possible the construction of a Categorical Holonomy Field from this system.
To the knowledgeable reader, the appearance of \emph{lassos} should not be surprising, since certain subsets of such loops yield free (in the sense that no algebraic relations are fulfilled) families in the group of reduced loops on the plane.
Theorem \ref{maintheoremzhanghol} is the central result of this work: a Categorical Holonomy Field can be built from any Quantum Lévy process with a braid-invariant distribution (this notion is recalled in the concerned section)

\subsubsection{Outline Section \ref{examples}}
In this section, we apply our Theorem \ref{maintheoremzhanghol}. First, we define for each integer $n\geq 1$, the \emph{free $n$-dimensional Brownian motion}. It is a free Quantum Lévy process over the $n$ free unitary dual group, the solution of a free stochastic differential equation. Alternatively, the free $n$-dimensional Brownian motion is the limit in high-dimensions and in non-commutative distribution of \emph{square blocks extracted from a Brownian motion on the unitary group}. We show that this process has the required invariance to apply Theorem \ref{maintheoremzhanghol}; the Quantum Holonomy Field we obtain is called the free Master Field of dimension $n$.
In a second example, we build a Quantum Holonomy Field associated with the limiting non-commutative distribution of \emph{rectangular blocks} extracted again from a Brownian motion on the unitary group.

\

{\bf Acknowledgements: }
I would like to thank Prof. Thierry Lévy for his help and support in developing this work, Dr. Franck Gabriel for his remarks on the first version of the article and the anonymous referee whose acute comments and advice have greatly pushed the article to a publishable state. The author acknowledges support from NTNU, (Trondheim, Norway), the University of Greifswald (Greifswald, Germany) and the ANR STARS\footnote{Project-ANR-20-CE40-0008} (IMT, Toulouse, France).

\section{$H$-algebras and categorical independence}
\label{sec:zhangalgebras}

{

In the first part of this Section, we review $H$-algebras, starting with a definition. Recall that all our algebras are over the field of complex numbers, unital and associative. In the second part, we review the categorical formulation (see \cite{franz2004theory}) of non-commutative independences.

}

The definition of $H$-algebras is very similar to that of Hopf algebras, the main difference being that the coproduct takes its values not in the monoidal product, but in the free product of the algebra with itself. There is therefore a variety of notions of $H$-algebras, corresponding to various notions of free products, or more properly speaking of \emph{categorical coproducts}. Since we will use several categorical coproducts, we prefer not to specify a particular one at this point, and we choose instead to adopt the more abstract point of view of {\em algebraic categories}, see Section \ref{sec:algcategories} and Section VI of the monograph\footnote{\url{http://www.tac.mta.ca/tac/reprints/articles/17/tr17.pdf}}\cite{adamek2004abstract}.
The main point of this Section is to swiftly introduce algebraic categories, monoidal categories and $H$-algebras to the readers.
We use the language of category theory and therefore, for the reader not acquainted with it, we refer him to the monograph \cite{adamek2004abstract} for a detailed exposition. In particular,
equalities between morphisms are expressed as a commutativity property of some \emph{diagrams}. A \emph{diagram} is a directed graph with objects labelled vertices and morphisms labelled edges. We say that a diagram is commutative if the composition of morphisms along any two directed paths with the same source and the target yield the same result. For clarity, we will frequently drop labels of edges if it is clear from the context how morphisms are assigned to edges of a given diagram.

\subsection{Algebraic categories}
\label{sec:algcategories}
Recall that in a category $\mathcal C$, an object $k$ is called an \emph{initial object} if, for every object $A$, there is exactly one morphism in $\Hom_{\mathcal{C}}(k, A)$. Recall also that a coproduct between two objects $A$ and $B$ is the data of an object $C$ and two morphisms $\iota_{A}:A\to C$ and $\iota_{B}:B\to C$ such that for any object $D$ and any two morphisms $f:A\to D$ and $g:B\to D$, there exists a morphism $h:C\to D$ such that $f=h\circ \iota_{A}$ and $g=h\circ \iota_{B}$. At times, it will be necessary to refer to $C$ in the notations and we will write $\iota_A^C$ and $\iota_B^C$ for $\iota_A$ and $\iota_B$ respectively. 

If a coproduct of two objects $A$ and $B$ exists, it is unique up to isomorphism and it is denoted by $A\free B$. Moreover, with the current notation, the morphism $h$ is denoted by $f\freem g: A\free B\to D$. The dot in the symbol $\freem$ indicates that the elements of $D$ obtained by applying $f$ and $g$ to the elements of $A$ and $B$ are multiplied in $D$. There is also a natural way of combining two morphisms $f:A\to D_{1}$ and $g:B\to D_{2}$ into a morphism $f\free g :A\free B \to D_{1} \free D_{2}$, by first forming $\iota_{D_{1}}\circ f:A\to D_{1}\free D_{2}$ and $\iota_{D_{2}}\circ g : B \to D_{1}\free D_{2}$ and then setting $f\free g=(\iota_{D_{1}}\circ f)\freem(\iota_{D_{2}}\circ g)$.

\begin{definition}[Algebraic category]
	\label{def:algcat}
	An algebraic category is a category with an initial object in which any two objects admit a coproduct.
\end{definition}
The product of a sequence of objects $(A_1,\ldots, A_n)$ is well defined. In particular, if $n=3$, $A_1 \sqcup A_2 \sqcup A_3 = (A_1\sqcup A_2)\sqcup A_3 = A_1 \sqcup (A_2\sqcup A_3)$ and the injections $\iota^{A_1\sqcup A_2 \sqcup A_3}_{A_1},\iota^{A_1\sqcup A_2 \sqcup A_3}_{A_2},\iota^{A_1\sqcup A_2 \sqcup A_3}_{A_3}$ satisfy
$$
\iota_{A_1,A_2}^{(A_1\sqcup A_2),A_3}\iota^{A_1,A_2}_{A_1} = \iota^{A_1\sqcup A_2 \sqcup A_3}_{A_1} = \iota^{A_1,(A_2\sqcup A_3)}_{A_1}
$$
See for example the first pages of \cite{zhang}
The first example of an algebraic category is provided by the category of complex unital commutative algebras endowed with the usual tensor product of algebras. The initial object is $\mathbb{C}$.
As the second example of an algebraic category, let $\Alg$ be the category of complex unital associative algebras (with algebra morphisms as morphisms of the category $\Alg$). The algebra $\C$ is an initial object of this category (here we use the fact that the algebras are unital). Moreover, given two algebras $A$ and $B$, we can form the algebra $A\sqcup B$ freely generated by $A$ and $B$, the units of $A$ and $B$ being identified with the unit of $A\free B$. This algebra can be described as a quotient of the tensor algebra:
\[A\free B=T(A\oplus B)/(a\otimes a'-aa', b\otimes b'-bb', 1_{A}-1,1_{B}-1 : a,a'\in A, b,b'\in B)\]
Concretely, $A\free B$ is the vector space of all formal linear combinations of alternating words in elements of $A$ and $B$. Any occurrence of the units of $A$ or $B$ in one of these words can be ignored, and the multiplication of words is given by concatenation followed, in the case where they belong to the same algebra, by the multiplication  of the last letter of the first word with the first letter of the second.
Then $\free$ is a coproduct in the category $\Alg$ so that $\Alg$ is an example of an algebraic category.

We will consider the following other examples of algebraic categories.
\begin{example}
\label{ex:algcat}
\begin{enumerate}[\indent 1.]
	\item The category $\Alg^{\star}$ of involutive (complex unital associative) algebras. The initial object of this category is still $\C$. The coproduct of two objects $A$ and $B$ of $\Alg^{\star}$ is, as an algebra, their coproduct in $\Alg$. Moreover, $A\free B$ is endowed with the unique antimultiplicative involution which extends those of $A$ and $B$. Concretely, the involution of $A\sqcup B$ reverses the order of the letters in a word and transforms each letter according to the involutions of $A$ and $B$.

	\item \label{enum:alg*} The category $\Alg(R)$ of (complex unital associative) algebras endowed with a structure of bimodule over a fixed unital algebra $R$. The initial object in this category is $R$. The coproduct of two objects $A$ and $B$ is the coproduct of the category $\Alg$ with amalgamation over $R$. It can be described as
	      \begin{align*}
		      \hspace{0.8cm}A\free_{R} B & =(R\oplus \bigoplus_{n\geq 1}T^{n}(A\oplus B))/(ar\otimes r'a'-arr'a', br\otimes r'b'-brr'b', ar\otimes b-a\otimes rb, \\
		                                 & \hspace{1.7cm} br\otimes a- b\otimes ra,r1_{A}r'-rr',r1_{B}r'-rr' : a,a'\in A, b,b'\in B,r,r'\in R).
	      \end{align*}
	      In plain words, it is the free product of $A$ and $B$ in which multiplication by elements of $R$ can circulate between neighbouring factors. Similarly, we denote by $\Alg^{\star}(R)$ the category of bimodules-algebras over an unital involutive associative algebra $R$, endowed with an involution compatible with the two actions of $R$ and the involution on $R$.
	\item The category $\mathbb{Z}_{2}$-Alg of complex unital associative $\mathbb{Z}_{2}$-graded algebras is algebraic. A $\mathbb{Z}_{2}$-graded algebra is the data  of a complex unital algebra and an unipotent morphism on that algebra. Morphisms are unital morphisms of algebras that preserve the grading. The free product of two graded algebras $(A,D_{A})$ and $(B,D_{B})$, is as an algebra $A \sqcup B$ and $D_{A\free B} = D_{A}\free D_{B}$. The initial object is $\mathbb{C}$ endowed with the trivial grading.

	\item \label{freeprodgrp} The category $\mathcal{G}rp$ of groups. The coproduct is the free product of groups and the initial object is the group having only one element. If $G$ and $H$ are groups, a word in $G$ and $H$ is a product of the form ${\displaystyle s_{1}s_{2}\cdots s_{n},}$ where each $s_{i},~i\leq n$ is either an element of the group $G$ or an element of the group $H$. Such a word may be reduced using the following operations:
	      \begin{enumerate}[\indent 1.]
		      \item Remove an instance of the identity element (of either $G$ or $H$).
		      \item Replace a pair of the form $g_{1}g_{2}$ by its product in $G$, or a pair $h_{1}h_{2}$ by its product in $H$, with obvious notation.
	      \end{enumerate}
	      Every reduced word is an alternating product of elements of $G$ and elements of $H$. The free product $G \sqcup H$ is the group whose elements are the reduced words in $G$ an $H$, under the operation of concatenation followed by reduction.
	\item The category $\biMod(R)$ of bimodules over a fixed unital algebra $R$ can be endowed with a coproduct with injections. Let $A$ and $B$ two $R$-bimodules. The product $A\free B$ in $\biMod(R)$ is, as a vector space, isomorphic to the sum of vector spaces $A \oplus B$. The $R$ bimodule structure on $A\oplus B$ is the sum of the two structures:
	      \begin{equation*}
		      r(a \oplus b)r^{\prime} = rar^{\prime} \oplus rbr^{\prime}, a\in A,~ b\in B,~ r,r^{\prime} \in R.
	      \end{equation*}
	      The initial object is again $R$.
	\item The category $\coModAlg(H)$ of comodule-algebras over a $H$-algebra $H$, which we will describe later (see Section \ref{sec:ZhangcoMod}).
\end{enumerate}
\end{example}
\subsection{$H$-algebras} 
We can now give the definition of a $H$-algebra in an algebraic category.
\label{zhangalgebras}
\begin{definition}[$H$-algebra \cite{zhang}]\label{def:Zhang}
	Let $\mathcal{C}$ be an algebraic category with initial object $k$ and coproduct $\free$.
	A $H$-algebra of $\mathcal{C}$ is a quadruplet $(H,\Delta,\varepsilon,\mathcal{S})$ where
	\begin{enumerate}[\indent 1.]
		\item $H$ is an object of $\mathcal C$,
		\item $\Delta : H \to H\free H$ is a morphism of $\mathcal C$ such that $(\Delta\free \id_{H})\circ \Delta=(\id_{H}\free \Delta)\circ \Delta$,
		\item $\varepsilon:H\to k$ is a morphism of $\mathcal C$ such that $(\varepsilon \freem \id_{H})\circ \Delta=\id_{H}=(\id_{H}\freem \varepsilon)\circ \Delta$,
		\item $\mathcal{S}:H\to H$ is a morphism of $\mathcal C$ such that $(\mathcal{S}\freem \id_{H})\circ \Delta=\eta\circ \varepsilon=(\id_{H}\freem \mathcal{S})\circ \Delta$, where $\eta$ is the unique morphism from $k$ to $H$.
	\end{enumerate}
\end{definition}

\begin{remarque}
Definition \ref{def:Zhang} is, formally, very similar to the definition of Hopf algebras. More succinctly, replacing $\mathcal C$ by the category of (complex unital) algebras, $k$ by $\C$ and $\free$ by the monoidal product in Definition \ref{def:Zhang} makes the definition \ref{zhangalgebras} closer the one of a Hopf algebra. However, the monoidal product is not a coproduct in the category of algebras, and Hopf algebras are not a special case of $H$-algebras, and the antipode of a Hopf algebra is \emph{not} algebra morphism; it is an antimorphism. Instead, if one chooses for the category $\mathcal{C}$ the category of (complex unital) \emph{commutative} algebras, then the tensor product of two algebras yield a coproduct and a $H$-algebra is a commutative Hopf algebra. 
\end{remarque}
\begin{remarque}
In an algebraic category, any object $A$ is an ``algebra'' : it comes equipped with an associative product $m_A:A\sqcup A\to A$, we refer to \cite{zhang} for the meaning of associative in this case. The product $m_A$ is uniquely caracterized by $m_A \circ \iota_1 = \id_A$, $m_A \circ \iota_2 = \id_A$ where $\iota_1$ and $\iota_2$ are the two injection of $A$ into $A\sqcup A$.

From Proposition 2.1 point 4 in \cite{zhang}, the structural morphisms $\Delta ,\varepsilon, S$ of a $H$-algebra are all ``algebra'' morphisms, in particular $m_{A\sqcup A} \circ \Delta = m_{A}$.
\end{remarque}
\begin{remarque}
\label{rk:tensorproductnotcoproduct}
It is instructive to understand the reason why the tensor product between two (non-commutative algebras) does not yield a coproduct on the category $Alg$ (see before Example \ref{ex:algcat}). Consider indeed two \emph{unital} algebras $A$ and $B$ with unit $\eta_A$ and $\eta_B$ respectively and two morphisms $f:A\to D$ and $g:B\to D$. Let $\iota_{A}:A\to A\otimes B$ and $\iota_{B}:B\to A\otimes B$ be the injections,
$$
\iota_A(a)=a\otimes \eta_B,\quad \iota_{B}(b)=\eta_{A} \otimes b,\quad a\in A,b\in B.
$$
Should there exist a morphism $h:A\otimes B\to D$ such that $f=h\circ \iota_{A}$ and $g=h\circ \iota_{B}$, the relation
\[\iota_{A}(a)\iota_{B}(b)=(a\otimes 1_{B})(1_{A}\otimes b)=a\otimes b=(1_{A}\otimes b)(a\otimes 1_{B})=\iota_{B}(b)\iota_{A}(a)\]
would impose the equalities $h(a\otimes b)=f(a)g(b)=g(b)f(a)$, the second of which has no reason of being satisfied unless, of course, $D$ is commutative.
\end{remarque}


We continue with examples of $H$-algebras.

\begin{example}
\label{ex:halgebra}
\begin{enumerate}[\indent 1.]
	\item Let $V$ be a complex vector space. We claim that $V$ is a $H$-algebra in the algebraic category $(\Vect_{\mathbb{C}},\oplus,\{0\})$. In fact, define:
	      \begin{equation*}
		      \Delta(x)=x \oplus x = \iota_1(x) + \iota_2(x) \in V\oplus V,~\mathcal{S}(x)=-x,~\varepsilon(x)=0.
	      \end{equation*}
	      It is easy to check that, with these definitions, $(V,\Delta,\mathcal{S},\varepsilon)$ is a $H$-algebra.
	\item Let $n \geq 1$ an integer. The Dual Voiculescu group $\mathcal{O}\langle n \rangle$ is the involutive unital associative algebra generated by $2n^{2}$ variables; $u_{ij},~u_{ij}^{\star}~i,1 \leq j\leq n$ subject to the relations:
	      \begin{equation*}
		      \sum_{k=1}^{n} u_{ik}u_{jk}^{\star} = \delta_{ij},~\sum_{k=1}^{n} u^{\star}_{ki}u_{kj}= \delta_{ij}, ~ 1 \leq i,j \leq n.
	      \end{equation*}
	      The dual Voiculescu group is turned into a $H$-algebra $(\mathcal{O}\langle n \rangle, \Delta,\varepsilon,\mathcal{S})$ if we define the structural morphisms by:
	      \begin{equation*}
		      \mathcal{S}(u_{ij})= u_{ji}^{\star},~\Delta({u_{ij}}) = \sum_{k=1}^{n} u_{ik}|_{1}u_{kj}|_{2},~ \varepsilon(u_{ij})=\delta_{ij},~1 \leq i,j \leq n.
	      \end{equation*}
       In the above formula for the coproduct $\Delta$, we have used the following standard notations : $\iota_j(a) = a_{|j} \in \mathcal{O}\langle n \rangle$, $a\in \mathcal{O}\langle n \rangle$.
	\item The rectangular unitary algebra $\runitaryalg$ is the involutive unital associative algebra generated by one unitary element $u$ and a complete set of mutually self-adjoint orthogonal projectors $\mathcal{R}=\{p_{i},~1 \leq i \leq n\}$,
    
    $$
    p_ip_j=\delta_{i=j}p_i,\quad \sum_{i=1}^n p_i = 1
    $$
    
    The rectangular unitary algebra $\mathcal{R}\mathcal{O}\langle n \rangle$ is a bimodule algebra over $\mathcal{R}$. The algebra $\runitaryalg$ is a $H$-algebra in the algebraic category $\Alg^{\star}(\mathcal{R})$ with structural morphisms:
	      \begin{equation*}
		      \mathcal{S}(u)=u^{\star}=u^{-1},~ \Delta(u)=u|_{1}u|_{2},~\varepsilon(u)=1\in R.
	      \end{equation*}
	\item Any commutative Hopf algebra is a $H$-algebra, thus if $G$ is a group then its space of polynomial functions $\mathcal{F}(G)$ is a $H$-algebra with structure morphisms given by:
	      \begin{equation*}
		      \Delta(f)(g,h) = f(gh),~\mathcal{S}(f)(g)=f(g^{-1}),~\varepsilon(f)=f(e).
	      \end{equation*}
\end{enumerate}
\end{example}
Consider a $H$-algebra $H$ on an algebraic category $\mathcal C$. For every object $A$ of $\mathcal C$, the set $\Hom_{\mathcal{C}}(H,A)$ is endowed with a \emph{group structure} by the formula for the product $\star$ between two elements $f,g \in \Hom_{\mathcal{C}}(H,A)$:
\begin{equation*}
f\ast g=(f\freem g)\circ \Delta.
\end{equation*}
The unit element of this group is $\eta_{A}\circ \epsilon_{H}$, where $\eta_{A}$ is the unique morphism from $k$ to $A$, and the inverse of an element $f$ of $\Hom(H,A)$ is $f\circ \mathcal{S}$.

In contrast with this situation, if $H$ is a Hopf algebra and $A$ is an algebra, the convolution product of two morphisms of algebras needs not to be a morphism of algebras unless $A$ is commutative.
\begin{example}
Consider for example the $H$-algebra $H$ of Example 1. above. Given two linear complex valued map $f,g: H \to \C$, one has for $x \in V$
$$f\ast g = (f \freem g) \circ \Delta (x) = f(x) + g(x)$$
If now $H$ is the $H$-algebra of Example 3. above, and $f=\delta_X,g=\delta_Y : \mathcal{F}(G)\to B$, where $B$ is any commutative algebra two algebra morphisms given by evaluation on $X,Y \in G$. One has, with $h \in \mathcal{F}(G)$
$$
f \ast g (h) = h(XY) = \delta_{XY}
$$
\end{example}
It is crucial that $\Hom_{\mathcal{C}}(H, A)$ is a group; it plays the role of a set of group-valued random variables and it is extremely natural for us to be able to take the inverse of a random variable and to multiply two of them.

Let us conclude this discussion with the following Theorem, which shows that $H$-algebras are in a sense exactly the right class of objects for our purposes.

\begin{theorem}[\cite{zhang}]
	Let $\mathcal{C}$ be an algebraic category. Let $H$ be an object of $\mathcal{C}$. Then ${\rm Hom}_{\mathcal{C}}(H,\cdot)$ is a functor from $\mathcal{C}$ to the category of groups if and only if there exists $\Delta, \varepsilon, \mathcal{S},$ such that $\left(H,\Delta,\varepsilon,\mathcal{S} \right)$ is a $H$-algebra of $\mathcal{C}$.
\end{theorem}

\subsection{Comodule-algebras over an $H$-algebra}\label{sec:ZhangcoMod}
In this section, we define the \emph{category of co-module algebras} over a $H$-algebra of an algebraic category $\mathcal C$ with initial object $k$ and coproduct $\free$. In this definition, and for every object $B$ of $\mathcal C$, we will identify without further mention the objects $B$, $B\free k$ and $k\free B$. They are indeed isomorphic by the maps $\id_{B}\freem \eta:B\free k\to B$ and $\eta \freem \id_{B}: k \free B\to B$, where $\eta:k\to B$ is the unique element of $\Hom(k,B)$.

For this section, let us fix an algebraic category $\mathcal C$ with initial object $k$ and coproduct $\free$, and a $H$-algebra $(H,\Delta,\varepsilon,\mathcal{S})$ of this category.

\begin{definition}[Right comodule-algebras] \label{def:comodZhang} A \emph{right $H$ comodule-algebra} of $\mathcal C$ is a pair $(M,\Omega)$, where $M$ is an object of $\mathcal C$ and $\Omega:M\to M\free H$ is a morphism such that
	\begin{enumerate}[\indent 1.]
		\item $M$ is an object of $\mathcal C$,
		\item $\Omega:M\to M\free H$ is a morphism of $\mathcal C$ satisfying the following two conditions:
		      \begin{equation}\tag{$1_{R}$}
			      \label{relationcomod}
			      \left(\Omega \free \id_{H}\right) \circ \Omega = \left(\id_{M} \free \Delta \right) \circ \Omega \ \text{ and } \  \left(\id_{M}\free \varepsilon \right) \circ \Omega = \textrm{\id}_{M}.
		      \end{equation}
	\end{enumerate}
\end{definition}
\begin{remarque}
The definition of a left comodule-algebra is deduced from the definition of a right comodule-algebra by replacing \eqref{relationcomod} with
		      \begin{equation}\tag{$1_{L}$}
			      \label{relationcomodL}
			      \left(\id_{H} \free \Omega\right) \circ \Omega = \left(\Delta \free \id_{M}  \right) \circ \Omega \ \text{ and } \  \left(\varepsilon \free \id_{M} \right) \circ \Omega = \textrm{\id}_{M}.
		      \end{equation}
\end{remarque}
\begin{remarque}
    Let us stop a moment for a terminological remark: why comodule-algebra over a $H$-algebra and not the shorter denomination \emph{comodule over an $H$-algebra}? Let us focus on the basic example of an algebra category, the category comAlg of commutative algebras (see before Example \ref{ex:algcat}). An $H$-algebra in comAlg if a commutative Hopf algebra $H$.  According to Definition \ref{def:comodZhang}, a comodule over the Zhang algebra $H$ is a comodule over $H$, in the usual sense\footnote{\url{https://ncatlab.org/nlab/show/comodule}} but the structural morphism $\Omega$ is required additionally to be an \emph{algebra morphism}, and therefore $M$ is a comodule-algebra in the usual sense\footnote{\url{https://ncatlab.org/nlab/show/comodule+algebra}}.
\end{remarque}
\setcounter{equation}{1}

A morphism between two right $H$-comodule-algebras $(M,\Omega_{M})$ and $(N,\Omega_{N})$ is, by definition, an element $f$ of $\Hom_{\mathcal C}(M,N)$ which respects the structure of $H$-comodule in the sense that
\begin{equation}\label{eq:morcomod}
	\Omega_{N}\circ f=(f\free \id_{H})\circ \Omega_{M}.
\end{equation}

We denote respectively by $\rCoMod\mathcal C(H)$ and $\lCoMod\mathcal C(H)$ the categories with objects the right and left $H$-comodule-algebras of $\mathcal C$ and for morphisms, the morphisms of $\mathcal{C}$ which are compatible with the comodule maps $\Omega$ on the source and the target. We will now state and prove these two categories are algebraic categories.
The three following Lemmas seem new.
We defer the proof of the following Lemma to the Annex, see Section \ref{annex:prooflemmacinq}.

\begin{lemma}\label{lemma_alge_alg_alg_star}
	The category $\rCoMod\mathcal C(H)$ is an algebraic category.
\end{lemma}

Of course, an analogous statement holds for the category $\lCoMod\mathcal C(H)$.

It is convenient to introduce a graphical calculus to perform computations involving coproducts. All morphisms we handle act on iterated folded coproducts of $H$ and are valued in the same type of objects. We do not forget that coproducts of objects in $\mathcal{C}$ come with injections. For $n \geq 1$ an integer, the injections from $H$ to $H^{\sqcup n}$ are denoted $\iota_{1},\ldots,\iota_{n}$. 
Let $f=(f_{1},\ldots,f_{n})$ a finite sequence of morphisms from $H$ to itself. To each permutation $\sigma$ of $\{1,\ldots,n\}$ is attached a morphism from $H^{\sqcup n}$ to $H^{\sqcup n}$, denoted $f_{\sigma}$ and defined as the unique morphism satisfying the property: $f_{\sigma} \circ \iota_{i} = i_{\sigma(i)} \circ f_{i}$. Such morphisms are depicted as follows: we draw $n$ vertical lines, labelled with the symbols $f_{1},\ldots,f_{n}$ from left to right. We add at the beginning of the $i^{th}$ vertical line the integer $i$ and the integer $\sigma(i)$ at the end of the line. We draw examples in Figure \ref{fig:diag} in case $n=3$.
\begin{figure}[!htb]\centering
	\begin{tikzpicture}[scale=1.2]
		\begin{scope}[]
			\draw(0,0)--node[left](){\tiny $f_{1}$} node[left,pos=0.9](){\tiny 1} node[left,pos=0.1](){\tiny $1$}(0,1);
			\draw(0.75,0)--node[left](){\tiny $f_{2}$} node[left,pos=0.9](){\tiny 2} node[left,pos=0.1](){\tiny 2}(0.75,1);
			\draw(1.5,0)--node[left](){\tiny $f_{3}$} node[left,pos=0.9](){\tiny 3} node[left,pos=0.1](){\tiny 3}(1.5,1);
		\end{scope}
		\node[]()at(1.75,0.5){,};
		\begin{scope}[shift={(3,0)}]
			\draw(0,0)--node[left](){\tiny $f_{1}$} node[left,pos=0.9](){\tiny 1} node[left,pos=0.1](){\tiny $1$}(0,1);
			\draw(0.75,0)--node[left](){\tiny $f_{2}$} node[left,pos=0.9](){\tiny 3} node[left,pos=0.1](){\tiny 2}(0.75,1);
			\draw(1.5,0)--node[left](){\tiny $f_{3}$} node[left,pos=0.9](){\tiny 2} node[left,pos=0.1](){\tiny 3}(1.5,1);
		\end{scope}
		\node[]()at(4.75,0.5){,};
		\begin{scope}[shift={(6,0)}]
			\draw(0,0)--node[left](){\tiny $f_{1}$} node[left,pos=0.9](){\tiny 2} node[left,pos=0.1](){\tiny $1$}(0,1);
			\draw(0.75,0)--node[left](){\tiny $f_{2}$} node[left,pos=0.9](){\tiny 1} node[left,pos=0.1](){\tiny 2}(0.75,1);
			\draw(1.5,0)--node[left](){\tiny $f_{3}$} node[left,pos=0.9](){\tiny 3} node[left,pos=0.1](){\tiny 3}(1.5,1);
		\end{scope}
	\end{tikzpicture}
	\caption{\small\label{fig:diag}Morphisms $f_{\id},~f_{(2,3)},~f_{(1,2)}$.}
\end{figure}

Of primary importance is the case where all the $f_{i}$ $^{'}s$ are equal to the identity of $H$. In that case, we use the notation $\tau_{\sigma}$ for the morphism $f_{\sigma}$. In words, $\tau_{\sigma}$ relabel the letters by substituting the label $\sigma_{i}$ to $i$. The Figure \ref{tauundeux} shows the graphical representation of the morphism $\tau_{12}$. To draw the graphical representations of the morphisms at stake, we make the convention that a sequence of edges starting and ending on same levels have ends labelled with increasing integers from left to right, the ends being labelled with the same integer.
\begin{figure}[!htb]\centering
	\begin{tikzpicture}[ever node/.style={fontscale=0.5},scale=1.2]
		\begin{scope}[shift={(-2,0)}]
			\draw(0,0)--(1,1);
			\draw(1,0)--(0,1);
		\end{scope}
		\begin{scope}[]
			\node () at (-0.5,0.5) {\tiny =};
			\draw(0,0)--node[right,pos=0.9](){\tiny 2}node[left,pos=0.1](){\tiny 1}(1,1);
			\draw(1,0)--node[left,pos=0.9](){\tiny 1}node[right,pos=0.1](){\tiny 2}(0,1);
		\end{scope}
		\begin{scope}[shift={(2,0)}]
			\node () at (-0.5,0.5) {\tiny =};
			\draw(0,0)--node[right,pos=0.9](){\tiny 2}node[right,pos=0.1](){\tiny 1}(0,1);
			\draw(1,0)--node[left,pos=0.9](){\tiny 1}node[left,pos=0.1](){\tiny 2}(1,1);
		\end{scope}
	\end{tikzpicture}
	\caption{\label{tauundeux} \small The permutation $\tau_{12}$ of labels.}
\end{figure}
Using graphical calculus, the structural morphisms $\mathcal{S}$, $\Delta$, $\varepsilon$ and $\mu = \id_{H} \freem \id_{H}$ are pictured as in Fig. \ref{diagrammstructuralmorphism} and the relations they are subject to are drawn in Fig. \ref{diagramerelationzang}.
\begin{figure}[!htb]\centering
	\begin{tikzpicture}[scale=1.2]
		\begin{scope}[]
			\node[]()at(-0.8,0.0){\tiny $\Delta=$};
			\begin{scope}[shift={(0,0)}, scale=0.5]
				\draw(-1,1)--(0,0)--(1,1);
				\draw(0,0)--(0,-1.40);
			\end{scope}
		\end{scope}

		\begin{scope}[shift={(0.8,0)}]
			\node () at (0,0.0) {\tiny,~$\mu=$};
			\begin{scope}[shift={(0.8,-0.7)},scale=0.5]
				\draw(-1,0)--(0,1)--(1,0);
				\draw(0,1)--(0,2.4);
			\end{scope}
		\end{scope}

		\begin{scope}[shift={(2.9,0)}]
			\node () at (-0.5,0.0) {\tiny ,~$\mathcal{S}=$};
			\draw(0,-0.7)--node[](){\tiny /}(0,0.5);
		\end{scope}

		\begin{scope}[shift={(4,0)}]
			\node () at (-.5,0) {\tiny ,~$\varepsilon=$};
			\begin{scope}[shift={(0,-.7)}]
				\draw(0,0)--(0,1.2);
				\node[draw,circle,inner sep=1pt,fill,draw] () at (0,1.2) {};
			\end{scope}
		\end{scope}

		\begin{scope}[shift={(5,0)}]
			\node () at (-.5,0) { \tiny ,~$\eta=$ };
			\begin{scope}[shift={(0,-.6)}]
				\draw(0,0)--(0,1.2);
				\node[draw,circle,inner sep=1pt,fill,draw] () at (0,0) {};
			\end{scope}
		\end{scope}
	\end{tikzpicture}
	\caption{\label{diagrammstructuralmorphism} \small Drawings of the structural morphisms of a $H$-algebra $(H,\Delta,\varepsilon,\mathcal{S})$ and $\mu = \id_{H}\freem \id_{H}$.}
\end{figure}

\begin{figure}[!htb]\centering
	\begin{tikzpicture}[scale=0.5]
		\draw(-1,1)--(0,0)--(1,1);
		\draw( 0,0)--(0,-1.4);
		\draw(-1,1.)--(-1,3.4);
		\begin{scope}[shift={(1,2.4)}]
			\draw(-1,1)--(0,0)--(1,1);
			\draw( 0,0)--(0,-1.4);
		\end{scope}
		\node () at (2.5,0.5) {=};

		\begin{scope}[shift={(5,0)}]
			\draw(-1,1)--(0,0)--(1,1);
			\draw( 0,0)--(0,-1.4);
			\draw(1,1.)--(1,3.4);
			\begin{scope}[shift={(-1,2.4)}]
				\draw(-1,1)--(0,0)--(1,1);
				\draw( 0,0)--(0,-1.4);
			\end{scope}
		\end{scope}
		\node () at (6.1,0.5) {,};
		\begin{scope}[shift={(8,0)}]
			\draw(-1,1)--(0,0)--(1,1);
			\draw( 0,0)--(0,-1.4);
			\draw(-1,1.)--node () { /}(-1,2.4)--(0,3.4)--(1,2.4);
			\draw(1,1)--(1,2.4);
		\end{scope}
		\node () at (10,.5) {=};
		\begin{scope}[shift={(12,0)}]
			\draw(-1,1)--(0,0)--(1,1);
			\draw( 0,0)--(0,-1.4);
			\draw(-1,1.)--(-1,2.4)--(0,3.4)--(1,2.4);
			\draw(1,1)--node(){/}(1,2.4);
		\end{scope}
		\node () at (14,0.5) {=};
		\begin{scope}[shift={(15,0)}]
			\draw(0,-1.4)--node[draw,fill,circle,pos=0.95,inner sep=1.pt](){}(0,1);
			\draw(0,3.4)--node[draw,fill,circle,pos=0.95,inner sep=1.pt](){}(0,1);
		\end{scope}
		\node () at (15.7,0.5) {,};
		\begin{scope}[shift={(18,0)}]
			\draw(1,1)--(0,0)--(-1,1);
			\draw(0,0)--(0,-1.4);
			\draw(-1,1)--node[draw,fill,circle,pos=0.98,inner sep=1.pt](){}(-1,3.4);
			\draw(1,1)--(1,3.4);
		\end{scope}
		\node () at (20,0.5) {=};
		\begin{scope}[shift={(22,0)}]
			\draw(1,1)--(0,0)--(-1,1);
			\draw(0,0)--(0,-1.4);
			\draw(-1,1)--(-1,3.4);
			\draw(1,1)--node[draw,fill,circle,pos=0.98,inner sep=1.pt](){}(1,3.4);
		\end{scope}
		\node () at (24,0.5) {=};
		\begin{scope}[shift={(25,0)}]
			\draw(0,-1.4)--(0,3.4);
		\end{scope}
	\end{tikzpicture}
	\caption{\label{diagramerelationzang} \small Relations amongst the structural morphisms of a $H$-algebra, from left to right: $\Delta \free \id_{H} \circ \Delta = \id_{H} \free \Delta \circ \Delta$, $\mathcal{S}\free \id_{H} \circ \Delta = \id_{H} \free \mathcal{S} \circ \Delta = \varepsilon \circ \eta, \varepsilon \free \id_{H} = \id_{H} \free \mathcal{S} = \id_{H}$.}
\end{figure}
Coassociativity of $\Delta$ pictured in the left corner of Fig \ref{diagramerelationzang} permits an unambiguous interpretation of a \emph{corolla} (see the bottom left corner of Fig. \ref{proofrelationantipodecoproduct}) as a literal expression in $\Delta$; it corresponds to an iterative composition of $\Delta$, either on its left or right leg (it does not matter, the two alternatives yield the same morphism). For example, $\mathrm{id}_H \sqcup \Delta \sqcup \mathrm{id}_H \circ {\rm id} \sqcup \Delta \circ \Delta$ is pictured as a corolla with three strands.
On a $H$-algebra, the antipode $\mathcal{S}$ is not a comorphism, however, a simple relation between $\mathcal{S} \sqcup \mathcal{S} \circ \Delta$ and $\Delta \circ \mathcal{S}$ can be deduced from the three structural relations drawn in Fig. \ref{diagramerelationzang} which are reminiscent from the fact that the antipode $\mathcal{S}$ of a Hopf algebra is an anti-co-morphism:
\begin{equation*}
	\tau_{(12)} \circ (\mathcal{S} \sqcup \mathcal{S}) \circ \Delta = \Delta \circ \mathcal{S}.
\end{equation*}
This last relation is pictured in Fig \ref{relationantipodecoproduct} and its proof can be found in the seminal article of Zhang \cite{zhang}. Let us, however, illustrate how the graphical calculus we introduce works in this first example. The morphism $\Delta \circ \mathcal{S}$ is the inverse of $\Delta$ in the group ${\rm Hom}_{\mathcal{C}}(H,H\sqcup H)$. We have to show that $((\tau_{12}\circ (\mathcal{S} \sqcup \mathcal{S})\circ \Delta) \freem \Delta) \circ \Delta =  \eta_{H\sqcup H}\circ \varepsilon$, where $\eta_{H\sqcup H}$ is the unique morphism from the initial object to $H\sqcup H$, graphical computations are performed in Fig. \ref{proofrelationantipodecoproduct}. By using the same method, it can be proved that $\mathcal{S}^{2} = \id_{H}$.
\begin{figure}[!htb]\centering
	\begin{tikzpicture}[scale=0.8]
		\begin{scope}[shift={(0,0.45)}]
			\draw(-1,1)--node[pos=0.1,right](){\tiny 1}(0,0)--node[pos=0.9,left](){\tiny 2}(1,1);
			\draw(0,0)--node[pos=0.9,right](){\tiny 1}node[](){/}(0,-1.4);
		\end{scope}
		\node[]()at(1.5,0.5){=};
		\begin{scope}[shift={(3.1,0.4)}]
			\draw(-1,1)--(0,0)--(1,1);
			\draw(0,0)--node[pos=0.9,right](){\tiny 1}(0,-1.4);
			\draw(-1,1)--node[pos=0.9,right](){\tiny 2}node[](){/}(-1,2);
			\draw(1,1)--node[pos=0.9,right](){\tiny 1}node[](){/}(1,2);
		\end{scope}
	\end{tikzpicture}
	\caption{ \small \label{relationantipodecoproduct} Diagram corresponding to the relation $\tau_{(12)} \circ (\mathcal{S} \sqcup \mathcal{S}) \circ \Delta = \Delta \circ \mathcal{S}$.}
\end{figure}
\begin{figure}[!htb]\centering
	\begin{tikzpicture}[scale=0.75]
		\draw(-1,1)--(0,0)--(1,1);
		\draw(-1,1)--node[](){/}(-1,2.4);
		\draw(-1,2.4)--(0.3,4.4);
		\draw(0.3,4.4)--(1,3.4);
		\draw(1,1)--(1,3.4);
		\draw(0.3,1)--(0.3,3.4);
		\draw(0.3,3.4)--(-.3,4.4)--(-1,3.4);
		\draw(-0.33,1)--node[](){/}(-0.33,2.4);
		\draw(-0.33,1)--(0,0);
		\draw(-.33,2.4)--(-1,3.4);
		\draw(0.33,1)--(0,0);
		\draw( 0,0)--(0,-1.4);
		\node[]()at(1.5,1){=};
		\begin{scope}[shift={(3,0)}]
			\draw(-1,1)--node[pos=0.1,right](){\tiny 1}(0,0)--node[pos=0.9,right](){\tiny 4}(1,1);
			\draw(-1,1)--node[](){/}node[pos=0.9,right](){\tiny 2}(-1,2.4);
			\draw(1,1)--node[pos=0.9,right](){\tiny 4}(1,2.4);
			\draw(0.3,1)--node[pos=0.9,right](){\tiny 3}(0.3,2.4);
			\draw(-0.33,1)--node[](){/}node[pos=0.9,right](){\tiny 1}(-0.33,2.4);
			\draw(-0.33,1)--node[pos=0.1,right](){\tiny 2}(0,0);
			\draw(0.33,1)--node[pos=0.1,right](){\tiny 3}(0,0);
			\draw( 0,0)--(0,-1.4);
			\draw(-0.33,2.4)--node[pos=1.16,](){\tiny 1}(0,3.4)--(0.3,2.4);
			\draw(-1,2.4)--node[pos=1.1,](){\tiny 2}(0,4.4)--(1,2.4);
		\end{scope}
		\node[]()at(4.8,1){=};
		\begin{scope}[shift={(5.7,0)}]
			\draw(0,0)--node[pos=0.95,circle,fill,inner sep=1pt](){}(0,1);
			\draw(-.33,1)--node[pos=0.1,circle,fill,inner sep=1pt](){}(-.33,2);
			\draw(.33,1)--node[pos=0.1,circle,fill,inner sep=1pt](){}(.33,2);
		\end{scope}
	\end{tikzpicture}
	\caption{\label{proofrelationantipodecoproduct}\small The graphical proof of the relation $\tau_{12}\circ(\mathcal{S}\sqcup \mathcal{S})\circ \Delta = \Delta \circ \mathcal{S}$.}
\end{figure}
The main goal of the following two lemmas (Lemma \ref{lem:zhangzhang} and Lemma \ref{lemma:zhangconjugation}) is to prove that a $H$-algebra in algebraic category $\mathcal{C}$ is also a $H$-algebra in the category ${\rm rCoMod}\mathcal{C}(H)$ of right comodules over $H$ in $\mathcal{C}$.
Let us denote by $\iota_{1}:{\color{cyan}H}\to {\color{cyan}H}\free H$ and $\iota_{2}:{\color{magenta}H}\to H\free {\color{magenta}H}$ the canonical maps. We define the morphism $\Omega_{c}:H\to H \free H$ by the following formula:

\begin{equation}
\label{eqn:gaugecoaction}
\Omega_{c}=(\iota_{1} \freem \iota_{2} \freem \iota_{1}) \circ (\id_{H\free H} \free \mathcal{S}) \circ (\Delta \free \id_{H}) \circ \Delta.
\end{equation}
At times, we will refer to $\Omega_c$ as the \emph{conjugacy coaction} of $H$.
For an integer $n\geq 1$, we denote by $\Omega_{c}^{n}$ the induced co-action on $H^{\sqcup n}$.
\begin{lemma}
	\label{lemma:zhangconjugation}
	The pair $(H,\Omega_{c})$ is a right comodule-algebra of $\mathcal C$ over $H$.
\end{lemma}
\begin{proof} The equality $(\Omega_{c}\free \id_H)\circ \Omega_{c}=(\id_{H}\free \Delta)\circ \Omega_{c}$ follows from the fact that both sides are equal to
	\[(\iota_{1}\freem\iota_{2}\freem\iota_{3}\freem\iota_{2}\freem\iota_{1}) \circ(\id_{H^{\free 4}} \free \mathcal{S})\circ \Delta^{4},\]
	as one checks using the coassociativity of $\Delta$ and the fact that it is a morphism. Let us write more details to convince the reader with the efficiency of the graphical calculus we introduced. The co-action $\Omega_{c}$ has the graphical presentation showed in Fig. \ref{co-action}. We begin with the first relation of \eqref{relationcomodL}.
	\begin{figure}\centering
		\begin{tikzpicture}[scale=0.6]
			\begin{scope}
				\draw(0,0)--node[](){}(-1,1);
				\draw(0,0)--(0,1);
				\draw(0,0)--(1,1);
				\draw(0,0)--(0,-1.4);
			\end{scope}
			\begin{scope}[shift={(0,1)}]
				\draw(0,0)--(0,1);
				\draw(-1,0)--(-1,1);
				\draw(1,0)--node[ ]{\tiny \bf /}(1,1);
			\end{scope}
			\begin{scope}[shift={(0,2)}]
				\draw(-0,1)--(1,0);
				\draw(0,0)--(1,1)--(1,2.7);
				\draw(-1,0)--(-1,1)--(-0.5,2)--(0,1);
				\draw(-.5,2)--(-.5,2.7);
			\end{scope}
		\end{tikzpicture}
		\caption{\label{co-action} Diagram representing the co-action $\Omega_{c}$.}
	\end{figure}
	In figure \ref{co-actionrelation}, we drew the sequence of diagram that proves the equality $(\id_{H} \free \Delta)\circ\Omega_{c} =  (\iota_{1}\freem\iota_{2}\freem\iota_{3}\freem\iota_{2}\freem\iota_{1}) \circ(\id_{H^{\free 4}} \free \mathcal{S})\circ \Delta^{4}$.
	\begin{figure} \centering
		\begin{tikzpicture}[scale=0.5]
			\begin{scope}
				\draw(0,0)--(-1,1);
				\draw(0,0)--(0,1);
				\draw(0,0)--(1,1);
				\draw(0,0)--(0,-1.4);
			\end{scope}
			\begin{scope}[shift={(0,1)}]
				\draw(0,0)--(0,1);
				\draw(-1,0)--(-1,1);
				\draw(1,0)--node[ ]{ \bf /}(1,1);
			\end{scope}
			\begin{scope}[shift={(0,2)}]
				\draw(0,1)--(1,0);
				\draw(0,0)--(1,1);
				\draw(-1,0)--(-1,1);
			\end{scope}
			\begin{scope}[shift={(0,3)}]
				\draw(-1,0)--( -0.5,1);
				\draw(0,0)--(-0.5,1);
				\draw(1,0)--( 1,1);
			\end{scope}
			\begin{scope}[shift={(-0.5,5.4)}]
				\draw(0,-1.4)--(0,0);
				\draw(-0.5,1)--(0,0)--(0.5,1);
				\draw(1.5,-1.4)--(1.5,1);
			\end{scope}
			\node() at (2,1) {=};
			\begin{scope}[shift={(4,0)}]
				\begin{scope}
					\draw(0,0)--(-1,1);
					\draw(0,0)--(0,1);
					\draw(0,0)--(1,1);
					\draw(0,0)--(0,-1.4);
				\end{scope}
				\begin{scope}[shift={(0,1)}]
					\draw(0,0)--(0,1);
					\draw(-1,0)--(-1,1);
					\draw(1,0)--node[ ]{ \bf /}(1,1);
				\end{scope}
				\begin{scope}[shift={(0,2)}]
					\draw(0,1)--(1,0);
					\draw(0,0)--(1,1);
					\draw(-1,0)--(-1,1);
					\draw(1,1)--node[pos=1.03](){}(1,4.2);
				\end{scope}
				\begin{scope}[shift={(-0.6,3)},xscale=0.4]
					\draw(-1,0)--(-1,0.7);
					\draw(-2,1.7)--(-1,0.7)--(0,1.7);
				\end{scope}
				\begin{scope}[shift={(0.4,3)},xscale=0.4]
					\draw(-1,0)--(-1,0.7);
					\draw(-2,1.7)--(-1,0.7)--(0,1.7);
				\end{scope}
				\begin{scope}[shift={(0,3)},xscale=0.4]
					\draw(-1.5,1.7)--(-1,2.2);
					\draw(-1,1.7)--(-1.5,2.2);
					\draw(-3.5,1.7)--(-3.5,2.2);
					\draw(1,1.7)--(1,2.2);
				\end{scope}
				\begin{scope}[shift={(0,3)},xscale=0.4]
					\draw(-3.5,2.2 )--node[pos=1.15](){}(-2.5,3.2)--(-1.5,2.2);
					\draw(-1,2.2)--node[pos=1.15](){}(0.1,3.2)--(1,2.2);
				\end{scope}
			\end{scope}
			\node () at (6,1) {=};
			\begin{scope}[shift={(9,0)}]
				\node [below](1) at (0,0) {\tiny 1};
				\draw(-2,1)--(0,0);
				\draw(-1,1)--(0,0);
				\draw(0,1)--(0,0);
				\draw(1,1)--(0,0);
				\draw(2,1)--(0,0);
				\draw(0,0)--(0,-1.4);
				\begin{scope}[shift={(0,1)}]
					\node[above](23)at (0,1){\tiny 3};
					\draw(2,0)--node[ ]{ \tiny /}(2,1);
					\draw(1,0)--node[ ]{ \tiny /}(1,1);
					\draw(0,0)--node[ ]{ }(0,1);
					\draw(-1,0)--node[ ]{ }(-1,1);
					\draw(-2,0)--node[ ]{ }(-2,1);
				\end{scope}
				\begin{scope}[shift={(0,2)}]
					\node[above]() at(0,1) {\tiny 2};
					\node[above]() at(0,2) {\tiny 1};
					\draw(-1,0)--(0,1)--(1,0);
					\draw(-2,0)--(0,2)--(2,0);
				\end{scope}
			\end{scope}
			\node[]()at(12,1){=};

			\begin{scope}[shift={(13.5,0)},scale=0.68]
				\begin{scope}
					\draw(0,0)--(-1,1);
					\draw(0,0)--(0,1);
					\draw(0,0)--(1,1);
					\draw(0,0)--(0,-1.4);
				\end{scope}
				\begin{scope}[shift={(0,1)}]
					\draw(0,0)--(0,1);
					\draw(-1,0)--(-1,1);
					\draw(1,0)--node[ ]{\tiny \bf /}(1,1);
				\end{scope}
				\begin{scope}[shift={(0,2)}]
					\draw(-0,1)--(1,0);
					\draw(0,0)--(1,1)--(1,2.7);
					\draw(-1,0)--(-1,1)--(-0.5,2)--(0,1);
					\draw(-.5,2)--(-.5,7.2);
				\end{scope}
				\begin{scope}[shift={(1,4.5)}]
					\begin{scope}
						\draw(0,0)--(-1,1);
						\draw(0,0)--(0,1);
						\draw(0,0)--(1,1);
						\draw(0,0)--(0,-1.4);
					\end{scope}
					\begin{scope}[shift={(0,1)}]
						\draw(0,0)--(0,1);
						\draw(-1,0)--(-1,1);
						\draw(1,0)--node[ ]{\tiny \bf /}(1,1);
					\end{scope}
					\begin{scope}[shift={(0,2)}]
						\draw(-0,1)--(1,0);
						\draw(0,0)--(1,1)--(1,2.7);
						\draw(-1,0)--(-1,1)--(-0.5,2)--(0,1);
						\draw(-.5,2)--(-.5,2.7);
					\end{scope}
				\end{scope}

			\end{scope}

		\end{tikzpicture}
		\caption{\label{co-actionrelation}Graphical proof of the relation $(\Delta \freem 1)\circ \Omega_{c}= (1\freem \Omega_{c})\circ \Omega_{c}$.}
	\end{figure}

	The verification of the equality $(\id_{H}\free \varepsilon)\circ \Omega_{c}=\id_{H}$ is simpler.

\end{proof}

Furthermore, we claim that $\Delta$, $\varepsilon$, and $\mathcal{S}$, which are defined as morphisms in the category $\mathcal C$, are in fact morphisms in the category $\rCoMod\mathcal C(H)$. This means that they satisfy the equation \eqref{eq:morcomod}.

\begin{lemma} \label{lem:zhangzhang}
	$((H,\Omega_{c}),\Delta,\varepsilon,\mathcal{S})$ is a $H$-algebra of the category $\rCoMod\mathcal C(H)$.
\end{lemma}
\begin{proof} We have to check that the three morphisms $\Delta,~\varepsilon,~\mathcal{S}$ are co-module morphisms. To that end, we use the graphical calculus introduced previously. We begin with the antipode, compatibility between $\mathcal{S}$ and the right co-action $\Omega_{c}$ means $\Omega_{c} \circ \mathcal{S} = \left(\id_{H} \sqcup \mathcal{S} \right)\circ \Omega_{c}$. The computations are pictured in Fig. \ref{antipodeconjugation}, to perform them we use the relation drawn in Fig. \ref{relationantipodecoproduct}.
	\begin{figure}[!htb]\centering
		\begin{tikzpicture}[scale=0.6]
			\begin{scope}
				\draw(0,0)--node[](){}(-1,1);
				\draw(0,0)--(0,1);
				\draw(0,0)--(1,1);
				\draw(0,0)--(0,-1.4);
			\end{scope}
			\begin{scope}[shift={(0,1)}]
				\draw(0,0)--(0,1);
				\draw(-1,0)--(-1,1);
				\draw(1,0)--node[ ]{/}(1,1);
			\end{scope}
			\begin{scope}[shift={(0,2)}]
				\draw(-0,1)--(1,0);
				\draw(0,0)--(1,1)--node[](){/}(1,2.7);
				\draw(-1,0)--(-1,1)--(-0.5,2)--(0,1);
				\draw(-.5,2)--(-.5,2.7);
			\end{scope}

			\draw(0,0)--node[](){/}(0,1);
			\node[]()at(1.6,1 ){=};
			\begin{scope}[shift={(3,-0.4)},scale=0.7]
				\begin{scope}
					\draw(0,0)--(-1,1);
					\draw(0,0)--(1,1);
					\draw(1,1)--(-1,2);
					\draw(-1,1)--(1,2);
					\draw(0,0)--(0,-1.4);
				\end{scope}
				\begin{scope}[shift={(1,3.4)}]
					\draw(0,0)--(-1,1);
					\draw(0,0)--(1,1);
					\draw(1,1)--(-1,2);
					\draw(-1,1)--(1,2);
					\draw(0,0)--(0,-1.4);
				\end{scope}
				\begin{scope}[,shift={(-1,3.4)}]
					\draw(0,-1.4)--(0,3);
					\draw(1,2)--(3,3);
					\draw(3,2)--(1,3);
				\end{scope}
				\begin{scope}[shift={(-1,6.4)}]
					\draw(0,0)--(0.5,0.5);
					\draw(1,0)--(0.5,0.5);
					\draw(0.5,0.5)--(0.5,1);
					\draw(3,0)--node[](){/}(3,1);
				\end{scope}
			\end{scope}
			\node[]()at(4.6,1 ){=};

			\begin{scope}[shift={(6.5,0)}]
				\begin{scope}
					\draw(0,0)--(-1,1);
					\draw(0,0)--(0,1);
					\draw(0,0)--(1,1);
					\draw(0,0)--(0,-1.4);
				\end{scope}
				\begin{scope}[shift={(0,1)}]
					\draw(0,0)--(0,1);
					\draw(-1,0)--(-1,1);
					\draw(1,0)--node[ ]{/}(1,1);
				\end{scope}
				\begin{scope}[shift={(0,2)}]
					\draw(0,1)--(1,0);
					\draw(0,0)--(1,1);
					\draw(-1,0)--(-1,1);
				\end{scope}
				\begin{scope}[shift={(0,3)}]
					\draw(-1,1)--(0,0);
					\draw(-1,0)--(0,1);
					\draw(1,0)--node[](){/}(1,1);
				\end{scope}
				\begin{scope}[shift={(0,4)}]
					\draw(-1,0)--(-0.5,0.5)--(0,0);
					\draw(-0.5,.5)--(-0.5,0.75);
					\draw(1,0)--(1,0.75);
				\end{scope}
			\end{scope}
			\node[]()at(8.25,1 ){=};

			\begin{scope}[shift={(10,0)}]
				\begin{scope}
					\draw(0,0)--(-1,1);
					\draw(0,0)--(0,1);
					\draw(0,0)--(1,1);
					\draw(0,0)--(0,-1.4);
				\end{scope}
				\begin{scope}[shift={(0,1)}]
					\draw(0,0)--(0,1);
					\draw(-1,0)--(-1,1);
					\draw(1,0)--node[ ]{/}(1,1);
				\end{scope}
				\begin{scope}[shift={(0,2)}]
					\draw(-1,1)--(1,0);
					\draw(0,0)--(1,1);
					\draw(-1,0)--(-1,1);
				\end{scope}
				\begin{scope}[shift={(0,3)}]
					\draw(-1,0)--(-1,1.75);
					\draw(1,0)--node[](){/}(1,1.75);
				\end{scope}
			\end{scope}
		\end{tikzpicture}
		\caption{\label{antipodeconjugation}\small The relation $\mathcal{S}\circ\Omega_{c} = (\id_{H} \sqcup \mathcal{S}) \circ \Omega_{c}$.}
	\end{figure}

	We now turn our attention to the relation that needs to be satisfied by the counit $\varepsilon$. The computations are drawn in Fig. \ref{compatcounitomega}. We have to check that $\eta_{H\free k} \circ \varepsilon = \left(\id_{H} \free \varepsilon\right) \circ \Omega_{c}$, because we identify $H\sqcup k$ with $H$, the last relation is written as $\eta_{H} \circ \varepsilon = \left(\id_{H} \free \varepsilon\right) \circ \Omega_{c}$.
	\begin{figure}[!htb]\centering
		\begin{tikzpicture}[scale=0.6]
			\begin{scope}[shift={(0,1.4)}]
				\begin{scope}
					\draw(0,0)--node[](){}(-1,1);
					\draw(0,0)--(0,1);
					\draw(0,0)--(1,1);
					\draw(0,0)--(0,-1.4);
				\end{scope}
				\begin{scope}[shift={(0,1)}]
					\draw(0,0)--(0,1);
					\draw(-1,0)--(-1,1);
					\draw(1,0)--node[ ]{\tiny \bf /}(1,1);
				\end{scope}
				\begin{scope}[shift={(0,2)}]
					\draw(-0,1)--(1,0);
					\draw(0,0)--(1,1)--node[pos=0.95,circle,fill,inner sep=1pt](){}(1,2.7);
					\draw(-1,0)--(-1,1)--(-0.5,2)--(0,1);
					\draw(-.5,2)--(-.5,2.7);
				\end{scope}
			\end{scope}

			\node () at (2,1) {=};

			\begin{scope}[shift={(4,0)}]
				\begin{scope}[]
					\draw(0,0)--(0,1);
					\draw(0,1) -- node[](){\tiny /}(-1,2);
					\draw(0,1)--(1,2);
				\end{scope}
				\begin{scope}[]
					\draw(-1,2)--(-1,3)--(0,4)--(1,3);
					\draw(1,2)--(1,3);
					\draw(0,4)--(0,5.4);
				\end{scope}
			\end{scope}

			\node () at (6,1) {=};

			\begin{scope}[shift={(7.5,0)}]
				\draw(0,0)--node[pos=0.96,draw,fill,circle,inner sep=1pt](){}(0,2.65);\draw(0,2.75)--node[pos=0.04,draw,fill,circle,inner sep=1pt](){}(0,5.4);
			\end{scope}
		\end{tikzpicture}
		\caption{\label{compatcounitomega}Diagrammatic proof of the relation $\eta_{H} \circ \varepsilon = \left(\id_{H} \free \varepsilon\right) \circ \Omega_{c}$.}
	\end{figure}
	The final relation that needs to be checked implies that $\Delta$ is a comodule morphism with respect to the co-actions $\Omega_{c}$ on $H$ and $\Omega_{c}^{2}$ on $H\sqcup H$:$\Omega_{c}^{2}\circ \Delta = (\id_{H} \sqcup \Delta) \circ \Omega_{c}$. Once again, we perform diagrammatic computations that are pictured in Fig. \ref{deltaconjugation}.

	\begin{figure}[!htb]\centering
		\begin{tikzpicture}[scale=0.5]
			\begin{scope}
				\begin{scope}[shift={(0,0)}]
					\draw(0,0)--(1,1);
					\draw(0,0)--(-1,1);
					\draw(0,0)--(0,-1.4);
				\end{scope}
				\begin{scope}[shift={(-1,2.1)},scale=0.8]
					\begin{scope}
						\draw(0,0)--node[](){}(-1,1);
						\draw(0,0)--(0,1);
						\draw(0,0)--(1,1);
						\draw(0,0)--(0,-1.4);
					\end{scope}
					\begin{scope}[shift={(0,1)}]
						\draw(0,0)--(0,1);
						\draw(-1,0)--(-1,1);
						\draw(1,0)--node[ ]{  \bf /}(1,1);
					\end{scope}
					\begin{scope}[shift={(0,2)}]
						\draw(-0,1)--(1,0);
						\draw(0,0)--(1,1)--(1,2.7);
						\draw(-1,0)--(-1,1)--(-0.5,2)--(0,1);
						\draw(-.5,2)--(-.5,2.7);
					\end{scope}
				\end{scope}
				\node ( ) at (2.5,1) {=};
				\begin{scope}[shift={(1,2.1)},scale=0.8]
					\begin{scope}
						\draw(0,0)--node[](){}(-1,1);
						\draw(0,0)--(0,1);
						\draw(0,0)--(1,1);
						\draw(0,0)--(0,-1.4);
					\end{scope}
					\begin{scope}[shift={(0,1)}]
						\draw(0,0)--(0,1);
						\draw(-1,0)--(-1,1);
						\draw(1,0)--node[ ]{  \bf /}(1,1);
					\end{scope}
					\begin{scope}[shift={(0,2)}]
						\draw(-0,1)--(1,0);
						\draw(0,0)--(1,1)--(1,2.7);
						\draw(-1,0)--(-1,1)--(-0.5,2)--(0,1);
						\draw(-.5,2)--(-.5,2.7);
					\end{scope}
				\end{scope}
				\begin{scope}[shift={(-0.4,5.85)}]
					\draw(-1,0)--(0,1)--(1,0);
					\draw(0.2,0)--(1,1);
				\end{scope}
				\begin{scope}[shift={(-0.4,5.85)}]
					\draw(0,1)--(0,2);
					\draw(1,1)--(1,2);
					\draw(2.2,0)--(2.2,2);
				\end{scope}
			\end{scope}
			\begin{scope}[shift={(6.5,0)}]
				\draw(0,0)--(-2.5,1);
				\draw(0,0)--(-1.5,1);
				\draw(0,0)--(-0.5,1);
				\draw(0,0)--(2.5,1);
				\draw(0,0)--(1.5,1);
				\draw(0,0)--(0.5,1);
				\draw(0,0)--(0,-1.4);
				\begin{scope}[shift={(0,1)}]
					\draw(2.5,0)--(2.5,1);
					\draw(1.5,0)--(1.5,1);
					\draw(0.5,0)--(0.5,1);
					\draw(-0.5,0)--(-0.5,1);
					\draw(-1.5,0)--(-1.5,1);
					\draw(-2.5,0)--(-2.5,1);
				\end{scope}
				\begin{scope}[shift={(0,2)}]
					\draw(2.5,0)--node[](){  /}(2.5,1);
					\draw(0.5,0)--(0.5,1);
					\draw(-0.5,0)--node[](){  /}(-0.5,1);
					\draw(-2.5,0)--(-2.5,1);
				\end{scope}
				\begin{scope}[shift={(0,3)}]
					\draw(-2.5,0)--(0,1);
					\draw(-0.5,0)--(0,1);
					\draw(0.5,0)--(0,1);
					\draw(2.5,0)--(0,1);
					\draw(0,1)--(0,2);
				\end{scope}
			\end{scope}
			\node ( ) at (10,1) {=};
			\begin{scope}[shift={(13.5,0)}]
				\draw(0,0)--(-2.5,1);
				\draw(0,0)--(-1.5,1);
				\draw(0,0)--(2.5,1);
				\draw(0,0)--(1.5,1);
				\draw(0,0)--(0.5,1);
				\draw(0,0)--(0,-1.4);
				\begin{scope}[shift={(0,1)}]
					\draw(2.5,0)--(2.5,1);
					\draw(1.5,0)--(1.5,1);
					\draw(0.5,0)--(0.5,1);
					\draw(-1.5,0)--(-1.5,1);
					\draw(-2.5,0)--(-2.5,1);
				\end{scope}
				\begin{scope}[shift={(0,2)}]
					\draw(2.5,0)--node[](){  /}(2.5,1);
					\draw(-2.5,0)--(-2.5,1);
				\end{scope}
				\begin{scope}[shift={(0,3)}]
					\draw(-2.5,0)--(0,1);
					\draw(2.5,0)--(0,1);
					\draw(0,1)--(0,2);
				\end{scope}
			\end{scope}
			\node ( ) at (17,1) {=};
			\begin{scope}[shift={(19,0)} ]
				\begin{scope}
					\draw(0,0)--(-1,1);
					\draw(0,0)--(0,1);
					\draw(0,0)--(1,1);
					\draw(0,0)--(0,-1.4);
				\end{scope}
				\begin{scope}[shift={(0,1)}]
					\draw(0,0)--(0,1);
					\draw(-1,0)--(-1,1);
					\draw(1,0)--node[ ]{  \bf /}(1,1);
				\end{scope}
				\begin{scope}[shift={(0,2)}]
					\draw(0,1)--(1,0);
					\draw(0,0)--(1,1);
					\draw(-1,0)--(-1,1)--(-0.5,2)--(0,1);
					\draw(1,1)--(1,2);
				\end{scope}
				\begin{scope}[shift={(1,3)}]
					\draw(0,1)--(-1,2);
					\draw(0,1)--(1,2);
					\draw(-1.5,1)--(-1.5,2);
				\end{scope}
			\end{scope}
		\end{tikzpicture}
		\caption{\label{deltaconjugation}\small Proof of the relation $\Omega_{c}^{2}\circ \Delta = (\id_{H} \sqcup \Delta) \circ \Omega_{c}$. }
	\end{figure}
\end{proof}
In the following Section, we recall basic facts on monoidal structures and inductive (direct) limits.
\subsection{Categorical independence}
\label{sec:monoidal}
In order to motivate our definitions, we start by reviewing the classical, commutative case. 
Let
\begin{equation}
	X : \left(\Omega,\mathcal{F},\mathbb{P}\right) \rightarrow \left(S,\mathcal{S}\right) \text{ and } Y : \left(\Omega,\mathcal{F},\mathbb{P} \right) \rightarrow \left(T, \mathcal{T} \right)
\end{equation}
be two essentially bounded random variables defined on the same classical probability space. Let $\tau_{X}$ and $\tau_{Y}$ denote the linear forms defined respectively on $L^{\infty}(S,\mathcal{S})$ and $L^{\infty}(T,\mathcal{T})$ by the distributions of $X$ and $Y$.

The random variables $X$ and $Y$ also induce  homomorphisms of the algebras of essentially bounded random variables
\[j_{X} : L^{\infty}(S,\mathcal{S}) \to L^{0}(\Omega,\mathcal{F}) \text{ and } j_{Y} : L^{0}(T,\mathcal{T}) \to L^{\infty}(\Omega,\mathcal{F}),\]
defined by $j_{X}(f)=f\circ X$ and $j_{Y}(g)=g\circ Y$.
We denote by ComProb the category of commutative probability spaces, with objects the algebras $L^{\infty}(\Omega,\mathcal{F},\mathbb{P})$ and morphisms the algebras morphisms $j_X$.

The independence of the random variables $X$ and $Y$ in the category ComProb is equivalent to the existence of a morphism $j$ such that the diagram in Fig. \ref{fig:comindep} commutes.

\begin{figure}[!htb]\centering
	\begin{tikzcd}[ampersand replacement=\&]
		\& \left(L^{\infty}(\Omega,\mathcal{F}), \mathbb{E}\right) \& \\
		\left(L^{\infty}(S,\mathcal{S}), \tau_{X}\right) \arrow{ru}{j_{X}} \arrow{r}{\iota} \& \arrow{u}{j} \left(L^{\infty}(S,\mathcal{S}),\tau_{X}\right)\otimes\left(L^{\infty}(T,\mathcal{T}),\tau_{Y}\right) \& \arrow[swap]{l}{\iota} \arrow{lu}[swap]{j_{Y}} \left(L^{\infty}(T,\mathcal{T}),\tau_{Y}\right)
	\end{tikzcd}
  \caption{\label{fig:comindep} Classical independence of two random variables $X$ and $Y$.}
\end{figure}

In order to generalize the notion of independence from the category of commutative probability spaces to an arbitrary category of non-commutative probability spaces, it appears that a notion of monoidal product is needed. A category in which the monoidal product of two objects is defined is called a monoidal category or a tensor category.
A monoidal category $\mathcal{C}$ is a category $\mathcal{C}$ together with a bifunctor $\otimes:\mathcal{C}\times\mathcal{C}\to\mathcal{C}$ which
\begin{enumerate}
	\item is associative under a natural isomorphism with components
	      \begin{equation*}
		      \alpha_{A,B,C}: A\otimes (B\otimes C) \overset{\cong}{\to}(A \otimes B) \otimes C
	      \end{equation*}
	      called associativity constraints,
	\item has a unit object $E \in \Obj(\mathcal{C})$ acting as left and right identity under natural isomorphisms with components:
	      \begin{equation*}
		      \ell_{A}:E\otimes A \overset{\cong}{\to} A,~ r_{A}:A\otimes E \overset{\cong}{\to}A
	      \end{equation*}
	      called left unit constraint and right unit constraint such that the pentagon and triangle identities hold, see Fig. \ref{pentagon} and Fig. \ref{triangle}
\end{enumerate}
\begin{figure}[!htb]\centering
	\adjustbox{scale=0.9}{
		\begin{tikzcd}[column sep=0pt]
			&[-1.6cm] &[-1.5cm](A \otimes (B\otimes C))\otimes D \arrow[start anchor={[shift={(2pt,0)}]east}, end anchor={[shift={(0pt,0pt)}]north}]{rrd}&[-1cm] &[-2cm] \\
			[10pt]((A\otimes B)\otimes C)\otimes D \arrow{rd}\arrow[start anchor={[shift={(0,0)}]north}, end anchor={[shift={(2pt,0pt)}]west}]{rru} & & & & A \otimes ((B\otimes C)\otimes D) \arrow{dl} \\[15pt]
			&(A\otimes B)\otimes (C\otimes D) \arrow{rr}& &A \otimes (B\otimes (C\otimes D))& \\
		\end{tikzcd}
	}
	\caption{\label{pentagon}\small Pentagonal coherence axiom for monoidal categories}
\end{figure}
\begin{figure}[!htb]\centering
	\begin{tikzcd}[column sep=0pt]
		(A\otimes E) \otimes B\arrow{rd} \arrow{rr} & &\arrow{ld} A\otimes(E\otimes B)\\
		&A\otimes B&
	\end{tikzcd}
	\caption{\label{triangle}\small Triangle coherence axiom for monoidal categories}
\end{figure}

The pentagon identities Fig.\ref{pentagon} and triangle identities Fig.\ref{triangle} imply commutativity of all diagrams which contain the associativity constraint, the natural isomorphisms $\ell$ and $r$. The following definition is motivated by the fact that in a tensor category, there is in general no canonical morphism from an object to its monoidal product with another object.

\begin{definition}[Monoidal category with inclusions]
	\label{monoidal}
	A monoidal category with inclusions $(\mathcal{C}, \otimes, \iota)$ is a tensor category $\left(\mathcal{C}, \otimes \right)$ in which, for any two objects $B_{1}$ and $B_{2}$, there exist two morphisms $\iota_{B_{1}} : B_{1} \rightarrow B_{1} \otimes B_{2}$ and $\iota_{B_{2}} : B_{2} \rightarrow B_{1} \otimes B_{2}$ such that for any two objects $\mathcal{A}_1$ and $\mathcal{A}_2$ and any two morphisms $f_{1} : \mathcal{A}_1 \rightarrow B_{1}, f_{2} : \mathcal{A}_2 \rightarrow B_{2}$, the following diagram (Fig \ref{fig:injectionnatural}) commutes:
	\begin{figure}[!htb]\centering
		\begin{tikzcd}[ampersand replacement=\&]
			\mathcal{A}_1 \arrow{r}{\iota_{\mathcal{A}_1}} \arrow{d}{f_{1}} \& \mathcal{A}_1 \otimes \mathcal{A}_2 \arrow{d}{f_{1} \otimes f_{2}}\& \arrow[swap]{l}{\iota_{\mathcal{A}_2}} \arrow{d}{f_{2}} \mathcal{A}_2 \\
			B_{1} \arrow{r}{\iota_{B_{1}}} \& B_{1} \otimes B_{2} \& \arrow[swap]{l}{\iota_{B_{2}}} B_{2}
		\end{tikzcd}
    \caption{\label{fig:injectionnatural} The natural transformations $\iota_{1}$ and $\iota_{2}$.}
	\end{figure}
\end{definition}
Let $\mathcal{P}_{i}:\mathcal{C}\times\mathcal{C}\to \mathcal{C}$, $i\in\{1,2\}$ be the projections functors on the first and second component. We can reformulate Definition \ref{monoidal}, a monoidal category with inclusions is a monoidal category with \emph{two natural transformations}\footnote{\url{https://ncatlab.org/nlab/show/natural+transformation}} $\iota_{1}: P_{1} \Rightarrow \otimes$ and $\iota_{2}:P_{2}\Rightarrow \otimes $. The natural transformation $\iota_{1}$ is called a \emph{left inclusion} and $\iota_{2}$ is called a \emph{right inclusion}.
We can now give a general definition of independence of two morphisms.
\begin{definition}
	\label{independence2}
	Let $\left(\mathcal{C},\otimes, \iota \right)$ be a monoidal category with injections. Two morphisms $f_{1} : C_{1} \rightarrow A$ and $f_{2} : C_{2} \rightarrow A$ are said to be \emph{independent} if there exists a third morphism $f: C_{1} \otimes C_{2} \rightarrow A$ such that the following diagram (Fig. \ref{fig:independencemorphisms})commutes:
	\begin{figure}[!htb]\centering
		\begin{tikzcd}[ampersand replacement=\&]
			\&  A   \& \\
			\arrow{ru}{f_{1}}\arrow[swap]{r}{\iota_{C_{1}}}C_{1} \& \arrow{u}[swap]{f}  C_{1} \otimes C_{2} \& \arrow{l}{\iota_{C_{2}}} C_{2} \arrow{lu}[swap]{f_{2}}
		\end{tikzcd}
    \caption{\label{fig:independencemorphisms} Independence of two morphisms $f_{1}$ and $f_{2}$.}
	\end{figure}
\end{definition}
If we want to be explicit about the monoidal structure that is involved, we will say that $f_{1}$ and $f_{2}$ are $\otimes$-independent. Definition \ref{independence} defines what it means for two morphisms to be independent but not what mutual independence of a finite set of morphisms is. To define mutual independence we need further assumptions on the monoidal structure, see just after Example \ref{ex:categomonoid}.
The morphism $f$ of the last definition is called an independence morphism. In most examples this morphism, if it exists, will be unique.
We warn the reader: a coproduct with injections is a monoidal product with injections for which \emph{any} two morphisms with the same target space are independent. A monoidal product with injections is not a coproduct.

Let $R$ be an unital associative algebra and denote by $\Prob(R)$ the category of algebras endowed with an expectation (we do not require the expectation to be a positive map). In Example \ref{ex:algcat}, we saw that the category $\Alg^{\star}$ is algebraic, and the coproduct is the free product of algebras.  Pick a monoidal product $\otimes$ on $\Prob$. Let $A$ and $B$ be two probability spaces and write $A \otimes B = \left(C, \tau_{A \otimes B}\right)$ with $C\in\Alg^{\star}$. From the universal property satisfied by the co-product, there exists a morphism of involutive algebras $\pi : A \sqcup B \rightarrow C$. Define the bi-functor $\otimes^{\prime}$ on $\Prob$ by, for two objects $A$ and $B$ $\in \Prob$:
\begin{equation}
	A \otimes^{\prime} B = \left(A \sqcup B, \tau_{A\otimes B}\circ \pi \right)
\end{equation}
and for two morphisms $f:\mathcal{A}_1\to \mathcal{A}_2$, $g:B_{1}\to B_{2}$, $\otimes^{\prime}(f,g) = f \free g$. We prove that $\otimes^{\prime}$ is a monoidal product with injections on $\Prob$. One fact needs to be proved. Consider two objects $(A_{i},\tau_{A_{i}}), (B_{i},\tau_{B_{i}}) \in \{1,2\}$ and two morphisms $f:(\mathcal{A}_1,\tau_{\mathcal{A}_1}) \to (\mathcal{A}_2,\tau_{\mathcal{A}_2})$ and $g:(B_{1},\tau_{B_{1}})\to (B_{2},\tau_{B_{2}})$ of $\Prob$. First, we prove that $\tau_{\mathcal{A}_2\otimes B_{2}}\circ (f \free g) = \tau_{\mathcal{A}_1\otimes B_{1}}$. To that end, we draw the commutative diagram in Fig. \ref{otimesprime}, blue arrows are morphisms of the category $\Prob$, while black arrows are morphisms in the category $\Alg^{\star}$.
\begin{figure}[!htb]\centering
	\begin{tikzcd}[every matrix/.append style={name=m},
			execute at end picture={
					\draw [<-,>=latex] ([yshift=6mm,xshift = 10mm]m-3-1.east) arc[start angle=-180,delta angle=270,radius=0.15cm];
					\draw [<-,>=latex] ([yshift=6mm,xshift = 0mm]m-3-2.east) arc[start angle=-180,delta angle=270,radius=0.15cm];
					\draw [<-,>=latex] ([yshift=5mm,xshift = 15mm]m-2-1.east) arc[start angle=-180,delta angle=270,radius=0.15cm];
					\draw [<-,>=latex] ([yshift=5mm,xshift = -18mm]m-2-3.west) arc[start angle=-180,delta angle=270,radius=0.15cm];
				}]
		&(\mathcal{A}_1 \otimes B_{1},\tau_{\mathcal{A}_1\otimes B_{1}})\arrow[bend right=90,cyan]{ddd}& \\
	(\mathcal{A}_1,\tau_{\mathcal{A}_1})\arrow[cyan]{r}\arrow[cyan]{ru}\arrow[cyan,"f"]{d}&(\mathcal{A}_1\free B_{1},\tau_{\mathcal{A}_1\otimes B_{1}} \circ \pi)\arrow[cyan]{u}\arrow["f\free g"]{d}& (B_{1},\tau_{B_{1}})\arrow[cyan]{lu}\arrow[cyan,"g"]{d}\arrow[cyan]{l}\\
	(\mathcal{A}_2,\tau_{\mathcal{A}_2})\arrow[cyan]{dr}\arrow[cyan]{r}& (\mathcal{A}_2 \free B_{2},\tau_{\mathcal{A}_2\otimes B_{2}} \circ \pi)\arrow[cyan]{d}& \arrow[cyan]{l}(B_{2},\tau_{B_{2}})\arrow[cyan]{dl} \\
		&(\mathcal{A}_2\otimes B_{2},\tau_{\mathcal{A}_2\otimes B_{2}})&
	\end{tikzcd}
	\caption{\label{otimesprime} The morphism $f\sqcup g$ is trace-preserving, $\tau_{\mathcal{A}_2\otimes B_{2}}\circ f \sqcup g  = \tau_{\mathcal{A}_1\otimes B_{1}}$.}
\end{figure}
In Fig. \ref{otimesprime}, the morphisms $\iota_{A_{i}}:A_{i}\to A_{i} \free B_{i}$ and $\iota_{B_{i}}:B_{i}\to A_{i} \free B_{i}$ are drawn in blue. In fact, it is easily seen that preserving the trace for these morphisms is equivalent to the commutativity of the four triangles in Fig. \ref{otimesprime}.
To show that the free product $f \free g$ preserves traces, it is enough to show the commutativity of the two outer faces, the ones bordered with the blue peripheral arrows sharing the curved downward arrow as an edge. The commutativity of these two faces is implied by the universal property satisfied by $\otimes$. In conclusion, the arrow $f\sqcup g$ is equal to a composition of blue arrows and is thus trace-preserving.
Assume that $\mathcal{A}_2 = B_{2}$. If the two morphisms $f$ and $g$ are $\otimes$-independent then they are also $\otimes^{\prime}$-independent. Hence, in the sequel, there is no loss in replacing the monoidal product $\otimes$ by $\otimes^{\prime}$ and thus assuming that the underlying involutive algebra of the monoidal product of two probability spaces is the free product of the algebras. We provide a few examples of monoidal structures with inclusions on the categories of amalgamated probability spaces, see \cite{franz2004theory} for a detailed overview. Of which we recall the definition. 

With $R$ an involutive unital complex algebra, we denote by $\Prob(R)$ the category whose objects are pairs $(A, E_A)$ where $A$ is \emph{bimodule algebra} over $R$ endowed with a compatible involution and an \emph{expectation} $E_A: A\to R$ is bimodule map compatible with the involutions on $A$ and $R$. Recall that a bimodule algebra is at the same time a right module algebra and a left module algebra, that the two actions commute and they satisfy the cross-associativity $(r_1a)r_2=r_1(ar_2)$, $r_1,r_2 \in R, a \in A$. Morphisms between two amalgamated probability spaces are expectations preserving involutive bimodule algebra morphisms. When $R=\C$, we simply call $\Prob(C)$ the category of probability spaces. 

When $R$ is \emph{commutative}, we denote by {\rm comProb}$(R)$ the category of commutative probability spaces over $R$, objects are probability spaces over $R$, but commutative as algebras and, moreover, commutative as bimodules: the left and right actions of $R$ are equal.
\begin{example}
\label{ex:categomonoid}
\begin{enumerate}[\indent 1.]
	\item The category $\Prob(\mathbb{C})$ of complex probability spaces can be endowed with several monoidal structures. Three of them that will be interesting for the present work and are denoted by $\hat{\otimes}$, $*$ and $\diamond$. Let $(A,E_{A})$ and $(B,E_{B})$ be two probability spaces; $E_{A}:A\to\mathbb{C}$ and $E_{B}:B\to\mathbb{C}$ are two complex positive unital linear forms.
	      \begin{enumerate}[\indent a.]
		      \item The linear for $E_{A}\,\hat{\otimes}\,E_{B}$ is a linear form on the free product of algebras $A \sqcup B$ and is defined, setting for an alternating word $s = a_1b_1\cdots a_pb_p\in A\free B$, $a_i \in A, b_i \in B, 1 \leq i \leq p$,
		            \begin{equation}
			            \label{tensorindep}
			            \left(E_{A}\hat{\otimes}E_{B}\right)(s) = E_{A}(a_{1}\cdots a_{p})E_{B}(b_{1}\cdots b_{p}).
		            \end{equation}
		            The bifunctor $\hat{\otimes}$ that send the pair of probability spaces $\left((A,E_{A}),(B,E_{B})\right)$ to $(A \sqcup B, E_{A} \hat{\otimes} E_{B})$ is a monoidal structure which is, in addition, \emph{symmetric}. This means in particular that there exists a \emph{commutativity constraint} $$S_{(A,E_A),(B,E_B)}:(A\sqcup B, E_A\hat{\otimes} E_B) \to (B\sqcup A, E_B\hat{\otimes} E_A)$$ for each pair of probability space $(A,E_A),(B,E_B)$. Besides $S$ is a natural transformation from $\hat{\otimes}$ to $\hat{\otimes}\circ \tau$ where $\tau$ swap the two arguments of $\hat{\otimes}$.
                    $S_{(A,E_{A}),(B,E_B)}$ is the algebra morphism
                    determined by $S_{(A,E_{A}),(B,E_B)} \circ \iota^{A\sqcup B}_{A} = \iota_{A}^{B\sqcup A}$ and $S_{(A,E_{A}),(B,E_B)} \circ \iota^{A\sqcup B}_{B} = \iota_{B}^{B\sqcup A}$. 
              This monoidal structure is related to the notion of \emph{universal tensor independence}. To define $\hat{\otimes}$ we have used the free product of star algebras. Instead, we could have used the monoidal product of algebras. In fact, the functor $\otimes$ (with a slight abuse of notations) that sends $(A, E_{A})$ and $(B,E_{B})$ to $(A\otimes B, E_{A}\otimes E_{B})$ is, first, well defined and is a monoidal structure on $\Prob(\mathbb{C})$.
		            With obvious notations, two morphisms $f:(A,E_{A}) \to (C,E_{C})$ and $g:(B,E_{B})\to (C,E_{C})$ are $\otimes$ independent if and only if:
		            \begin{enumerate}[\indent 1.]
			            \item \label{tensorind1} $\forall a\in A,~b \in B,~[f(a),g(b)] = 0$,
			            \item \label{tensorind2}$E_{C}\left(f(a_{1})g(b_{1}) \cdots f(a_{p})g(b_{p})\right) = E_{A}\left(f(a_{1}\cdots a_{p})\right)E_{B}\left(g(b_{1}\cdots b_{p})\right)$.
		            \end{enumerate}
		            Notice that the two morphims $f$ ang $g$ are $\hat{\otimes}$ independent if only point $\ref{tensorind2}$ holds. 
		            Of course, the map $\tau_{A\otimes B}: A \otimes B \to B\otimes A$ equal to the identity on $A$ and $B$ is a state preserving morphism: $\tau_{A \otimes B}\circ E_{B} \otimes E_{A} = E_{A}\otimes E_{B}$ and $\otimes$ is symmetric as well.
		      \item The second monoidal structure is related to the notion of \emph{boolean independence}. The state $E_{A} \diamond E_{B}$ is defined on the free product $A\free B$ and satisfies:
		            \begin{equation*}
			            \left(E_{A}\diamond E_{B}\right)(s_{1} \cdots s_{p}) = E_{\varepsilon_{1}}(s_{1}) \cdots E_{\varepsilon_{p}}(s_{p}).
		            \end{equation*}
		            with $s_{1},\ldots,s_{p} \in A\cup B$ and $s_1\cdots s_p$ an alternated word, $\varepsilon_{i} = A$ if $s_{i} \in A$ and $\varepsilon_{i} = B$ if $s_{i} \in B$. Since $E_{A}$ and $E_{B}$ are unital map, the boolean product is well defined. The bifunctor $\diamond$ that send the pair of probability spaces $\left(A,E_{A}\right)$, $\left(B,E_{B}\right)$ to $\left(A\free B, E_{A}\diamond E_{B} \right)$ is a symmetric monoidal structural.
		      \item The third and last monoidal structural we define is related to the notion of \emph{free independence}, at the origin of free probability theory. The free product of $E_{A}$ and $E_{B}$ is defined by the following requirement. For any alternating word with $s_i \in A\free B$,
		            \begin{equation}
			            \label{freeproduct}
			            \left(E_{A}\,\star\,E_{B}\right)(s) = 0,~s_{i} \in {\sf ker}(E_{\varepsilon_{i}}),~i \leq p,~j\leq m.
		            \end{equation}
              where ${\sf ker}$ denotes the kernel.
	      \end{enumerate}
	      A formula for $E_{A}\,\freem\,E_{B}(s)$ can be computed inductively. Computation of the free product of $E_{A}$ and $E_{B}$ are simple for words with small lengths,
	      \begin{equation*}
		      \begin{split}
			      &(E_{A}\star E_{B})(ab)= E_{A}(a)E_{B}(b), \\
			      &(E_{A}\star E_{B})(a_{1}ba_{2}) = E_{A}(a_{1})E_{B}(b)E_{A}(a_{2}) + E_{A}(a_{1}E_{B}(b)a_{2}) + E_{A}(a_{1})E_{B}(b)E_{A}(a_{2}) \\ &\hspace{10cm}-2E_{A}(a_{1})E_{B}(b)E_{A}(a_{2}) \\
                 &\hspace{2.8cm}=E_A(a_1E_B(b)a_2)
		      \end{split}
	      \end{equation*}
	      See \cite{speicher1998combinatorial} and \cite{nica2006lectures}.
	\item Let $R$ be an unital associative algebra, we define two monoidal structures on the category of amalgamated probability spaces, $\Prob(R)$. In case $R$ is commutative, we define a monoidal structure on the category com\Prob$(R)$. 
    Let $(A,E_{A})$ and $(B,E_{B})$ be two probability spaces in $\Prob(R)$.
	      \begin{enumerate}[\indent a.]
		      \item We begin with the amalgamated free product $E_{A}\star_{R}E_{B}$ of $E_{A}$ and $E_{B}$. It is defined by the equation \eqref{freeproduct}, for all alternating word $s \in A\free_{R}B$ in the amalgamated free product $A\free_{R}B$, $E_{A}\star_{R}E_{B}$ is defined by requiring:
		            \begin{equation*}
			            \left(E_{A}\,\star_{R}\,E_{B}\right)(s) = 0,~s_{i} \in {\sf ker}(E_{\varepsilon_{i}}),~i \leq p,~j\leq m.
		            \end{equation*}
		      \item The amalgamated boolean product $E_{A}\diamond_{R}E_{B}$ of the two expectations $E_{A}$ and $E_{B}$ is defined by:
		            \begin{equation*}
			            \left(E_{A} \diamond_{R} E_{B}\right)(s_{1}\cdots s_{p}) = E_{\varepsilon_{1}}(s_{1}) \cdots E_{\varepsilon_{p}}(s_{p})
		            \end{equation*}
		            with $s_{1},\ldots,s_{p} \in A \cup B$ and $s_1\cdots s_p$ an alternated word , $\varepsilon_{i} = A$ if $s_{i} \in A$ and $\varepsilon_{i} = B$ if $s_{i} \in B$.
	      \end{enumerate}
	\item We defined amalgamated versions of the notion of boolean and free independence, using essentially the same formulae as for the non-amalgamated case. We can not do so for tensor independence, at least for two reasons. First, the amalgamated monoidal product $A\otimes_{R}B$ with bimodule structure given by
	      \begin{equation*}
		      r(a \otimes_{R} b)r^{\prime}=ra\otimes_{R} br^{\prime},~a \in A,~b \in B,~r,r^{\prime} \in R,
	      \end{equation*}
	      is not naturally an algebra, meaning that the canonical product $A\otimes B$ does not descend to a product on $A{\otimes}_{R}B$. Secondly, formula \eqref{tensorindep} can not be used to define an amalgamated version of tensor independence, since there may be $a_{1},a_{2}\in A$, $b_{1} \in B$ and $r \in R$ with $E(a_{1}ra_{2})E(b_{1})\neq E(a_{1}a_{2})E(rb_{1})$,~
	      Finally, in the last section of the present work we use amalgamated probability spaces over commutative algebras $R$.
    \end{enumerate}
\end{example}
We mentioned earlier there might be an issue if we try to define independence of more than two morphisms. An additional constraint needs to be satisfied by the natural morphisms $\iota^{1}$, $\iota^{2}$ and the unit $E$. This is the content of the next definition.
\begin{definition}
	\label{inclusionunit}
	Inclusions $\iota^{1}$ and $\iota^{2}$ are called compatible with unit constraints if the diagram
	\begin{figure}[!htb]\centering
		\begin{tikzcd}
			E \otimes A\arrow{rd} & \arrow{l}A\arrow{r}\arrow{d}&\arrow{dl} A \otimes E \\
			&A&
		\end{tikzcd}
		\caption{\label{fig:inclusionunit}\small Compatibilty constraint between inclusions and unit}
	\end{figure}
\end{definition}
A more vernacular way to put constraints of Definition \ref{inclusionunit} is that inclusions $\iota^{1}$ and $\iota^{2}$ and left and right unit morphisms are, respectively, inverse from each other.
\begin{theorem}[Theorem 3.6 in \cite{gerhold2016categorial}] Let $(\mathcal{C},\otimes,E,\alpha,\ell,r,\iota^{1},\iota^{2})$ be a monoidal category.
	\begin{enumerate}[\indent 1. ]
		\item If $\iota^{1}$ and $\iota^{2}$ are inclusions which are compatible with unit constraints, the unit object $E$ is initial, i.e there is a unique morphism $\eta_{A}:E \to A$ for every object $A \in \Obj(\mathcal{C})$.
		      Furthermore,
		      \begin{equation}
			      \tag{$\star$}
			      \label{inclusionsinitial}
			      \iota^{1}_{A,B} = \left(\id_{A} \otimes \eta_{B} \right) \circ r_{A}^{-1},\quad \iota^{2}_{A,B} = \left(\eta_{A} \otimes \id_{B} \right) \circ \ell_{B}^{-1}.
		      \end{equation}
		      holds for all objects $A,B,C \in \Obj(\mathcal{C})$.
		\item Suppose that the unit object $E$ is an initial object. Then \eqref{inclusionsinitial} read as a definition yields inclusions $\iota^{1}$, $\iota^{2}$ which are compatible with the unit constraints.
	\end{enumerate}
\end{theorem}
As pointed out in \cite{gerhold2016categorial}, Maclane's coherence theorem can be extended to all diagrams built up from the associative constraints, left and right units, and the natural morphisms $\eta$. This extended Maclane coherence theorem implies, in particular, commutativity of the diagrams in Fig. \ref{communitconstraintun} and Fig. \ref{communitconstraintdeux} in case compatibility with the unit constraints of inclusions are satisfied.
\begin{figure}[!htb]\centering
	\begin{minipage}{0.48\textwidth}
		\centering
		\begin{tikzcd}[ampersand replacement=\&]
			(A \otimes C)\arrow["\id_{A}\otimes\iota^{2}"]{d} \arrow["\iota^{1}\otimes\id_C"]{r} \& (A\otimes B) \otimes C \arrow["\alpha_{A,B,C}"]{ld}\\
			A \otimes (B\otimes C)
		\end{tikzcd}
	\end{minipage}%
	\begin{minipage}{0.48\textwidth}
		\centering
		\begin{tikzcd}[ampersand replacement=\&]
			A \otimes B \arrow["\iota^{1}"]{d}\&\arrow[swap,"\iota^{2}"]{l}B\arrow["\iota^{1}"]{r}\&B \otimes C \arrow[swap,"\iota^{2}"]{d} \\
			(A\otimes B) \otimes C \arrow["\alpha_{A,B,C}"]{rr}\& \& A \otimes (B \otimes C)
		\end{tikzcd}
	\end{minipage}
	\caption{\label{communitconstraintun} Compatibility between the injections and the associativity constraints.}
\end{figure}
\begin{figure}[!htb]\centering
	\begin{tikzcd}
		C \arrow["\iota^{2}"]{d} \arrow["\iota^{2}"]{r} & B \otimes C \arrow["\iota^{2}"]{d} \\
		(A \otimes B) \otimes C \arrow["\alpha_{A,B,C}"]{r} & A \otimes (B \otimes C)
	\end{tikzcd}
	\caption{\label{communitconstraintdeux} Compatibility between the injections and the associativity constraints.}
\end{figure}
From here onwards, monoidal structures with injections we consider satisfy the compatibility with unit constraints. If $1 \leq i_{1}<\cdots<i_{k}\leq n$, there are unique morphisms $\iota^{i_{1},\ldots,i_{k};n}_{\mathcal{A}_1,\ldots,A_{n}}:A_{i_{1}}\otimes\cdots\otimes A_{i_{i_k}} \to \mathcal{A}_1\otimes\cdots A_{n}$ that constitute natural transformations $\iota^{i_{1},\ldots,i_{k};n}$, referred to as inclusion morphisms.
In the sequel, the symbol $\mathcal{C}$ stands for a monoidal category such that the unit object is initial.
\begin{definition}[Mutual independence, Definition 3.7 in \cite{gerhold2016categorial}]
	Let $B_{1},\ldots,B_{n},A$ be objects in the category $\mathcal{C}$ and $f_{i}:B_{i}\to A$ morphisms. Then $f_{1},\ldots,f_{n}$ are called independent if there exists a morphism $h:B_{1}\otimes \cdots \otimes B_{n} \to A$ such that the following diagram (Fig. \ref{indepseq}) commutes:
	\begin{figure}[!htb]\centering
		\begin{tikzcd}
			&A \\
			B_{i}\arrow["f_{i}"]{ru}\arrow["\iota^{i;n}"]{r} &\arrow["h"]{u} B_{1} \otimes \cdots \otimes B_{n}
		\end{tikzcd}
		\caption{\label{indepseq} Mutual independence of a sequence of morphisms.}
	\end{figure}
\end{definition}
As a consequence of the existence of the natural injections $\iota_{\mathcal{A}_1,\ldots,A_{n}}^{i_{1},\ldots,i_{k}}$, a subsequence $(f_{i_{1}},\ldots,f_{i_{k}})$ of a sequence of independent morphisms $(f_{1},\ldots,f_{n})$ is independent for any tuple $(1 \leq i_{1} < \cdots < i_{k} \leq n)$. We warn the reader, we speak about independence of sequences of morphisms, since, with obvious notations, $(f,g)$ may be an independent sequence without $(g,f)$ being independent. The morphism $h$ in the commutative diagram of Fig. \ref{indepseq} is called an independence morphism. In most cases, this morphism is unique, and we use the notation $h_{f_{1},\ldots,f_{n}}$.

We end this section with a result on the category of probability spaces over an algebra $R$, central to the construction we will provide of a Master Field based on a $H$-algebra for symmetries, given the data of the Quantum Lévy process. We refer the reader to the Annexes for the definition of an inductively complete category and to the monograph.
\par The following theorem states that certain categories of bimodule-algebras over an unital associative algebra at stake in this work are closed for taking colimits; this is a rather abstract statement; the non-commutative analogue of Caratheodory's extension theorem in classical probability. Recall that the category $\Alg^{\star}(R)$ has been defined in point \ref{enum:alg*} of Example 1, the category $\biMod(R)$ in point \ref{def:bimod} of the same Example and the category $\Prob(R)$ right before Example \ref{ex:categomonoid}. We defer the proof of the next theorem to the Annexes.

We recall here the definition of a colimit. Pick a category $\mathcal{C}$ and a small category $\mathcal{D}$; a category whose class of objects is a set. Such a category is sometimes called a diagram. Pick a functor $F : \mathcal{C} \to \mathcal{D}$. A colimit of $\mathcal{F}$ in $\mathcal{C}$ is object $\lim_{\rightarrow} F$ together with morphisms $\iota_d : F(d)\mapsto \lim_{\rightarrow} F$ satisfying 
$$
\iota_{d^{\prime}} \circ F(f) = \iota_{d},~d,d^{\prime},~f:d\to d^{\prime},
$$
such that:
\begin{enumerate}
\item if $C \in \mathcal{C}$ is another object of $\mathcal{C}$,
\item and $\varphi_d : F(d)\to \mathcal{C},~d\in D$ are morphisms in $\mathcal{C}$ such that for objects $d$ and $d^{\prime}$ in $\mathcal{D}$ and morphism $f : d \to d^{\prime}$, 
$$
\varphi_{d^{\prime}} \circ F(f) = \varphi_{d}. 
$$
\end{enumerate}
There there exists an unique morphism $K : \lim_{\rightarrow} F \to C$ of $\mathcal{D}$ such that $K\circ \iota_{d} = \varphi_d$. The tuple $(\lim_{\rightarrow} F, (\iota_d)_{d \in \mathcal{D}})$ is called an universal (co)cone. 
When a colimit exists for any functor over any diagram exists, we say that the category is (co)complete.
We refer to the Annex for the definition of an upward-directed set and the proof of the following Theorem.
\begin{theorem}
\label{thm:takeda}
	When $\mathcal{C}=\Prob(R)$, $\Alg^{\star}(R)$ any functor from an upward direct set to $\mathcal{C}$ has a colimit in $\mathcal{C}$.
\end{theorem}
The above Theorem states a slightly weaker form of cocompletness. However, this is all what is needed to perform the construction of a Quantum Holonomy Field. In the more abstract setting of a Categorical Holonomy Field (see below), we should however require cocompletness, in order not to make too specific requirements.
\section{Quantum and Categorical Holonomy Fields}
\label{noncommasterfields}

This Section constitutes the core of our work. We begin with a gentle review of classical \emph{lattice} Holonomy Fields in the following Section \ref{sec:classicallattivehfield}. Sections \ref{sec:defiZhangHF} and \ref{sec:constructionlevyproc} are devoted to constructing Quantum Holonomy Fields from either a diagram of probability spaces over a directed set of sequences of lassos or a Quantum Lévy process. We will be more general than this and construct Holonomy Fields from a diagram of objects in a (co)complete category. 
\subsection{Classical lattice Holonomy Fields}
\label{sec:classicallattivehfield}
Before giving the main definitions of this paper, namely the definition of a Quantum Holonomy Field on a $H$-algebra (see Definition \ref{definition_master_field}), we will review briefly the classical notion of a Holonomy field on a lattice. 

In this paper, all the paths that we will consider will be piecewise affine continuous paths on the Euclidean plane $\R^{2}$. We denote by $\path(\R^{2})$ the set of all these paths. Each path $c$ has a starting point $\underline{c}$, an end point $\overline{c}$, an orientation, but it has no preferred parametrisation. Constant paths are allowed. Given two paths $c_{1}$ and $c_{2}$ such that $c_{1}$ finishes at the starting point of $c_{2}$, the concatenation of $c_{1}$ and $c_{2}$ is defined in the most natural way and denoted by $c_{1}c_{2}$. Reversing the orientation of a path $c$ results in a new path denoted by $c^{-1}$. A path which finishes at its starting point is called a {\em loop}. A loop which is as injective as possible, that is, a loop which visits twice its starting point and once every other point of its image, is called a {\em simple loop}. A path of the form $c\ell c^{-1}$, where $c$ is a path and $\ell$ is a simple loop, is called a {\em lasso} (see Fig. \ref{lasso}).

\begin{figure}[!htb]\centering
	\includegraphics[scale=1.5]{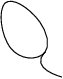}
	\caption{\label{lasso} \small A lasso drawn on the plane.}
\end{figure}

Let us call {\em edge} a path that is either an injective path (thus with distinct endpoints) or a simple loop. By a {\em graph} on $\R^{2}$, we mean a finite set $\mathbb E$ of edges with the following properties:
\begin{enumerate}[\indent 1.]
	\item for all edge $e$ of $\mathbb E$, the edge $e^{-1}$ also belongs to $\mathbb E$,
	\item any two edges of $\mathbb E$ which are not each other's inverse meet, if at all, at some of their endpoints,
	\item the union of the ranges of the elements of $\mathbb E$ is a connected subset of $\R^{2}$.
\end{enumerate}
From the set $\mathbb E$, we can form the set $\mathbb V$ of {\em vertices} of the graph, which are the endpoints of the elements of $\mathbb E$, and the set $\mathbb F$ of {\em faces} of the graph, which are the bounded connected components of the complement of the union of the range of the edges of the graph. Although it is entirely determined by $\mathbb E$, it is the triple $\mathbb G=(\mathbb V,\mathbb E,\mathbb F)$ that we regard as the graph.

From the graph $\mathbb G$, and given a vertex $v$, we form the set ${\sf L}_{v}(\mathbb G)$ of all loops based at $v$ that can be obtained by concatenating edges of $\mathbb G$.  The operation of concatenation makes ${\sf L}_{v}(\mathbb G)$ a monoid, with the constant loop as the unit element. In order to make a group out of this monoid, one introduces the backtracking equivalence of loops and the notion of reduced loops. A loop is {\em reduced} if in its expression as the concatenation of a sequence of edges (which is unique) one does not find any two consecutive edges of the form $ee^{-1}$. We denote by ${\sf RL}_{v}(\mathbb G)$ the subset of ${\sf L}_{v}(\mathbb G)$ formed by reduced loops. It is however not a submonoid of ${\sf L}_{v}(\mathbb G)$, for the concatenation of two reduced loops needs not to be reduced. The appropriate operation on ${\sf RL}_{v}(\mathbb G)$ is that of opposite concatenation followed by reduction where, as the name indicates, one concatenates two loops, the second one preceding the first one, and then erases sub-loops of the form $ee^{-1}$ until no such loops remain. It is true, although perhaps not entirely obvious, that this operation is well defined, in the sense that the order in which one erases backtracking sub-loops does not affect the final reduced loop.

From the graph $\mathbb G$ and the vertex $v$, we thus built a group ${\sf RL}_{v}(\mathbb G)$. This group is in fact isomorphic in a very natural way to the fundamental group based at $v$ of the subset of $\R^{2}$ formed by the union of the edges of $\mathbb G$. An important property of the group ${\sf RL}_{v}(\mathbb G)$ is that it is generated by lassos (see Fig. \ref{declasso})

\begin{figure}[!htb]\centering
	\includegraphics[scale=1]{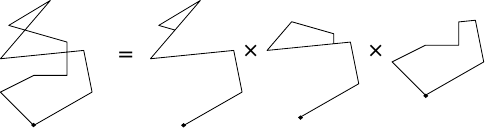}
	\caption{\label{declasso}\small Decomposition of a loop into a product of lassos based at $v=\bullet$.}
\end{figure}

In fact, ${\sf RL}_{v}(\mathbb G)$ is a free group, with a rank equal to the number of faces of $\mathbb G$, and it admits bases formed by lassos. More precisely, we will use the following description of a basis of this group. Let us say that a lasso $c\ell c^{-1}$ surrounds a face $F$ of the graph if the loop $\ell$ traces the boundary of $F$.

\begin{proposition}[Proposition 5.12 in \cite{gabriel2015planar}]
\label{prop:fondresult}
The group ${\sf RL}_{v}(\mathbb G)$ admits a basis formed by a collection of lassos, each of which surrounds a distinct face of $\mathbb G$.
\end{proposition}

In fact, the group ${\sf RL}_{v}(\mathbb G)$ admits many such bases, but it will be enough for us to know that there exists one. A proof of this proposition, and more details about graphs in general, can be found in \cite{levy2008two}.

A classical lattice gauge field on $\mathbb G$ with structure group $G$ is usually described as an element of the configuration space
\[\mathcal C_{\mathbb G}=\{g=(g_{e})_{e\in \mathbb E} \in G^{\mathbb E} : \forall e\in \mathbb E, g_{e^{-1}}=g_{e}^{-1}\}.\]
This configuration space is acted on by the lattice gauge group $\mathcal J_{\mathbb G}=G^{\mathbb V}$, according to the formula
\[(j\cdot g)_{e}=j_{\overline{e}}^{-1} g_{e} j_{\underline{e}}.\]
Let us fix a vertex $v$ of our graph. Any element $g$ of the configuration space $\mathcal C_{\mathbb G}$ induces a Holonomy map ${\sf L}_{v}(\mathbb G)\to G$, which to a loop $\ell$ written as a concatenation of edges $e_{1}\ldots e_{n}$ associates the element $g_{e_{n}}\ldots g_{e_{1}}$ of $G$. This map descends to the quotient by the backtracking equivalence relation, and induces a morphism of groups ${\sf RL}_{v}(\mathbb G)\to G$. The action of a gauge transformation $j \in \mathcal J_{\mathbb G}$ on $\mathcal C_{\mathbb G}$ modifies this morphism by conjugating it by the element $j_{v}$.  These observations can be turned in the following convenient description of the quotient $\mathcal C_{\mathbb G}/\mathcal J_{\mathbb G}$.

\begin{proposition} For all $v\in \mathbb V$, the Holonomy map induces a bijection
	\[\mathcal C_{\mathbb G}/\mathcal J_{\mathbb G} \longrightarrow {\rm Hom}({\sf RL}_{v}(\mathbb G),G)/G.\]
\end{proposition}

It follows from this proposition that describing a probability measure on the quotient space $\mathcal C_{\mathbb G}/\mathcal J_{\mathbb G}$ is equivalent to describing the distribution of a random group homomorphism from ${\sf RL}_{v}(\mathbb G)$ to $G$, provided this random homomorphism is invariant under conjugation. Combining this observation with the fact that ${\sf RL}_{v}(\mathbb G)$ is a free group, and choosing for instance a basis $l_{1},\ldots,l_{n}$ formed by lassos surrounding the $n$ faces of $\mathbb G$, we see that a probability measure on $\mathcal C_{\mathbb G}/\mathcal J_{\mathbb G}$ is the same thing as a $G^{n}$-valued random variable $(H_{l_{1}},\ldots, H_{l_{n}})$, invariant in distribution under the action of $G$ on $G^{n}$ by simultaneous conjugation on each factor.

Let us finally introduce some further notation. We denote by ${\sf L}_{0}(\R^{2})$ the set of loops on $\R^{2}$ based at the origin which are concatenation of a finite sequence of segments and by ${\sf RL}_{0}(\R^{2})$ the group of reduced loops; this is the quotient of ${\sf L}_{0}$ by the backtracking relation. 

From now on, we will always assume that $0$ is a vertex of all the graphs that we consider. 

It is important to observe that any reduced loop on $\R^{2}$ belongs to ${\sf RL}_{v}(\mathbb G)$ for some graph $\mathbb G$. Thus,
\[{\sf RL}_{0}(\R^{2})=\bigcup_{\mathbb G \text{ graph}} {\sf RL}_{0}(\mathbb G).\]

Let us be more precise on that point. We explain how to build a graph $\mathbb{G}_{L}$, from any finite sequence of reduced loops $L$. We label this construction since we will refer to it in the next Section. 

We let ${\sf Lass}_0(\mathbb{R}^2)$ be the set of anti-clockwise oriented lassos in ${\sf RL}_0(\mathbb{R}^2)$.

\begin{construction}
\label{graph}
\par For any finite set of loops $L$ in ${\sf RL}_0(\mathbb{R}^2)$, there exists a set ${\sf Lass_0}({L}) \subset {\sf Lass_0}(\mathbb{R}^2)$ such that ${\sf RL}_0\langle L \rangle \subset {\sf RL}_0\langle{\sf Lass}_0({L})\rangle$.

In fact, any pair of loops in $L$ has the \emph{finite self-intersections property}. An intersection point between two loops ${\ell,\ell^{\prime}}$ is defined as a meeting point of two transverse segments, one of $\ell$ and one of $\ell^{\prime}$. The sequence $L$ has the finite self-intersections property if the total number of meeting points of any pair of loops in ${\sf L}$ is finite.

One builds an oriented graph $\mathbb{G}_L$ embedded in $\mathbb{R}^2$ out of the set of loops $L$ by declaring:

\begin{enumerate}
\item the vertices of $\mathbb{G}_L$ are the meeting points of the loops in $L$ 
\item the oriented edges of $\mathbb{G}_L$ are the segments of the loops in $F$ connecting two meeting points.
\end{enumerate}

See Fig. \ref{fig:graphgl} for an example of the graph $G_L$ and Fig. \ref{fig:graphlassos} for a basis of its group of reduced loops.
\begin{figure}[!htb]\centering
	\scalebox{0.4}{
	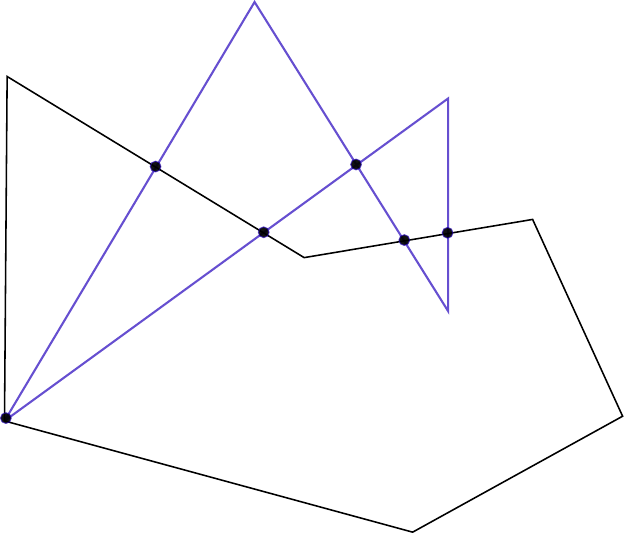
	}
	\caption{\label{fig:graphgl} Two affine loops are drawn on the plane, one is coloured in violet and the other in black. The vertices of the graph $\mathbb{G}_{L}$ are pictured with $\bullet$, the edges are the piecewise linear paths connecting these vertices.}
\end{figure}

\begin{figure}[!htb]\centering
	\scalebox{0.6}{
	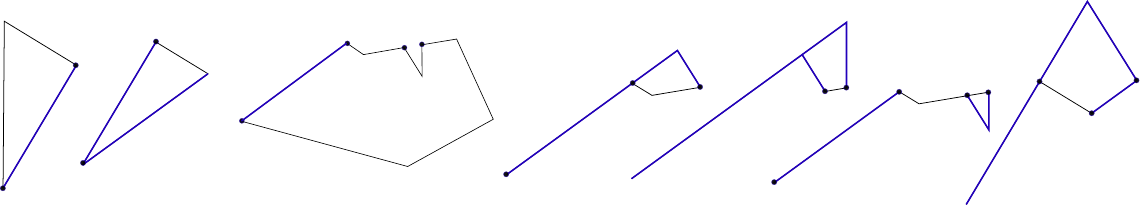
	}
	\caption{\label{fig:graphlassos} A basis of anticlockwise oriented lassos for the group of reduced loops drawn on the graph of Fig. \ref{fig:graphgl}.}
\end{figure}
\end{construction}

 We shall need the \emph{refinement order} on graphs. For two graphs $\mathbb{G}$ and $\mathbb{G}^{\prime}$, we write $\mathbb{G} \prec \mathbb{G}^{\prime}$ if all edges of $\mathbb{G}$ are obtained as concatenations of edges in $\mathbb{G}^{\prime}$.

 Now, with these definitions, 
 $$
 {\sf RL}_0(\mathbb{R}^2) = \underset{\rightarrow}{\lim}~{\sf RL}_0(\mathbb{G}).
 $$

\subsection{Holonomy Fields on a $H$-algebra: Definition(s)}
\label{sec:defiZhangHF}
We proceed with the definitions of Quantum Holonomy Field and of Categorical Holonomy Field in the next two sections.
\subsubsection{Quantum Holonomy Field}
\label{sec:quantumholo}
In the following, we set $\mathbb{K} = \mathbb{R}$ or $\mathbb{C}$. Let $R$ be a unital and involutive algebra over $\mathbb{K}$. Recall from Example \ref{enum:alg*} page \pageref{enum:alg*} that the category of involutive bimodule algebras over $R$ is denoted $\Alg^{\star}(R)$. We fix a non-commutative algebraic probability space $\mathcal{A}\in \Prob(R)$, that is, a pair $(\mathcal{A},\tau_{\mathcal{A}})$ where $\mathcal A$ is an object of $\Alg^{\star}(R)$ and $\tau_{\mathcal{A}}:\mathcal A\to R$ is, in particular a morphism of bimodules,
$$
	\tau_{\mathcal{A}}(r_{1}ar_{2}) = r_{1}\tau_{\mathcal{A}}(a)r_{2},~a \in \mathcal{A},~r_{1},r_{2} \in R.
$$

We choose a monoidal structure $(\otimes, R)$ with unit $R$ on the category \Prob($R$). At that point, we should make two \emph{restrictive assumptions}:
\begin{enumerate}
	\item \label{restrictiveun} With $(\mathcal{A}_{1}, \tau_{\mathcal{A}_1}), (\mathcal{A}_{2}, \tau_{\mathcal{A}_2}) \in \Prob(B)$ the underlying algebra of the product $(\mathcal{A}_{1},\tau_{\mathcal{A}_1})\otimes (\mathcal{A}_{2},\tau_{\mathcal{A}_{2}})$ is assumed to be the free product of algebra (with identification of units) $\mathcal{A}_1 \free \mathcal{A}_2$, that is:
 $$
 (A_1, \tau_{A_1})\otimes(A_2,\tau_{A_2}) = (A_1\sqcup A_2, \tau_{A_1} \otimes \tau_{A_2})
 $$
	\item \label{restrictivedeux} the monoidal product $\otimes$ is assumed to be \emph{symmetric} with commutativity constraint given by $S_{\mathcal{A}_1\sqcup \mathcal{A}_2}: \mathcal{A}_1 \sqcup \mathcal{A}_2 \rightarrow \mathcal{A}_2 \sqcup \mathcal{A}_1$, the coproduct of the two injections $\iota_{A_1}^{A_2 \sqcup A_1}$ and $ \iota_{A_2}^{A_2\sqcup A_1}$. In particular:
 $$
 \tau_{A_2} \otimes \tau_{A_1} \circ S_{A_1\sqcup A_2} = \tau_{A_1} \otimes \tau_{A_2}.
 $$
\end{enumerate}
These two assumptions have two important consequences:
\begin{enumerate}
	\item any permutation of a sequence of mutually independent morphisms is a sequence of mutually independent morphisms,
	\item the independence morphism associated with a sequence of independent morphisms is given by the free product of these morphisms.
\end{enumerate}
Finally, we choose a $H$-algebra $(H,\Delta, \varepsilon, \mathcal{\mathcal{S}})$ of $\Alg^{\star}(B)$. 
We are now ready to define the notion of interest for the present work.

\begin{definition}{(Quantum Holonomy Field on a $H$-algebra)}
	\label{definition_master_field}
	Under the assumptions \ref{restrictiveun} and \ref{restrictivedeux} on the monoidal product $\otimes$, \emph{An $H$-algebraic Holonomy field} is a group homomorphism 
    $$\H: {\sf RL}_{0}(\mathbb{R}^{2})\to \Hom_{\mathrm{Alg}^{\star}(B)}(H,\mathcal{A})$$ satisfying the following three properties:
	\begin{enumerate}[$\indent 1.$]
		\item \label{gaugeinv}\emph{(Gauge invariance)} Let $\ell_{1},\ldots,\ell_{n} \in {\sf L}_{0}(\mathbb{R}^{2})$ be a finite sequence of loops. For any morphism $\tau_{H}:H\to B$ of $\Alg^{\star}({B})$:
		      \[(\tau_{H} \otimes \tau_{\mathcal{A}}) \circ ((\H_{\ell_{n}} \freem \ldots \freem {\sf H}_{\ell_{1}})\free \id_H) \circ \Omega_{c}^{n} =\tau_{\mathcal A} \circ (\H_{\ell_{1}} \freem \ldots \freem {\sf H}_{\ell_{n}}).\]
        where we recall that:
        \[\Omega_{c}=(\iota_{1} \freem \iota_{2} \freem \iota_{1}) \circ (\id_{H\free H} \free \mathcal{\mathcal{S}}) \circ (\Delta \free \id_{H}) \circ \Delta.\]
		\item \emph{(Independence)} If $(\ell_{1},\ldots,\ell_{n})$ and $(\ell'_{1},\ldots,\ell'_{m})$ are two finite sequences of loops such that 
  $$\bigcup_{i=1}^{n} \mathrm{Int}(\ell_{i}) \textrm{ and } \bigcup_{j=1}^{m} \mathrm{Int}(\ell'_{j}) \textrm{ are disjoint },$$
then $\H_{\ell_{1}}\freem \ldots \freem  \H_{\ell_{n}}$ and $\H_{\ell'_{1}}\freem \ldots \freem  \H_{\ell'_{m}}$ are $\otimes$-independent.
    
		\item \emph{(Invariance by area-preserving homeomorphisms)} For any area-preserving diffeomorphism $T:\R^{2}\to \R^{2}$ and all loops $(\ell_{1},\ldots,\ell_{n})$, we have the  equality 
        $$\tau_{\mathcal{A}} \circ (\H_{\ell_{1}}\freem \ldots \freem \H_{\ell_{n}}) = \tau_{\mathcal{A}} \circ (\H_{T(\ell_{1})}\freem \ldots \freem \H_{T(\ell_{n})}).$$
	\end{enumerate}
\end{definition}

This generalisation of the notion of Planar Holonomy Field, although already wide enough to encompass classical gauge fields with compact structure groups as well as their large $N$ limits --- the so-called Master Field --- is set in a context that is less general than our definition of a $H$-algebra. Indeed, in this definition, we work on the particular algebraic category $\Alg^{\star}(B)$. It is not much more difficult to extend the definition to a setting where the algebraic category in which we take our probability spaces is almost arbitrary.

\subsubsection{Categorical Holonomy Field}
\label{sec:catholofield}
Let $(\mathcal{C},\free,k)$ be an algebraic category with coproduct $\free$ and unit $k$. Pick a second category $\mathcal{D}$ and $F:\mathcal{C}\to\mathcal{D}$ a \emph{faithful} (injective on the homomorphisms sets) and \emph{wide} (surjective on the class of objects) functor.
The category $\mathcal{D}$ is considered as an \emph{enlargement} of the category $\mathcal{C}$: the classes of objects are the same whereas the homomorphism sets between pairs of objects $O,O^{\prime} \in \mathcal{C}$ is larger in $\mathcal{D}$. Finally, set $F(k) = \star$. Morphisms of the category $\mathcal{D}$ with target spaces the object $\star$ are seen as \emph{distributions or states}. In fact, define a third category $\mathcal{B}$ which objects are pairs $(O,\phi_{O})$ with $O$ an object of $\mathcal{C}$ and $\phi$ a morphism of $\mathcal{D}$ from $F(O)$ to the initial object $k_{\mathcal{D}}\in\mathcal{D}$. A morphism between two objects $(O,\phi_{O})$ and $(O^{\prime},\phi_{O^{\prime}})$ in $\mathcal{B}$ is a morphism $f:O \to O^{\prime}$ of the category $\mathcal{C}$ such that $\phi_{O^{\prime}}\circ F(f) = \phi_{O}$. Although this construction may seem artificial, it is well-known in category theory. In fact, the category $\mathcal{B}$ is called a \emph{comma category} and we then write, after the work of Lawvere,
\begin{equation*}
	\mathcal{C} \rightarrow \mathcal{D} \leftarrow \star,~\mathcal{B} = \mathcal{C}\,\downarrow \, \star.
\end{equation*}
The comma category $\mathcal{B}$ comes with two important forgetful functors:
\begin{enumerate}[\indent $\scriptstyle \bullet$]
	\item {\sf Domain:} $\mathcal{B} \to \mathcal{C}$, $(O,\phi_{O})\mapsto O$, $f \mapsto f$
	\item {\sf Arrow:}  $\mathcal{B} \rightarrow \mathcal{D}^{\rightarrow}$, $(O,\phi_{O}) \mapsto \phi_{O} $, $ f \mapsto F(f)$.
\end{enumerate}

\begin{definition}[Lift of a morphism of $\mathcal{C}$]
	\label{def:morphism}
	Let $A$ be an object of $\mathcal{C}$ and $(B,\tau_{B})$ an object  $\mathcal{B}$. A morphism $f:A\to B $ of $\mathcal{C}$ lifts to a morphism $\mathcal{B}$, denoted $\hat{f}$, with source $(A,\tau_{B}\circ F(f))$ and target space $(B,\tau_{B})$.
\end{definition}

\begin{assumptions}
\label{assumptions:triptic}
We assume $\mathcal{B}$ to be endowed with a \emph{symmetric monoidal product} with injections $\left(\otimes, E\right)$ with $E$ an initial object of $\mathcal{B}$.
We assume that {\sf Domain} functor exchanges the monoidal structures of $\mathcal{B}$ and the algebraic structure of $\mathcal{C}$; it is a \emph{monoidal functor}. This implies:
\begin{equation*}
	{\sf Domain}((A,\tau_{A}) \otimes (B,\tau_{B})) = A \free B,\quad A,B \in \Obj(\mathcal{C}),~{\sf Domain} (f\otimes g) = f \sqcup g.
\end{equation*}
\end{assumptions}
The probabilistic setting of the last definition corresponds to the following triptic, with $R$ an unital involutive associative algebra:
\begin{equation*}
	\mathrm{Alg}^{\star}(B) \rightarrow \mathrm{BiMod}(B) \leftarrow B,~ \Prob(B) = \mathrm{Alg}^{\star}(B)\, \downarrow B
\end{equation*}

Again, pick $(\mathcal{A},\tau_{\mathcal{A}}) \in \mathcal{B}$
and a $H$-algebra $(H,\Delta,\varepsilon,\mathcal{S})$ of $\mathcal{C}$.
\begin{notation}
In the following definition, we use the shorter notation $A_{\ell_{1},\ldots,\ell_{n}}$ for the coproduct of a sequence $(A_{\ell_{1}},\ldots,A_{\ell_{n}})$ of morphisms of the category $\mathcal{C}$ from $H$ to $\mathcal{A}$,
$$
A_{\ell_1,\ldots,\ell_n} : H^{\free \, n} \to \mathcal{A},~A_{\ell_1,\ldots,\ell_n}\circ\iota_i = A_{\ell_i},~ 1 \leq i \leq n.
$$
\end{notation}
\begin{definition}{(Categorical Holonomy Field on a $H$-algebra)}
	\label{definition_master_field_categorical}
	Under Assumptions \ref{assumptions:triptic}, a $H$-categorical Holonomy field is a \emph{group homomorphism} $$\H: {\sf RL}_{0}(\mathbb{R}^{2})\to \Hom_{\mathcal{C}}(H,\mathcal{A})$$ satisfying to the following three assumptions:
	\begin{enumerate}[\indent 1.]

		\item \label{gaugeinv2}\emph{(Gauge invariance)} Let $\ell_{1},\ldots,\ell_{n} \in {\sf L}_{0}(\mathbb{R}^{2})$ be a finite sequence of loops. For each morphism $\tau_{H}:H\to k$ of the category $\mathcal{D}$:
		      \[(\tau_{\mathcal{A}} \otimes \tau_{H}) \circ ((\H_{\ell_{1},\ldots,\ell_{1}}) \free \id_H) \circ \Omega_{c}^{n} =\tau_{\mathcal A} \circ \H_{\ell_{1},\ldots,\ell_{n}}.\]
		\item \label{independence}\emph{(Independence)} If $(\ell_{1},\ldots,\ell_{n})$ and $(\ell'_{1},\ldots,\ell'_{m})$ are two finite sequences of loops such that 
  $$\bigcup_{i=1}^{n} \mathrm{Int}(\ell_{i}) \textrm{ and } \bigcup_{j=1}^{m} \mathrm{Int}(\ell'_{j}) \textrm{ are disjoint },$$ then
			      $\H_{\ell_{1}, \ldots, \ell_{n}}$ and $\H_{\ell'_{1}, \ldots , \ell'_{m}}$ are $\otimes$-independent.
		\item \label{areainv}\emph{(Invariance by area-preserving homeomorphisms)} For all area-preserving diffeomorphism $T:\R^{2}\to \R^{2}$ and any sequence of loops $(\ell_{1},\ldots,\ell_{n})$, the following equality holds: $$\tau_{\mathcal{A}} \circ F(\H_{\ell_{1},\ldots,\ell_{n}}) = \tau_{\mathcal{A}} \circ F(\H_{T(\ell_{1}), \ldots, T(\ell_{n})}).$$
	\end{enumerate}
\end{definition}

\begin{remarque}
In the next Section, we construct a categorical Holonomy Field that satisfies a smilingly strengthened gauge-invariance property. In fact, $\mathcal{A}$ will be a right comodule over $H$ built as a direct limit of right comodules over $H$ and the morphism $\tau_{\mathcal{A}}$ will be gauge-invariant, which means
\begin{equation}
\label{eqn:strenghtened}
	(\tau_{\mathcal{A}}\otimes \tau_{H}) \circ \overline{\Omega} = \tau_{\mathcal{A}}, \text{ for all }\phi_{H} \in \Hom_{D}(H, k).
\end{equation}
where ${\bar \Omega}: \mathcal{A} \to \mathcal{A} \free H$ is the right comodule morphism on $\mathcal{A}$, see Lemma \ref{lemma:lemmaomegabar} and the diagram in Fig. \ref{diag:co-action}. This will be explained in much detail in the next Section.
Equation \eqref{eqn:strenghtened} trivially implies gauge-invariance property $\ref{gaugeinv}$.
\end{remarque}
In Definition \ref{definition_master_field_categorical}, we use the co-product structure $\free$ of the category $\mathcal{C}$ to state what gauge-invariance and invariance by area-preserving homeomorphisms of the plane mean. This is not compulsory as we should see now and further in the next section on the course of defining a categorical Holonomy field.

Let $\ell_{1},\ldots,\ell_{n}$ be a finite sequence of loops and $c_{1},\ldots,c_{p}$ a family of lassos with disjoint bulks such that $\ell_{1},\ldots,\ell_{n} \in {\sf RL}_{0}(c_{1},\ldots,c_{p})$, we will see in the next section how such a family $(c_{1},\ldots,c_{p})$ can be obtained from $\ell_{1},\ldots,\ell_{n}$.
We claim that property \ref{gaugeinv2} for the sequence $(\ell_{1},\ldots,\ell_{n})$ can be obtained from the fact that \ref{gaugeinv} holds for the sequence of lassos $(c_{1},\ldots,c_{p})$. In fact, let $w = w_1\cdots w_q$ be a word on the symbols $c_1,\ldots,c_p$. We denote by $m^{w}$ the unique morphism from the coproduct $H_{w}:=H_{w_1}\free \cdots \free H_{w_q}$ to $H_c := H_{c_1}\free \cdots \free H_{c_p}$ where $H_{w_{i}} = H_{c_j} = H$, $1 \leq i \leq q$, $1 \leq j \leq p$ satisfying: 
$$ m^{w}\circ \iota_{H_{w_k}}^{H_{w}} =  \iota_{H_{w_{k}}}^{H_{c_1}\cdots H_{c_p}}.$$

We write each $\ell_i$ as a reduced word $w_i(c)$ on the lassos, the $c$'s. Then, with $\tau_H : H \to k$ a morphism of $\mathcal{C}$,
\begin{equation*}
	\begin{split}
		&(\tau_{\mathcal{A}} \otimes \tau_H) \circ (({\sf H}_{w_{1}(c)} \freem {\sf H}_{w_{2}(c)}\freem \cdots \freem {\sf H}_{w_{n}(c)}) \free \id_H)\circ \Omega_{c}^n
		\\
		&\hspace{1cm} = (\tau_{\mathcal{A}} \otimes \tau_H) \circ (({\sf H}_{c_{1}} \freem \cdots \freem {\sf H}_{c_{p}}) \free \id_H) \circ (m^{w_{1}(c)\cdots w_{n}(c)} \sqcup \id_H) \circ \Omega_{c}^n \\
		&\hspace{1cm} = (\tau_{\mathcal{A}}\otimes \tau_{H}) \circ (\id_{H} \free ({\sf H}_{c_{1}} \freem \cdots \freem {\sf H}_{c_{p}})) \circ \Omega_{c}^{p}\circ m^{w_{1}(c)\cdots w_{n}(c)}\\
		&\hspace{1cm} =\tau_{\mathcal{A}}\circ {\sf H}_{c_{1}}\freem \cdots \freem {\sf H}_{c_{n}} \circ m^{w_{1}(c)\cdots w_{n}(c)} \\
		&\hspace{1cm}= \tau_{\mathcal{A}}\circ ({\sf H}_{\ell_{1}}\freem\cdots\freem{\sf H}_{\ell_{p}}),
	\end{split}
\end{equation*}
where the first and last equality follow from  ${\sf H}_{c_{1}}\freem \cdots \freem {\sf H}_{c_{n}} \circ m^{w_{1}(c)\cdots w_{n}(c)} = {\sf H}_{\ell_{1}}\freem\cdots\freem{\sf H}_{\ell_{p}}$ which follows from unicity of the solution to the universal product defining the coproduct $\sqcup$. For the second equality, we remark that 
$$
 \Omega_{c}^{p}\circ m^{w_{1}(c)\cdots w_{n}(c)} \circ \iota_{H_{{w(k)}}^{H_{w}}} = (({\sf H}_{w_{1}(c)} \freem {\sf H}_{w_{2}(c)}\freem \cdots \freem {\sf H}_{w_{n}(c)}) \free \id_H)\circ \Omega_{c}^n \circ \iota_{H_{{w(k)}}^{H_{w}}}
$$
which follows in turn from the elementary fact that $\Omega_c$ is an ``algebra'' morphism, $m_{H\sqcup H} \circ \Omega_c = \Omega_c \circ m_H $. 
We leave to the reader the verification that invariance by area-preserving homeomorphisms property \ref{areainv} and independence property \ref{independence} hold for any sequence of loops if it holds for all sequence of lassos.

To sum up, if ${A}$ is a right comodule on $H$ with co-action $\overline{\Omega}$, properties $1.$ -- $3.$ of Definition $\ref{definition_master_field_categorical}$ are implied by:
\begin{enumerate}[\indent $1^{\prime}$.]
	\item \label{independenceprime}\emph{(Independence)} If $(c_{1},\ldots,c_{n})$ is a finite sequences of lassos with two by two disjoint bulks, ${\sf H}_{c_{1}},\ldots,{\sf H}_{c_{n}}$ is a $\otimes$ - mutually independent family of morphisms.
	\item \label{gaugeinvprime}\emph{(Gauge invariance)}
	      For any lasso $c$, $\H_{c}$ is a morphism of $H$-comodules and
	      \begin{equation*}
		      (\tau_{H}\otimes \tau_{\mathcal{A}}) \circ \overline{\Omega} = \tau_{\mathcal{A}}, \text{ for all }\phi_{H} \in \Hom_{D}(H, k).
	      \end{equation*}
	\item \label{areainvprime}\emph{(Invariance by area-preserving homeomorphisms)} For any area-preserving diffeomorphism $T:\R^{2}\to \R^{2}$ and any lasso $c$, we have the equality
	      $$\tau_{\mathcal{A}} \circ F(\H_{c}) = \tau_{\mathcal{A}} \circ F(\H_{T(c)}).$$
\end{enumerate}

\begin{definition}
\label{def:distroHolonomy}
\emph{The distribution of a categorical Holonomy field} ${\sf H}$ is the collection $\{\Phi^{{\sf H}}_{\ell},~\ell \in {\sf RL}_{\text{Aff},0}(\mathbb{R}^{2})\}$ of morphisms from $H$ to $k$ defined by:
\begin{equation}
	\Phi^{{\sf H}}_{\ell} = \tau_{\mathcal{A}} \circ {\sf H}_{\ell},~\ell \in {\sf RL}_{0}(\mathbb{R}^{2}).
\end{equation}
\end{definition}
We denote by $\star_{\otimes}$ the product on the space homomorphisms from $H$ to $k$ in the category $\mathcal{D}$ and defined by:
\begin{equation*}
	\alpha \star \beta = (\alpha \otimes \beta) \circ \Delta,~ \alpha,\beta \in \Hom_{\mathcal{D}}(H,k).
\end{equation*}
We do not record in the symbol for the product between morphisms in $\textrm{Hom}_{\mathcal{D}}(H,k)$ the dependence toward $\otimes$. It should be clear from the context.
Properties $1.$--$3.$ of Definition $\ref{definition_master_field_categorical}$ for the categorical Holonomy field implies the following ones for its distribution:
\begin{proposition}
\begin{enumerate}[\indent 1.]
	\item $\Phi^{\H}_{\ell_{1}\ell_{2}\ell_{1}^{-1}} = \Phi^{\H}_{\ell_{2}},~\Phi^{\H}_{\ell_{1}^{-1}} = \Phi^{\H}_{\ell_{1}}\circ \mathcal{S}~\text{ for all } \ell_{1},\ell_{2} \in {\sf RL}_{0}(\mathbb{R}^{2})$,
	\item $\Phi^{\H}_{\ell_{1}\ell_{2}}=\Phi^{\H}_{\ell_{1}} \star \Phi^{\H}_{\ell_{2}}$, for all simple loops $\ell_{1},\ell_{2} \in {\sf RL}_{0}(\mathbb{R}^{2})$ with disjoint interiors,
	\item $\Phi^{\H}_{\ell} = \Phi^{\sf H}_{T(\ell)}$ for all area-preserving homeomorphisms $T$ of the Euclidian plane $\mathbb{R}^{2}$.
\end{enumerate}
\end{proposition}

\subsection{Categorical Holonomy Field: colimit over lassos}
\label{construction_generalized_master_field}
In this Section, we explain how to construct a Categorical Holonomy Field starting with a certain \emph{family of objects in $\mathcal{B}$ indexed by the set ${\sf Lass}_0(\mathbb{R}^2)$ of anticlockwise oriented lassos drawn on the plane and based at the origin}; that is from a family 
\begin{equation*}
\mathcal{H} = \{ (H_c, \tau_c) \in \mathcal{B},~c \in {\sf Lass}_0(\mathbb{R}^2)\}.
\end{equation*}

It is natural to construct the Holonomy Field ${\sf H}$ as a certain limit since the group of reduced loops drawn on the plane is itself a colimit (or a direct limit, we want to be coherent with the categorical terminology in use so far, we refer the reader to the Annexes for the definition of a colimit of a functor), as we have seen in the previous Section. We should give now an alternative description of ${\sf RL_0}(\mathbb{R}^2)$ as a colimit over finite sequences of lassos. 

We denote by $\mathcal{P}\left({\sf RL}_{0}\left(\mathbb{R}^{2}\right) \right)$ the set of finite sequences of distinct and reduced loops. We denote by $\mathcal{P}\left( {\sf Lass}_0(\mathbb{R}^2) \right)$ the set of finite sequences of lassos in ${\sf Lass}_0(\mathbb{R}^2)$ \emph{with disjoint bulks}. Alternatively, these are free (in the group ${\sf RL_0}(\mathbb{R}^2))$ sequences of lassos, basis of the subgroup of ${\sf RL_0}(\mathbb{R}^2)$ they generate.
We set ${\sf RL}_0\langle L \rangle \subset {\sf RL}_0(\mathbb{R}^{2})$ the subgroup of the group of reduced loops on the plane that are concatenation and reduction of loops in $L$. 

Given an integer $n\geq 1$, we recall that a permutation $\sigma \in \mathfrak{S}_n$ acts on a sequence of loops $(\ell_1,\ldots,\ell_n)$ by 
$$
\sigma \cdot (\ell_1,\ldots,\ell_n)=\ell_{\sigma^{-1}(1)},\ldots,\ell_{\sigma^{-1}(n)}.
$$


For two finite sequences of reduced loops $L$ and $L^{\prime}$, we will write 
$$L \prec L^{\prime} \textrm{~if~} L \subset {\sf RL}_0(L^{\prime}).$$ 

The relation $\prec$ is transitive but not antisymmetric: it is a \emph{preorder}. In fact, for any permutation $\sigma \in \mathfrak{S}_n$, one has $L \prec \sigma\cdot L$ and $\sigma \cdot L \prec L$. For any pair of sequences of loops $L \prec L^{\prime}$, there exists a group morphism $\phi_{L^{\prime},L}$ such that:
\begin{equation*}
	\begin{array}{cccc}
		\phi_{L,L^{\prime}}: & {\sf RL}_0\langle L \rangle & \rightarrow & {\sf RL}_0\langle L^{\prime}\rangle \\
		                    & \ell                       & \mapsto     & \ell
	\end{array}.
\end{equation*}
Note that $(\mathcal{P}(\sf {\sf RL}_0(\mathbb{R}^2)),\prec)$ yields a small category with objects $\mathcal{P}(\sf {\sf RL}_0(\mathbb{R}^2))$ and morphisms $L \prec L^{\prime}$, $L,L^{\prime}\in \mathcal{P}(\sf {\sf RL}_0(\mathbb{R}^2))$. The $\phi$'s yield a covariant functor ${\sf L}: \mathcal{P}\left({\sf L}_{0}(\mathbb{R}^{2}) \right) \to {\rm Grp}$ with values in the category Grp of groups defined by: 
$${\sf L}(L) = {\sf RL}_0\langle L \rangle, \textrm{ and }{\sf L}(L \prec L^{\prime}) = \phi_{L,L^{\prime}}.$$ 
The family of morphisms $\phi_{L, L^{\prime}},~L\prec L^{\prime} \in \mathcal{P}({\sf L}_0(\mathbb{R}^{2}))$ enjoys the following trivial properties:

\begin{enumerate}[$\indent \scriptstyle \bullet$]
	\item Let $L,~L^{\prime}$ be two finite sequences of loops with $L\prec L^{\prime}$, then:
	      \begin{equation*}
		      \phi_{L,L^{\prime}} = \phi_{\alpha \cdot L, \beta \cdot L^{\prime}},~\alpha \in \mathfrak{S}_{|L|},~\beta \in \mathfrak{S}_{|L^{\prime}|},
	      \end{equation*}
	\item with $L_{1} \prec M_{1}$ and $L_{2} \prec M_{2}$ four finite sequences of affine loops, one has:
	      \begin{equation*}
		      \phi_{M_{1},L_{1}}(\ell) = \phi_{M_{2},L_{2}}(\ell),~ \ell \in {\sf R}{\sf L}\langle L_{1}\rangle\cap {\sf R}{\sf L}\langle L_{2}\rangle.
	      \end{equation*}
\end{enumerate}
The set of affine loops drawn on the plane is the colimit of the functor ${\sf L}$ restricted to $\mathcal{P}({\sf Lass}_0(\mathbb{R}^2))$, as it follows from the previous Section:
\begin{equation*}
	{\sf RL}_{0}\left(\mathbb{R}^{2}\right) = \underset{{\longrightarrow}}{\lim}~{\sf L}~|~\mathcal{P}({\sf Lass}_0(\mathbb{R}^2)).
\end{equation*}
where $|$ is used for restriction.
This implies that to build a group morphism ${\sf H}: {\sf RL}_0(\mathbb{R}^2) \to G$ for a certain group $G$, one can start from the data of a group morphism ${\sf H}_L:{\sf RL}_0(L)\to G$ for any sequences of lassos in $\mathcal{P}({\sf RL}_0({\sf Lass}_0(\mathbb{R}^2))$ compatible with the $\phi$'s in the following sense:
$$
{\sf H}_{L^{\prime}} \circ \phi_{L^{\prime},L}= {\sf H}_{L}. 
$$

\begin{assumptions}
\label{assumptionsdeux}
For the entire Section, the algebraic setting is the one introduced at beginning of Section \ref{sec:catholofield}. In addition to the requirements on the three categories $\mathcal{B},~\mathcal{C}$ and $\mathcal{D}$ we made in Assumptions \ref{assumptions:triptic}, we require for the categories $\mathcal{C}$ and $\mathcal{D}$ to be \emph{cocomplete}. This means that $\mathcal{C}$ and $\mathcal{D}$ have colimits: it implies that any functor over a small category with values in $\mathcal{C}$ or in $\mathcal{D}$ (resp. in $\mathcal{D}$) admits a colimit in $\mathcal{C}$. We will only use this property for showing the existence of a colimit of a functor over the small category implied by the preorder $\prec$ on $\mathcal{P}({\sf Lass}_0(\mathbb{R}^2))$. In that case, the term direct limit instead of colimit is used at times. We refer to \cite{adamek2004abstract}, Chapter III, Section 11 for details. It implies that the comma category $\mathcal{B}$ is also cocomplete\footnote{\url{https://ncatlab.org/nlab/show/comma+category\#completeness_and_cocompleteness}}.
\end{assumptions}
First, we focus on the construction of the morphism
$${\sf H}:{\sf RL}_0(\mathbb{R}^{2}) \to \textrm{Hom}_{\mathcal{C}}(H,\mathcal{A})$$ of the Definition \ref{definition_master_field_categorical}, regardless of the properties $1-3$. Thereafter, we will make additional hypothesises on $ \mathcal{H} = (H_c, \tau_c: H \to k), c \in {\sf Lass}_0(\mathbb{R}^2)$ to ensure them. The first step is to construct the object $\mathcal{A} \in \mathcal{C}$ together with a morphism $\tau_{\mathcal{A}}: \mathcal{A}\to k$. 

To do this, from the family $\mathcal{H}$ of objects in $\mathcal{B}$, we build an associated functor $\tilde{\sf A}$ with values in $\mathcal{B}$ over the finite sequences of \emph{lassos} (equipped with the preorder $\prec$).
It exists thanks to Assumptions \ref{assumptionsdeux}.
%
Next, to build a Categorical Holonomy Field, we will start from the data, for every sequence of lassos $L \in \mathcal{P}({\sf Lass}_0(\mathbb{R}^2))$, of a group morphism
$$\H_{L} \in {\rm Hom}_{\mathcal{G}rp}({\sf RL}_0\langle L \rangle, {\rm Hom}_{\mathcal{C}}(H,\mathcal{A}))$$
and assume that
\begin{equation}
	\label{projective}
	{\H}_{L} = {\H}_{L^{\prime}} \circ \phi_{L,L^{\prime}},\quad L \prec L^{\prime}.
\end{equation}

We may conclude the existence of a Categorical Holonomy Field by using the universal property of the colimit of ${\sf L}$ and equation \eqref{projective}; there exists a morphism ${\sf H}$ from ${\sf RL}_0(\mathbb{R}^{2})$ into $\Hom_{\mathcal{C}}(H,\mathcal{A})$ such that the diagram in Fig. \ref{holo} is commutative.
\begin{figure}[!htb]\centering
	\begin{tikzcd}
		{\sf RL}_0(L)\arrow{rr}\arrow{rd} \arrow["{{\sf H}}_{L}",swap]{rdd}& & {\sf RL}_0(L^{\prime})\arrow{ld}\arrow["{{\sf H}}_{L^{\prime}}"]{ldd} \\
		&{\sf RL}_0(\mathbb{R}^{2})\arrow["{\sf H}"]{d}& \\
		&\Hom_{\mathcal{C}}(H,\mathcal{A})&
	\end{tikzcd}
	\caption{\label{holo}\small The Holonomy ${\sf H}$ obtained as a solution of an universal problem.}
\end{figure}
We now enter into the definition of the functor $\tilde{\sf A}:\mathcal{P}({\sf Lass}_0(\mathbb{R}^2)) \to \mathcal{B}$ from the family $\mathcal{H}$. 
Let $\left(c_{1},\ldots,c_{p}\right)$ be a finite sequence of lassos in $\mathcal{P}({\sf Lass}_0(\mathbb{R}^2))$ and define $H_{(c_{1},\ldots,c_{p})} \in \mathcal{B}$ by:
\begin{equation*}
	H_{\left(c_{1},\ldots,c_{p}\right)} = (H ,\tau_{c_{1}}) \otimes \cdots \otimes (H ,\tau_{c_{p}})=(H^{\sqcup\,p},\tau_{c_{1}}\otimes \cdots \otimes \tau_{c_{p}}).
\end{equation*}
\begin{notation}
\label{marginals}
In the following, we use the notation shorter notation:
$$\tau_{(c_{1},\ldots,c_{p})} = \tau_{c_{1}} \otimes \cdots \otimes \tau_{c_{p}}.$$
Given a sequence of loops $(\ell_1,\ldots,\ell_p)$ and an integer $1 \leq i \leq p$, we denote by $\iota_{{i}}^{(\ell_{1},\ldots,\ell_{p})}$ the $i^{th}$ injection in $\mathcal{C}$ from $H$ into the product $H^{\sqcup p}$. 
\end{notation}

We define the group morphism ${\sf H}_{(c_{1},\ldots,c_{p})}: {\sf RL}\langle(c_{1},\ldots,c_{p}) \rangle \to \Hom_{\mathcal{C}}(H, H^{\free p})$ by the following prescription on its values on $(c_1,\ldots,c_p)$:
\begin{equation}
	\label{def_hol}
	\tag{Hol}
	{\sf H}_{(c_{1},\ldots,c_{p})}(c_{i}) = i_{{i}}^{(c_{1},\ldots,c_{p})}.
\end{equation}
Recall that $\Hom_{\mathcal{C}}(H, H^{\free p})$ is endowed with the product
\begin{equation*}
	A \times B = (A \freem B) \circ \Delta,~A,B \in \Hom_{\mathcal{G}rp}(H,H^{\free p}).
\end{equation*}
One has for example $H_{(c_1,\ldots,c_p)}(c_1c_2)=\iota_1^{(c_1,\ldots,c_p)}\times \iota_2^{(c_1,\ldots,c_p)}$. We introduce next \emph{the Holonomy of a sequence of loops in} ${\sf RL}_0\langle (c_1,\ldots,c_p) \rangle $ \emph{respectively to} $(c_1,\ldots,c_p)$. Let $(\ell_{1},\ldots,\ell_{n})$ be a finite sequence of loops in ${\sf RL}\langle (c_{1},\ldots,c_{p})\rangle$, we define the morphism of the category $\mathcal{C}$:
$${\sf H}_{c_{1},\ldots,c_{p}}(\ell_{1},\ldots,\ell_{n}): H^{\free n} \to H^{\free p}$$ by,
\begin{equation}
\label{eqn:mhol}
\tag{mHol}
	{\sf H}_{c_{1},\ldots,c_{p}}(\ell_{1},\ldots,\ell_{n}) = {\sf H}_{c_{1},\ldots,c_{p}}(\ell_{1}) \freem \cdots \freem {\sf H}_{c_{1},\cdots,c_{p}}(\ell_{n}).
\end{equation}
Finally, we set $$\tau^{\ell_{1},\ldots,\ell_{n}}_{c_{1},\ldots,c_{p}} = \tau_{c_{1},\ldots,c_{p}} \circ {\sf H}_{c_{1},\ldots,c_{p}}(\ell_{1},\ldots,\ell_{n}): H^{\sqcup n} \to k.$$
and call this morphism of $\mathcal{D}$ \emph{the distribution of $\ell_1,\ldots,\ell_n$ relatively to $(c_1,\ldots,c_p)$}.
We emphasize the dependence of $\tau^{\ell_{1},\ldots,\ell_{n}}_{c_{1},\ldots,c_{p}}$ toward the sequence of lassos we picked initially.
We would like now to suppress this dependence of the distribution $\tau_{\ell_{1},\ldots,\ell_{n}}^{c_{1},\ldots,c_{p}}$ toward a choice of basis of ${\sf RL}_0((c_1,\ldots,c_p))$.
We need, first, more assumptions about the morphisms $\tau_{c},~c\in {\sf Lass}_0(\mathbb{R}^2)$ and, second, a description of all the basis of lassos ${\sf RL}_0((c_1,\ldots,c_p))$.
Luckily, for the second point, we can rely on a result of Artin \cite{artin1947theory} but we need to introduce \emph{braid groups}. Let $n\geq 1$ an integer. Without going too much into details, the braid group $\mathcal{B}_{n}$ (the braid group on $n$ strands) has the following presentation:
\begin{equation*}
	\mathcal{B}_{n} = \langle \beta_{1},\ldots,\beta_{n-1}\,|\, \beta_{i}\beta_{i+1}\beta_{i}=\beta_{i+1}\beta_{i}\beta_{i+1},~\beta_{i}\beta_{j} = \beta_{j}\beta_{i},~|i-j|\geq 2\rangle.
\end{equation*}
Recall that the group of permutations $\mathfrak{S}_{n}$ admits the same presentation with the additional property $\sigma_{i}^{2} = 1$. Incidentally, to a braid $\beta \in \mathcal{B}_{n}$, we can associate a permutation $\sigma_{\beta}$. It is defined on an elementary braid $\beta_{i}$ by
\begin{equation*}
\label{eq:action}
	\sigma_{\beta_{i}} =  (i,i+1),~ 1 \leq i \leq n-1,
\end{equation*}
and extended as a group morphisms on $\mathcal{B}_{n}$.
The braid group $\mathcal{B}_{n}$ acts on a finite sequence of loops $\ell_{1},\ldots,\ell_{p}$ with $\ell_{i} \in {\sf RL}_{0}(\mathbb{R}^{2})$. The action of the elementary braid $\beta_{i},~ 1 \leq i \leq n-1$ is
\begin{equation*}
	\beta_{i} \cdot (c_{1},\ldots,c_{n}) = (c_{1},\ldots, c_{i+1},c_{i+1}c_{i}c_{i+1}^{-1},\ldots,c_{n}).
\end{equation*}
In fact, the simple computations
\begin{equation*}
	(\beta_{i} \beta_{i+1} \beta_{i}) (c_{1},\ldots,c_{n})= (c_{1},\ldots,c_{i+1},c_{i+2},c_{i+2}c_{i+1}c_{i}(c_{i+2}c_{i+1})^{-1},\ldots,c_{n}) = \beta_{i+1}\beta_{i}\beta_{i+1}.
\end{equation*}
anb $\beta_i\beta_j\cdot (c_1,\ldots,c_n)=\beta_j \beta_i \cdot c_1,\ldots,c_n)$ for $|i-j| > 1$ show that \eqref{eq:action} defines an action of $\mathcal{B}_n$ on ${\sf RL}_0\langle (c_1,\ldots,c_n) \rangle$.
\begin{proposition}[see Proposition 6.8 in \cite{gabriel2015planar}]
	\label{prop:artin}
	Let $(c_{1},\ldots,c_{n})$ and $(c^{\prime}_{1},\ldots,c^{\prime}_{p})$ be two sequences of anticlockwise oriented lassos based at the origin $0$. For each family, we assume that the interiors of the bulks are pairwise distinct and that ${\sf R}{\sf L}\langle c_{1},\ldots,c_{n}\rangle = {\sf RL}\langle c^{\prime}_{1},\ldots,c^{\prime}_{p}\rangle$. Then $n = p$ and there exists a braid $\beta \in \mathcal{B}_{n}$ such that
	\begin{equation*}
		\sigma_{\beta}\cdot(c_{1}^{\prime},\cdots,c_{n}^{\prime}) = \beta \cdot(c_{1},\cdots,c_{n}).
	\end{equation*}
\end{proposition}

\begin{definition}[see Definition 6.10 in \cite{gabriel2015planar}]
	\label{def:braid}
	Let $\mathcal{H} = (H,\tau_{c})_{c\in {\sf Lass_0}(\mathbb{R}^2)}$ a family of objects in $\mathcal{B}$ indexed by lassos.
	\begin{enumerate}[\indent $\scriptstyle \bullet$]
		\item The family $\mathcal{H}$ is said to be \emph{purely invariant by braids} if for any sequence $\left(c_{1},\ldots,c_{n}\right)$ of lassos with disjoint interiors, one has
		      \begin{equation}
			      \label{eq:braiinv}
			      \tag{braid}
			      \tau_{\beta \cdotp c} = \tau_{\sigma_{\beta} \cdot c} \circ {\sf H}_{\sigma_{\beta}\cdotp c}(\beta\,\cdotp c),\quad  \textrm{ for any braid } \beta \in \mathcal{B}_{n}.
		      \end{equation}
		\item The family $\mathcal{H}$ is said \emph{infinitely divisible} if for any pair of lassos $\left(c_{1}, c_{2}\right)$ with disjoint interiors, we have:
		      \begin{equation}
			      \label{eq:infinidiv}
			      \tag{$\infty$-div}
			      \tau_{c_{1}\cdotp c_{2}} = \tau_{c_{1},c_{2}} \circ {\sf H}_{\left(c_{1},c_{2}\right)}\left(c_{1}c_{2}\right) \quad \textrm{if}\quad c_{1}c_{2} \quad \textrm{is a lasso in ${\sf RL}_0(\mathbb{R}^2)$.}
		      \end{equation}
	\end{enumerate}
\end{definition}

Proposition \ref{prop:artin} motivates the first point of Definition \ref{def:braid}. 

\begin{assumptions}
\label{binvinfidiv}
{From now on, we assume the family $\mathcal{H}= (H_{c},\tau_{c})_{c\in {\sf Lass}_0(\mathbb{R}^2)}$ to be \emph{braid-invariant} and \emph{infinitely divisible.}}
\end{assumptions}

\begin{proposition}
	\label{prop:inductivelimit}
	With the notations introduced so far, let $L \subset \mathcal{P}({\sf Lass}_0(\mathbb{R}^2)$ be a finite sequence of loops. Let $(c_{1},\ldots,c_{p})$ and $(c^{\prime}_{1},\ldots,c^{\prime}_{q})$ be two sequences of lassos in $\mathcal{P}({\sf RL}_0(\mathbb{R}^2)$ and assume that $$L \subset {\sf RL}_0\langle(c_{1},\ldots,c_{p}) \rangle \cap {\sf RL}_0\langle(c^{\prime}_{1},\ldots,c^{\prime}_{q})\rangle,$$ then the distribution of $L$ respectively to $(c_1,\ldots,c_p)$ is equal to the distribution of $L$ respectively to $(c^{\prime}_1,\ldots,c^{\prime}_q)$:
    $$\tau_{c_{1},\ldots,c_{p}}^{L} = \tau_{c^{\prime}_{1},\ldots,c^{\prime}_{q}}^{L}.$$
    This implies that if $ (c_1,\ldots,c_p) \prec (c^{\prime}_1,\ldots,c^{\prime}_q)$ is pair of finite sequences of lassos in $\mathcal{P}({\sf Lass}_0(\mathbb{R}^2))$, one has in fact:
    $${\sf H}_{(c^{\prime}_1,\ldots,c^{\prime}_q)}(c_1,\ldots,c_p):(H^{\free p},\tau_{c_{1},\ldots,c_{p}})\to(H^{\free q},\tau_{c^{\prime}_{1},\ldots,c^{\prime}_{q}}).$$
\end{proposition}

\begin{proof}


 Let $L=(\ell_1,\ldots,\ell_n) \in \mathcal{P}({\sf RL}_0(\mathbb{R}^2))$ be a sequence of loops. Let $f_{1},\ldots,f_{q}$ be an enumeration of the faces of $\mathbb{G}_{L}$ and pick a sequence of anticlockwise affine lassos $(c_{1},\ldots,c_{p})$ with $c_{i}$ surrounding the face $f_{i}$. First, since we assumed that the monoidal product is symmetric, $\tau_{c_{1},\ldots,c_{p}}$ does not depend on the enumeration of the faces we choose.
	Then, if $\beta$ is a braid in $\mathcal{B}_{p}$,
	\begin{equation*}
		\begin{split}
			\tau^{L}_{\beta \cdot (c_{1},\ldots,c_{p})} &= \tau_{\beta \cdot (c_{1},\ldots,c_{p})} \circ {\sf H}_{\beta \cdot (c_{1},\ldots,c_{p})}(\ell_{1},\ldots,\ell_{q})\\ &= \tau_{c_{\beta(1)},\ldots c_{\beta(q)}}\circ {\sf H}_{\sigma_{\beta}\cdot c}(\beta \cdot c)\circ {\sf H}_{\beta \cdot c}(\ell_{1},\ldots,\ell_{p}) = \tau_{c_{\beta(1)},\ldots,c_{\beta(n)}}\circ {\sf H}_{\sigma_{\beta}\cdot c}(\ell_{1},\ldots,\ell_{q})\\
			&= \tau_{c_{\beta(1)},\ldots,c_{\beta(n)}}^{L} = \tau_{c_{\beta(1)},\ldots,c_{\beta(n)}}^{L} = \tau_{c_{1},\ldots,c_{n}}^{L}.
		\end{split}
	\end{equation*}
	In conclusion, $\tau_{L}^{c_{1},\dots, c_{p}}$ does not depend on the basis of lassos $(c_{1},\ldots,c_{p})$ we choose for ${\sf RL}(\mathbb{G}_{L})$.

	Let $C^{\prime} = (c^{\prime}_{1},\ldots,c^{\prime}_{p})$ be a finite sequence of affine lassos such that $L\subset {\sf RL}\langle (c^{\prime}_{1},\ldots,c^{\prime}_{p})\rangle$. The graph $\mathbb{G}_{C^{\prime}}$ is finer than the graph $\mathbb{G}_{L}$ and can thus be obtained by iterative application of two transformations, starting from the graph $\mathbb{G}_{L}$:
	\begin{enumerate}[\indent\indent 1.]
		\item \label{trun} adding a vertex on an edge,
		\item \label{trdeux}connecting two vertices (they can be equal).
	\end{enumerate}

	\par Let $\mathbb{G}_{L} \prec \mathbb{G}_{1} \prec \cdots \prec \mathbb{G}_{n} \prec \mathbb{G}_{C^{\prime}}$ be a sequence of graphs obtained by successive applications of the transformations \ref{trun} and \ref{trdeux}: $\mathbb{G}_{i+1}$ is obtained from $\mathbb{G}_{i}$ by one of the transformation \ref{trun}, \ref{trdeux}.

	Next, we define inductively a sequence of objects $H^{i},~{0\leq i \leq n+1}$ in the category $\mathcal{B}$.

	We first define a sequence $(c^{(1)},\ldots,c^{(n+1)})$ with $c^{(i)} \in \seqlasso$ such that for each integer $1 \leq i\leq n$, $c^{(i)}$ is a basis of ${\sf RL}\langle \mathbb{G}_{i}\rangle$. Put $c^{(n+1)} = (c^{\prime}_{1},\ldots,c^{\prime}_{p})$. If $\mathbb{G}_{i+1}$ is obtained by adding a vertex to $\mathbb{G}_{i}$, the groups of loops drawn on $\mathbb{G}^{i+1}$ is equal to the group of loops drawn on $\mathbb{G}^{(i)}$. In that case we set $c^{i} = c^{i+1}$.

	On the contrary, assume that two vertices of $\mathbb{G}_{i}$ are connected by an affine path $e$ to obtain $\mathbb{G}_{i+1}$. If a bounded face $f$ of $\mathbb{G}_{i}$ is cut into two faces $f_{1}$ and $f_{2}$, we obtain a basis $c^{(i)}$ of ${\sf RL}(\mathbb{G}_{i})$ as follows. First, we pick the lassos $b_{1},\ldots,b_{q}$ that do not surround the faces $f_{1}$ and $f_{2}$. It may happen that the tail of $b_{i}$ contains $e$, in that case we can certainly replace this lasso by another one surrounding the same face but with a tail avoiding $e$. Then we take the product of the two lassos surrounding the faces $f_{1}$, $f_{2}$ to obtain a lasso $c_{f}$. If the added edge is on the boundary of the unique unbounded face of $\mathbb{G}_{i+1}$ then we simply remove the lasso that surround the bounded face created by adding the edge $e$.

	We eventually pick any enumeration of the set of lassos we obtain this way. See Fig. \ref{fig:exdecomposition} for an example of these elementary transformations.

	\begin{figure}[!htb]\centering
		\scalebox{0.7}{
		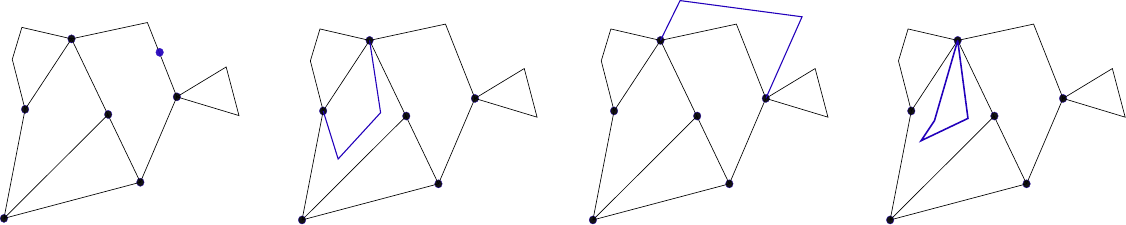
		}
		\caption{\label{fig:exdecomposition} Example of elementary transformations of a graph. The added edge $e$ is coloured in violet, as well as the added vertex.}
	\end{figure}
	Next, we define inductively, the sequence of objects in $\mathcal{B}$ by the equations:
	\begin{equation*}
		\begin{split}
			H^{(0)} = (H^{\free |L|}, \tau_{L}),~
			H^{i} = \left(H^{\free |L|},~\tau_{c^{(i)}} \circ {\sf H}_{c^{(i)}}\left(\ell_{1},\ldots,\ell_{p}\right)\right) = (H^{|L|}, \tau_{i}), 1 \leq i \leq n
		\end{split}
	\end{equation*}
	and set $\tau^{(n+1)}=\tau^{L}_{c^{\prime}_{1},\ldots,c_{p}^{\prime}}$.
	Let $i \leq n$ an integer such that $\mathbb{G}^{(i+1)}$ is obtained by cutting a face of $\mathbb{G}^{(i)}$ in two. We claim that the diagram in Fig. \ref{comtracepreserv} is a commutative diagram of morphisms in $\mathcal{C}$. In Fig. \ref{comtracepreserv}, the blue arrows are morphisms of $\mathcal{B}$. The upper arrow in Fig. \ref{comtracepreserv} is thus painted in blue owing to the infinite divisibility of the morphisms $\tau_{c},~c \in {\sf Lass}_0(\mathbb{R}^2)$.

	\begin{figure}[!htb]\centering
		\begin{tikzcd}
			(H_{c^{(i)}},E_{c^{(i)}})\arrow["{\sf H}_{c^{(i+1)}}(c^{(i)})",cyan]{rr} & & (H_{c^{(i+1)}},E_{c^{(i+1)}})\\
			H^{(i)}\arrow["\id"]{rr}\arrow[cyan]{u}{{\sf H}_{c^{(i)}}(l_{1},\ldots,l_{p})}& &H^{(i+1)}
			\arrow[cyan,swap]{u}{{\sf H}_{c^{(i+1)}}(\ell_{1},\ldots,\ell_{p})}
		\end{tikzcd}
		\caption{\label{comtracepreserv}\small Compatibility of the holonomies above the loops $\ell_{1},\ldots,\ell_{p}$ computed in the basis $c_{i}$ and $c_{i+1}$.}
	\end{figure}

	From Fig. \ref{comtracepreserv}, the sequence of morphisms $\left(\tau_{i}, 1 \leq i\leq n\right)$ is a constant sequence which leads readily to the conclusion since $\tau^{L}_{(c_{1},\ldots,c_{p})}= \tau_{0} = \tau_{n+1} = \tau^{L}_{(c^{\prime}_{1},\ldots,c^{\prime}_{p})}$.
\end{proof}
Following \ref{prop:inductivelimit}, given a sequence of loops $L$ we denote by $\tau^{L}$ the distribution of $(\ell_1,\ldots,\ell_p)$ respectively to any sequence of lassos generating the group of reduced loops ${\sf RL}_0(\mathbb{G}_L)$.
Proposition \ref{prop:inductivelimit} then implies that the functor ${\sf \tilde{A}}$ defined for any sequences of lassos $~ L \prec L^{\prime} \in \mathcal{P}({\sf Lass}_0(\mathbb{R}^2))$ by, 
\begin{equation*}
	{\sf \tilde{A}}(L) = (H^{|L|}, \tau^{L}),~ {\sf \tilde{A}}(L,L^{\prime}) = {\sf H}_{L^{\prime}}(L),
\end{equation*}
is well-defined.
We denote by $\mathcal{A}=(A,\tau_{\mathcal{A}})$ the colimit of ${\sf A}$ in $\mathcal{B}$. Finally, to define ${\sf H}_{L}$ with $L = (c_1,\ldots,c_p) \in \mathcal{P}({\sf Lass}_0(\mathbb{R}^2)$, we simply take the product in the group of homomorphism $(\Hom_{\mathcal{C}}(H,\mathcal{A}), \times)$ of the Holonomies $c_i$ that is of the canonical injections $\iota_{i}^{(c_1,\ldots,c_p)}$ in the order prescribed by the sequence $(c_1,\ldots,c_p)$. By construction, the property \ref{independence} holds.

We recall that the right comodule $(H,\Omega_c)$ where $\Omega_c$ the conjugacy coaction, in the category rCoMod$\mathcal{C}(H)$ has been defined in Section \ref{sec:ZhangcoMod}, equation \eqref{eqn:gaugecoaction}.
\begin{assumptions}
Assume further that each morphism $\tau_{c},~c\in{\sf Lass}_0(\mathbb{R}^2)$ is \emph{gauge-invariant}:
\begin{equation}
	\label{eqn:gaugeinvariant}
	\tag{G-inv}
	\tau_{c}\circ (({\sf K} \freem \id_{H}) \circ \Omega_{c}) = \tau_{c},~ c\in{\sf Lassos}.
\end{equation}
for each morphism $K: H \to H_{c}$ such that $\hat{K}$ (see Definition \ref{def:morphism}) is independent from $\id_{H_{c}}$.
\end{assumptions}

\begin{lemma}
\label{lemma:lemmaomegabar}
 With the notations introduced so far, the colimit $\mathcal{A}$ can be endowed with a co-action $\overline{\Omega}_{c}$ of the $H$-algebra $H$ that makes the diagram in Fig. \ref{diag:co-action} commutative.

\begin{figure}[!htb]\centering
	\begin{tikzcd}
		A \arrow{r}{\overline{\Omega}_{c}}&H \free A\\
		H^{\free |L|}\arrow{u}{j_{H_{L}}^{A}} \arrow{r}{\Omega_{c}^{|L|}} & H \free H^{\free |L|} \arrow{u}[swap]{\id_{H}\free j_{H_{L}}^{A}}
	\end{tikzcd}
	\caption{\label{diag:co-action}\small The canonical injection $j^{A}$ are comodule morphisms if the direct limit $A$ is endowed with $\bar{\Omega}_{c}$.}
\end{figure}
\end{lemma}
\begin{proof}
First, notice that for any pair of sequences of loops $L \prec L^{\prime}$, the Holonomy ${\sf H}_{L^{\prime}}(L)$ is gauge covariant,
\begin{equation}
	\label{eqn:gaugelassosun}
\Omega_{c}^{|L^{\prime}|}\circ {\sf H}_{L^{\prime}}(L)=(\id_{H} \freem {\sf H}_{L^{\prime}}(L))\circ \Omega_{c}^{L}.
\end{equation}
Hence, the co-action $\overline{\Omega}_{c}:A\to H \free A$, defined by
\begin{equation*}
	\overline{\Omega}_{c}([X_{L}]) = ((\id_{H}\free j_{H_{L}}^{\mathcal{A}})\circ \Omega_{c})(X_{L}),~ [X_{L}] \in A
\end{equation*}
is well defined, since:
\begin{equation*}
	\begin{split}
		((\id_{H} \free j_{H_{L^{\prime}}}^{\mathcal{A}})\circ \Omega_{c}^{L^{\prime}}\circ {\sf H}_{L^{\prime}}(L))(X_{L})&=(\id_{H}\free( j_{H_{L^{\prime}}}\circ{\sf H}_{L^{\prime}}(L)))\circ \Omega_{c}^{L})(X_{L})\\ &=((\id_{H}\free\iota_{H_{L}})\circ\Omega_{c}^{L})(X_{L}).
	\end{split}
\end{equation*}
\end{proof}
The property $\ref{gaugeinv}$ is satisfied by every morphism $\tau_{c_{1},\ldots,c_{p}}$, $(c_{1},\ldots,c_{p}) \in \mathcal{P}({\sf Lass}_0(\mathbb{R}^2))$. In fact, let $\tau_{H}:H \to k$ be a morphism of the category $\mathcal{D}$. For all integer $1 \leq k \leq p$, owing to equations \eqref{eqn:gaugeinvariant} and \eqref{eqn:gaugelassosun},
\begin{equation}
	\label{gauginvlasso}
	(\tau_{H} \otimes (\tau_{c_{1}} \otimes \cdots \otimes \tau_{c_{p}})) \circ (\id_{H} \free \iota^{(c_1,\ldots,c_p)}_{k}) \circ \Omega_{c} = \tau_{c_{k}}.
\end{equation}
Since the morphisms $\iota^{(c_1,\ldots,c_p)}_{k}$ are mutually independent, the morphisms $\iota^{(c_1,\ldots,c_p)}_{{i}} \circ \Omega_{c_{i}}: H_{c_{i}} \to (H,\tau_{H}) \otimes H_{c_{1}} \otimes \cdots \otimes H_{c_{p}}$ are also mutually independent.
Hence, owing to equation \eqref{gauginvlasso},
\begin{equation*}
	\Omega_{c}^{(n)} = \iota^{(c_1,\ldots,c_p)}_{{1}} \circ \Omega_{{1}} \freem \cdots \freem \iota^{(c_1,\ldots,c_p)}_{{p}}\circ \Omega_{{p}} \in \Hom_{\mathcal{B}}(H_{c_{1}} \otimes \cdots \otimes H_{c_{p}}, (H,\tau_{H}) \otimes H_{c_{1}} \cdots \otimes H_{c_{p}}).
\end{equation*}

Gauge invariance of $\tau_{c_{1},\ldots,c_{p}}$ implies gauge-invariance of the morphisms $\tau_{\ell_{1},\ldots,\ell_{p}}$, which in turn implies gauge invariance of their direct limit, $\phi_{\mathcal{A}}$.
\par Invariance by area-preserving homomorphisms of $\mathbb{R}^{2}$ is implied by invariance by area-preserving homomorphisms of each morphism $\tau_{c},~c\in{\sf Lass}_0(\mathbb{R}^2)$: for each area-preserving diffeomorphism $T:\mathbb{R}^2\to\mathbb{R}^2$, we have the equality $\tau_{c}=\tau_{T(c)}$.
Let us collect in the following Theorem the main outcome of this section.
\begin{theorem}
	For each lasso $c \in {\sf Lass}_0(\mathbb{R}^2)$, pick $\tau_{c} : H \rightarrow k$ a morphism of $\mathcal{D}$. Assume that the family $\{\tau_{c},~c \in {\sf Lass}_0(\mathbb{R}^2)\}$ is
	\begin{enumerate}[\indent $\sbt$]
		\item purely braid invariant, see equation \eqref{eq:braiinv},
		\item infinitely divisible, see equation \eqref{eq:infinidiv}
		\item gauge-invariance, see equation  \eqref{eqn:gaugeinvariant},
		\item invariance by homomorphisms preserving the area: for each lasso $c\in {\sf Lass}_0(\mathbb{R}^2)$ and for each area-preserving diffeomorphism ~$T : \mathbb{R}^2 \rightarrow \mathbb{R}^{2}$ 
		\begin{equation*}
			\tau_{c} = \tau_{T(c)} 
		\end{equation*}
		Then there exists a categorical Holonomy field $\H$ in the sense of Definition \ref{definition_master_field_categorical} such that for each anticlockwise lasso, $\tau_{\mathcal{A}} \circ \H(c) = \tau_{c}$.
	\end{enumerate}
\end{theorem}

\par In the next section, we expose how to obtain a family $\{\tau_{c}: H\to k,~c\in{\sf Lass}_0(\mathbb{R}^2)\}$ that is braid-invariant, infinitely divisible, gauge-invariant and invariant by area-preserving homomorphisms starting from a quantum L\'evy process on a $H$-algebra $H$.

\subsection{Categorical Holonomy Field: construction from a Quantum L\'evy process} 
\label{sec:constructionlevyproc}

In this section, we work in the algebraic settings we progressively put in place in the previous section. Let $(\mathcal{C},\free,k)$ be an \emph{algebraic category}. Let $F:\mathcal{C}\to\mathcal{D}$ be a \emph{wide} and \emph{faithful} functor from $\mathcal{C}$ to $\mathcal{D}$. Set $F(k)=\star$.
	We assume that $\mathcal{C}$ and $\mathcal{D}$ are \emph{cocomplete}. Set $\mathcal{B} = \mathcal{C} \, \downarrow \, \star$ and pick a symmetric monoidal product $\otimes$ on $\mathcal{B}$ such that \emph{${\sf Domain}$ is a monoidal functor};
\begin{equation*}
	\mathcal{C} \overset{F}{\rightarrow} \mathcal{D} \leftarrow \star,\quad \mathcal{B} = \mathcal{C}\,\downarrow \, \star.
\end{equation*}

Let $(H,\Delta,\varepsilon,\mathcal{S})$ be a $H$-algebra in $\mathcal{C}$ and $\mathcal{A}=(A,\tau_{\mathcal{A}})$ be an object in the comma category $\mathcal{B}$. We recall the definition of a quantum L\'evy processes on $H$. 
\begin{definition}[L\'evy process \cite{franz2004theory}]
	\label{Levy}
	 For any time $s > 0$, let $j_{s} : H \rightarrow \mathcal{A}$ be a morphism of $\mathcal{C}$. Set, for any pair of times $0 < s < t$, $j_{s,t} = j_t \freem j_s \circ \Delta$. We say that $j=(j_{s})_{s>0}$ is a \emph{quantum L\'evy process on $H$} if
	\begin{enumerate}[\indent 1.]
		\item (Increments 1) \label{increment1} for all triple of times $u < s <t$, $j_{u,s}\freem j_{s,t} = j_{u,t}$ ,
		\item (Increments 2) \label{increment2}for all time $s\geq 0$ and $b \in B$, $j_{s,s}(b) = \varepsilon(b)$, 
		\item \text{(Independence)} \label{increment3}for any tuple $(s_{1} < t_{1} \leq s_{2} < t_{2} \ldots \leq s_{p} < t_{p})$,
		      \begin{equation*}
			      \tau \circ j_{s_{1},t_{1}} \freem \cdots \freem j_{s_{p},t_{p}} = \tau \circ j_{s_{1},t_{1}} \otimes \cdots \otimes \tau \circ j_{s_{p},t_{p}} \quad ,
		      \end{equation*}
		\item (Stationnarity) \label{increment4} $\tau \circ j_{s,t} = \tau \circ j_{t-s}$ 
	\end{enumerate}
\end{definition}
\begin{remarque}
In case we consider the category $\Alg^{\star}(B)$ of involutive bimodule algebras over an \emph{Banach algebra} $B$ and $\mathcal{B}$ the category \Prob(B) of non-commutative probability spaces (that is the usual settings of non-commutative probability theory), since the initial object is $B$, we require also
\begin{equation*}
	\label{continuity}
	\lim_{t \to s^{+}} \tau \circ j_{s,t} = \varepsilon.
\end{equation*}
\end{remarque}

We let $j = (j_{s,t})_{s \leq t}$, $j_{s,t}:H \to (\mathcal{A},\tau_{\mathcal{A}})$, $s \leq t$ be a L\'evy process on the $H$-algebra $H$ taking values in an object $(\mathcal{A}, \tau_{\mathcal{A}})$ of $\mathcal{B}$. Let $c \in {\sf Lass}_0(\mathbb{R}^2)$ be a lasso drawn on the plane and denote by $|c|$ the area enclosed by the bulk of $c$. For each anticlockwise oriented lasso, we define the object $H_{c}$ in the category $\mathcal{B}$ by 
\begin{equation*}
H_{c}=(H,\tau_{c}) = (H, \tau_{\mathcal{A}} \circ j_{|c|}).
\end{equation*}

\begin{definition}
     Let $j$ be a quantum L\'evy process on a $H$-algebra $H$,
	\begin{enumerate}[\indent 1.]
		\item we say that $j$ is \emph{braid invariant} if for all integer $n\geq1$ and braid $\beta \in \mathcal{B}_{n}$, $\beta\cdot (j_{s_{1},t_{1}},\ldots,j_{s_{n},t_{n}})$ has the same distribution as $\sigma_{\beta} \cdot (j_{s_{1},t_{1}},\ldots,j_{s_{n},t_{n}})$ for tuples of times $s_{1} < t_{1} \leq s_{2} < t_{2} \leq \ldots \leq s_{n} < t_{n}$,
		\item we say that $j$ is \emph{gauge-invariant} if for all pair of times $0 < s,t$, $(\tau_{\mathcal{A}} \otimes \tau_{H}) \circ (j_{s,t} \free \id_H) \circ \Omega_{c}$.
	\end{enumerate}
\end{definition}
\begin{theorem}
	\label{maintheoremzhanghol}
 In the algebraic setting recalled at the beginning of this Section, let $j=(j_{s})_{s > 0}$ be a braid and gauge-invariant Lévy process.
	There exists a Categorical Holonomy Field ${\sf H }$, in the sense of Definition $\ref{definition_master_field_categorical}$ satisfying the following property. For any one-parameter family of growing simple loops $\gamma = (\gamma_{t})_{t\geq 0}$ based at $0 \in \mathbb{R}^{2}$, surrounding a domain Int($\gamma_{t}$) with
	\begin{enumerate}[\indent 1.]
		\item for all time $t \geq 0$, $|\mathrm{Int}(\gamma_{t})| = t$,
		\item for all times $s \leq t$, $\mathrm{Int}(\gamma_{s}) \subset \mathrm{Int}(\gamma_{t})$.
	\end{enumerate}
	The process $(\H(\gamma_{t}))_{t\geq 0}$ has the same non-commutative distribution as the initial L\'evy process $j$.
\end{theorem}

\begin{proof}
    Infinite divisible property of the family $H_{c},~c \in {\sf Lassos}$ is implied by the fact that $j$ is a L\'evy process. In fact, let $c_{1}$ and $c_{2}$ two lassos such that $c_{1}c_{2}$ is also a lasso. In that case, the area enclosed by the bulk of $c$ is the sum of the two areas enclosed by $c_{1}$ and by $c_{2}$. Hence, 
    \begin{align*}
    \tau_{c_{1}c_{2}} &= \tau_{\mathcal{A}} \circ j_{|c_{1}|+|c_{2}|} \\
    &= \tau_{\mathcal{A}} \circ (j_{0, c_{1}}\times j_{|c_{1}|, |c_{1}| + |c_{2}|}) & \textrm{(By definition of the increments)} \\ 
    &= (\tau_{c_{1}} \otimes \tau_{c_{2}}) \circ (\iota_{c_{1}}\times \iota_{c_{2}})& (\otimes-\textrm{independence of the two increments $j_{0,|c_{1}|}$ and $j_{|c_{1}|,|c_{1}|+|c_{2}|}$})
    \end{align*}
    The braid- and gauge-invariance of $\mathcal{H}$ is downwardly implied by the braid- and gauge-invariance of $j$.
    So does invariance by area-preserving homomorphisms.
\end{proof}
\section{Examples of Quantum Holonomy Fields}
\label{examples}
We let $\mathbb{K}$ be one of the three algebras with division $\mathbb{C},\mathbb{R}$ and the quaternions $\mathbb{H}$ and pick an integer $N\geq 1$. We denote by $\mathcal{M}_{N}(\mathbb{K})$ the algebra of $N\times N$ matrice with entries in $\mathbb{K}$.
This Section is devoted to examples. In Section \ref{yangmillsfields}, we apply the method exposed in the previous sections to build Quantum Holonomy Fields over the $H$ algebra of polynomials functions $\mathcal{F}(\mathbb{U}(N,\mathbb{K}))$ (the structural morphisms will be introduced in due time) on the unitary group $\mathbb{U}(N,\mathbb{K})$. We recover a class of \emph{generalized (free and classical) master fields} introduced by the author in \cite{cebron2017generalized}. In this case, thanks to the Fubini theorem, the construction of a Quantum Holonomy Field can be performed from any gauge-invariant Lévy process (braid-invariance is implied by gauge-invariance, which is \emph{not} true in full generality).

In Section \ref{higherdimmasterfield}, we build higher dimensional Master Fields: Quantum Holonomy Fields over the dual Voiculescu groups $\mathcal{O}\langle n \rangle$, $n\geq 1$, obtained from a braid-invariant free Lévy process. Finally, in Section \ref{sec:amalgamated}, we provide an example of a Quantum Holonomy Fields over a $H$-algebra in the category ${\rm Prob}(\C^d)$, the category of amalgamated probability spaces over $\mathbb{C}^d$, seen as a commutative algebra. 
We begin by introducing the (classical) stochastic process at the basis of all the Quantum Holonomy Fields defined in this Section: a Brownian diffusion over $\mathbb{U}(N,\mathbb{K})$.
We refer the reader to the first part of \cite{nico2} for further details.
We denote by $({\sf i},{\sf j},{\sf k})$ the linear real basis of $\mathbb{H}$:
$$
{\sf i}^{2} = {\sf j}^{2} = {\sf k}^{2} = -1, \quad {\sf ij}={\sf k}, {\sf jk}={\sf i}, {\sf ki} = {\sf j}.
$$
The adjoint of an element $x \in \mathbb{K}$ is denoted $x^{\star}$; $\star : \mathbb{K}\to \mathbb{K}$ is a $\mathbb{R}$-linear involution such that:  
$$
{\sf i}^{\star} = -{\sf i},~{\sf j}^{\star} = -{\sf j},~{\sf k}^{\star} = -{\sf k}.
$$
The adjoint $M^{\star}$ of a matrix $M = \left(M_{ij}\right)_{1 \leq i,j\leq N}\in \mathcal{M}_{N}(\mathbb{K})$ is defined by:$$M^{\star} = \left(M^{\star}_{ij}\right)_{ 1 \leq i,j \leq N} = \left( M^{\star}_{ji}\right)_{1 \leq i,j \leq N}.$$
The group of unitary matrices with entries in $\mathbb{K}$ is the connected subgroup of $\mathcal{M}_{N}(\mathbb{K})$ defined by
\begin{equation*}
	\mathbb{U}(N,\mathbb{K}) = \{M \in \mathcal{M}_{N}\left(\mathbb{K}\right),~MM^{\star} = M^{\star}M = 1 \}^{0}.
\end{equation*}
where the exponent ${ }^{0}$ means that we take the connected component of the identity (it is needed for the real case). If $K = \mathbb{R}$ the group $\mathbb{U}(N,\mathbb{R})^0$ is the group of special orthogonal matrices $SO(N,\mathbb{R})$ and for $K = \mathbb{C}$, $\mathbb{U}(N,\mathbb{C})^0$ is the whole group of unitary matrices with complex entries. The Lie algebra $\mathfrak{u}(N,\mathbb{K})$ is given by
\begin{equation*}
	\mathfrak{u}(N,\mathbb{K}) = \{H \in \mathcal{M}_{N}(\mathbb{K}) : H^{\star} + H = 0 \}.
\end{equation*}
The real Lie algebra  of skew-symmetric matrices of size $N\times N$ is denoted $a_{N}$ and the vector space of symmetric matrices of size $N\times N$ is denoted $s_{N}$. As real Lie algebras, one has the direct-sum decompositions:
\begin{equation}
	\label{liealgebras}
	\mathfrak{so}_{N} = \mathfrak{a}_{N}, \quad \mathfrak{u}_{N} = \mathfrak{a}_{N} \oplus {\sf i}\mathfrak{s}_{N}, \quad \mathfrak{sp}_{N} = \mathfrak{a}_{N} \oplus {\sf i}\mathfrak{s}_{N} + {\sf j} \mathfrak{s}_{N} \oplus {\sf k}\mathfrak{s}_{N},~ N \geq 1.
\end{equation}
It follows that, with $\beta = \textrm{dim}_{\mathbb{R}}(\mathbb{K})$,
$$
\textrm{dim}(\mathfrak{u}(N,\mathbb{K})) = \frac{N(N-1)}{2} + (\beta - 1) \frac{N(N + 1)}{2}, N \geq 1.
$$
To define a Brownian motion on the group $\mathbb{U}(N,\mathbb{K})$ one needs to pick first a scalar product on the Lie algebra $\mathfrak{u}\left(N,\mathbb{K} \right)$ invariant by the Adjoint action of $U(N,\mathbb{K})$ over its Lie algebra. 
Notice that since $U\left(N,\mathbb{K}\right)$ is semi-simple, its Killing form is non-degenerate. Besides, as the group $\mathbb{U}(N,\mathbb{K})$ is compact, the negative of the Killing form is an invariant scalar product. Since we are going to let the dimension $N$ tends to $+\infty$, we care about the normalization of the Killing form. Let $\langle \cdot, \cdot \rangle_{N}$ be the scalar product
\begin{equation*}
	\langle X,Y \rangle_{N} = \frac{\beta N }{2}\mathcal{R}e({\sf Tr}(X^{\star}Y)),~ X,Y \in \mathfrak{u}\left(N,\mathbb{K}\right).
\end{equation*}
The direct sums in the equations \eqref{liealgebras} are decompositions into mutually orthogonal summands for $\langle \cdot, \cdot \rangle_{N}$. Let $\{H^{N}_{k}\}$ be an orthonormal basis of $\mathfrak{u}\left(N,\mathbb{K}\right)$, \emph{the Casimir element} $C_{\mathfrak{u}(N,\mathbb{K})}$ is a bivector in the real two-fold tensor product $ \mathfrak{u}(N,\mathbb{K}) \otimes_{\mathbb{R}} \mathfrak{u}(N,\mathbb{K})$ defined by the formula:
\begin{equation*}
	C_{\mathfrak{u}(N,\mathbb{K})} = \sum_{k = 1}^{\beta} H_{k} \otimes H_{k}.
\end{equation*}
Set $c_N^{\mathbb{K}} = m_{\mathcal{M}(N,\mathbb{K})}(C_{\mathfrak{u}(N,\mathbb{K})}$ where $m_{\mathcal{M}(N,\mathbb{K})}:\mathcal{M}(N,\mathbb{K})^{\otimes 2}\to \mathcal{M}(N,\mathbb{K})$ is the matrix multiplication.
\emph{The unitary Brownian motion} $\munitaryfd = (\munitaryfd(t))_{t\geq 0}$ is the stochastic process which is the strong solution of the following classical stochastic differential equation:
\begin{equation*}
	\begin{split}
		&\textrm{d}\munitaryfd(t) = \textrm{d}{\sf W}_{N}^{\mathbb{K}}(t)\munitaryfd(t) + \frac{1}{2}c^{\mathbb{K}}_{N}\munitaryfd(t) \textrm{dt}\\
		&\munitaryfd(0) = I_{N},
	\end{split}
\end{equation*}
with ${\sf W}_{N}^{\mathbb{K}}$ a standard Brownian motion $\mathfrak{u}_N(\mathbb{K})$ with respect to the scalar product $\langle -,-\rangle_N$ and $I_N$ the identity matrix in $\mathcal{M}(N,\mathbb{K})$. The entries of ${\sf W}_{N}^{\mathbb{K}}$ are, up to symmetries, independent Brownian motions whose variance scales as the inverse of the dimension $N$. 
We refer the reader to \cite{levy2008two} for the proof of the following Lemma.
\begin{lemma}
    The matricial stochastic process $\munitaryfd$ exists for any time $t\geq 0$, is almost surely valued in the compact group $\mathbb{U}(N,(\mathbb{K})$ and is gauge-invariant;
    $$
    UU_N^{\mathbb{K}}U^{-1} \overset{distr}{=} U_N^{\mathbb{K}}
    $$
    with $U$ a Haar distributed unitary matrix.
\end{lemma}

This diffusion yields three quantum processes that are of interest for the present work and are defined in \cite{nico2}. Of these three quantum processes, one is a Quantum Lévy process. They are introduced in the next Section.

\subsection{Classical Yang-Mills fields}
\label{yangmillsfields}
 Recall that we denote by $\F(\U(N,\K))$ the algebra of complex polynomial functions on $\U(N,\K)$ (polynomials functions in the coefficients of the matrix). First, we define a \emph{classical L\'evy process $j_{N}$} over $\F(\U(N,\K))$, (the increments are tensor independent), by setting for all time $s\geq 0$:
\begin{equation*}
	\begin{array}{cccc}
		j_{N}^{\mathbb{K}}(s): & \mathcal{F}(\mathbb{U}(N,\mathbb{K})) & \to     & \left(L^{\infty}(\Omega,\mathcal{F},\mathbb{P}),\mathbb{E}\right) \\
		                       & f                                     & \mapsto & f(U_{N}^{\mathbb{K}}(s)).
	\end{array}
\end{equation*}
We recall that $\mathcal{F}(U(N,\mathbb{K}))$ is an involutive $H$-algebra, being a commutative Hopf algebra with structure morphisms:
\begin{equation*}
	\Delta(f)(U,V) = f(UV),~\mathcal{S}(f)(U)=f(U^{-1}),~\varepsilon(f)=f(I_{N}),~\star(f)=\bar{f}.
\end{equation*}
The law of $j^{\mathbb{K}}_{N}$ is invariant by conjugation by any unitary matrices in $\U(N,\K)$ since this property holds for the driving noise ${\sf W}_{N}^{\mathbb{K}}$. To prove braid-invariance for $j_{N}^{\mathbb{K}}$, it is sufficient to prove:
\begin{equation*}
	(j_{N}^{\K}(t)\times j_{N}^{\K}(s,t)\times \left[j_{N}^{\K}\right]^{-1}(t),~j^{\mathbb{K}}_{N}(t))\overset{\text{distrib.}}{=}(j_{N}^{\K}(s,t),j_{N}^{\K}(t)).
\end{equation*}
This last equation is readily implied by gauge-invariance and independence of the increments through Fubini's Theorem. We can therefore apply Theorem \ref{maintheoremzhanghol} to obtain a Holonomy field associated with $j^{\mathbb{K}}_{N}$. This field is the $\mathbb{U}(N)$-\emph{Yang-Mills field} on the plane with structure group $U(N,\mathbb{K})$; we denote it by $\Phi^{\mathbb{K}}_N$.

There are two other gauge- and braid-invariant processes associated with the unitary diffusion $\munitaryfd$. The first of these quantum processes depends on two integers $n,d \geq 1$, it extracts $d\times d$ square blocks from the matricial process $\munitaryfd$ with $N=nd$ and is defined by, for any time $t\geq 0$:
\begin{equation}
	\label{def:processrectextraction}
	\begin{array}{cccc}
		U^{\mathbb{K}}_{n,d}: & \mathcal{O}\langle n \rangle & \to     & \left(\mathcal{M}_{d}(L^{\infty}(\Omega,\mathcal{F},\mathbb{P})),~\mathbb{E}\otimes(\frac{1}{d}{\sf Tr})\right) \\
		                      & u_{ij}                       & \mapsto & \munitaryfd(i,j).
	\end{array}
\end{equation}
where for $nd \times nd$ matrix $A$ and integers $1 \leq i,j \leq n$, $A(i,j)$ is the $d\times d$ sub-matrix at position $(i,j)$ in $A$. Notice that, $U^{\mathbb{K}}_{n,d}$ is not a Lévy process over $\mathcal{O}\langle n \rangle$. However, as it is proved in \cite{nico2} \emph{it converges as $d\to +\infty$ and $n$ is maintained constant toward a free Lévy process over $\mathcal{O}\langle n \rangle$}. In the next Section we define this process as the solution of a free stochastic differential equation.

\par The third quantum process we consider extracts rectangular blocks from the process $\munitaryfd$. Let $n\geq 1$ an integer and $d_{N}=(d^{1}_{N},\ldots,d^{n}_{N})$ a partition of $N$, which means:
\begin{equation*}
	1 \leq d^{i}_{N},~ d^{1}_{N}+\ldots+d^{n}_{N} = N,~\text{ for all } 1 \leq i \leq n.
\end{equation*}
The $H$-algebra $\runitaryalg$ we call \emph{rectangular dual group} is an object in the category $\Alg^{\star}(\mathcal{R})$, we defined it in Section \ref{zhangalgebras}. Recall that $\mathcal{R}$ is an unital involutive algebra generated by a \emph{complete set} of $n$ self-adjoint projectors $\{p_i,~1\leq i \leq n\}$:
$$
p_ip_j = \delta_{i=j} p_i~,~ \sum_{i=1}^n p_i = 1.
$$
The $\mathcal{R}$-amalgamated quantum stochastic process $\runitaryfdN$ over $\mathcal{R}\mathcal{O}\langle n \rangle$ takes values in the rectangular probability space $(\mathcal{M}(N,(L^{\infty-}(\Omega,\mathcal{F},\mathbb{P},\mathbb{K})), \mathbb{E}_{d_N})$ (see \cite{nico2}) over random matrices whose entries have finite moments of any order. For an integer $1 \leq i \neq N$, we let $p_{i}^{d_n}$ be the matrix in $\mathcal{M}(d_N,\mathbb{K})$ with coefficient $(i,i)$ set to $1$ and the other entries are set to $0$ and define:
\begin{equation*}
	\begin{array}{cccc}
		\runitaryfdN: & \runitaryalg & \mapsto & \mathcal{M}_{d_{N}}(L^{\infty-}(\Omega,\mathcal{F},\mathbb{P},\mathbb{K})) \\
		              & u            & \mapsto     & \munitaryfd             \\
		              & p_{i}        & \mapsto     & p^{d_{n}}_{i}.
	\end{array}
\end{equation*}
 We will use the shorter notation: 
$$
\mathcal{M}_{d_N}=\mathcal{M}_{d_N}(L^{\infty-}(\Omega,\mathcal{F},\mathbb{P},\mathbb{K}))
$$
The bimodule morphism $\mathbb{E}_{d_N}:\mathcal{M}(d_N,\to \mathcal{R})$ we be defined in due time, in Section \ref{sec:amalgamated}. In \cite{nico2}, we have proved that the \emph{above-defined processes of square and rectangular extractions converge in non-commutative distributions}, which means that for any $w\in\mathcal{O}\langle n \rangle$ and $w^{\prime} \in \mathcal{R}\mathcal{O}\langle n \rangle$,
$$
\mathbb{E}[\frac{1}{d}{\sf Tr}(U_{n,d}^{\mathbb{K}}(w))] \textrm{ and } \mathbb{E}_{d_N}(U_{d_N}^{\mathbb{K}}))(w^{\prime})]
$$
have limits when $N$ tends to infinity while $n$ ($nd=N$ and $n$ is the number of parts of $d_N$) is maintained constant.
The limiting distributions are free (with amalgamation over $\mathcal{R}$ for the rectangular extraction) semi-groups, which are in addition braid-invariant. We use these semi-groups in the forthcoming Sections. The limiting semi-group of the square extraction process will be realised as the non-commutative distribution of a free Lévy process, solution to a free stochastic differential equation.

\subsection{Higher dimensional Master Fields}
\label{higherdimmasterfield}
In \cite{nico2}, we define a higher dimensional counterpart of the free unitary Brownian motion \cite{biane1997free} as a solution to a free stochastic differential equation. More precisely, pick an integer $n\geq 1$ and a von Neumann algebra $\mathcal{A}$ endowed with a tracial state $\tau$, ($\tau(aba^{-1}) = \tau(a),~\tau(aa^{\star}) \geq 0,~\tau(a^{\star})=\overline{\tau(a)}$, $a,b \in A$). We define $\mfreeunitary=(\mfreeunitary(t))_{t\geq 0}$, with $\mfreeunitary(t) \in \mathcal{A} \otimes \mathcal{M}_{n}(\mathbb{C})$ as the process  whose matrix entries (in the algebra $\mathcal{A}$) are the solution of the following free stochastic differential system:
\begin{align*}
	&\text{d}{\sf U}^{\langle n \rangle}_{i,j}(t)=\frac{\sf i}{\sqrt{n}}\sum_{k=1}^n \text{d}\mdnoise_{i,k}({t}){\sf U}^{\langle n \rangle}_{k,j}(t) - \frac{1}{2}{\sf U}^{\langle n \rangle}_{i,j}(t)\text{dt},~t\geq 0,~1\leq i,j \leq n, \\
	&{\sf U}^{\langle n \rangle}(0) = {I}_{n}.
\end{align*}
In the equation above for integers $1 \leq i,j \leq n$, ${\sf W}_{i,j} = {\sf W}_{j,i}$ is a \emph{free Brownian motion} and $\{{\sf W}_{i,j},~1 \leq i<j \leq n\}$ is a mutually free family of free Brownian motions. We refer to \cite{biane1997free} for the global existence of the solution ${\sf U}^{\langle n \rangle}$ and unicity.
From ${\sf U}^{\langle n \rangle}$, we build the \emph{free unitary Brownian motion of dimension $n$}; for any time $t\geq 0$, we set:
\begin{equation*}
	\freeunitary(t):\mathcal{O}\langle n \rangle \to \mathcal{A},~\freeunitary(t)(u_{ij}) = \mfreeunitary(t)(i,j),~1 \leq i,j \leq n.
\end{equation*}
We leave to the reader the proof of the following Lemma.
\begin{lemma}
    The process $U^{\langle n \rangle}$ is a Lévy process on $\mathcal{O}\langle n \rangle$ with free increments.
\end{lemma}
The above lemma implies that, for $t\geq 0$, the distribution of $\freeunitary$ defined by 
$$
\tau^{\langle n \rangle}({t}) := \tau \circ \freeunitary(t): \mathcal{O}\langle n \rangle \to \mathbb{C}
$$
satisfies for any $s\geq 0$ (see Example \ref{ex:categomonoid}, fourth item):
$$
\tau^{\langle n \rangle}({s+t}) = \tau^{\langle n \rangle}({s}) \star \tau^{\langle n \rangle}({t}).
$$
Let $V$ be a unitary matrix in $\mathcal{M}_{n}(\mathbb{C})\otimes\mathcal{A}$ and set, for any time $t\geq 0$, 
$$[\mfreeunitary]^{V}(t)=V\mfreeunitary(t)V^{-1}.$$ 
We assume further that the involutive subalgebra of $\mathcal{A}$ generated by the entries of $V$ is free from the algebra generated by the entries of $\mfreeunitary(t)$ for all times $t\geq 0$. We defer the proof of the following Lemma to the Annexes.
\begin{lemma}
    \label{lemma:gaugeinv}
     Let $t \geq 0$ a time, let $u \in \udualgroup$ and $V$ be an unitary element of $\mathcal{M}_{n}(\mathbb{C})\otimes \mathcal{A}$ free from $\{ {\sf W}_{i,j}(t),t\geq 0,1 \leq i,j \leq n \}$. Then the non-commutative process
     $
     [{\sf W}]^V(t) = V{\sf W}(t)V^{-1},~t \geq 0,
     $
     has the same distribution as ${\sf W}(t)$ : ${\sf W}$ is gauge-invariant.
\end{lemma}

\begin{corollaire} Let $t \geq 0$ a time, let $u \in \udualgroup$ and $V$ be an unitary element of $\mathcal{M}_{n}(\mathbb{C})\otimes \mathcal{A}$, with the notations introduced so far,
	\begin{equation}
		\label{gaugeinvu}
		\tau \circ [\mfreeunitary]^{V}(t)(u)= \tau \circ \mfreeunitary(t)(u),~t\geq 0,~u\in \udualgroup.
	\end{equation}
\end{corollaire}

\begin{proof}
The process ${\sf U}^{{V}}$ is the solution of the following free stochastic differential system, with obvious notation,
	\begin{equation*}
		\text{d}[\mfreeunitary]^{V}(t)=\frac{\sf i}{\sqrt{n}}\text{d}\mdnoise^{V}(t)[\mfreeunitary]^{V}(t)-\frac{1}{2}[\mfreeunitary]^{V}(t),~t\geq 0.
	\end{equation*}
\end{proof}
We saw that for a classical Levy process on the $H$-algebra of function on a group, gauge-invariance and independence of increments imply braid-invariance. It seems difficult to prove the braid-invariance of $\freeunitary$ with the same arguments. However and as already highlighted, we have proved in \cite{nico2}, see also \cite{ulrich2015construction} that $\freeunitary$ is the limit in the non-commutative distribution of the braid-invariant process $U^{\mathbb{C}}_{n,d}$. Hence $U^{\langle n \rangle}$ is braid-invariant.

\begin{remarque}
	It would be nice to have a direct proof of braid-invariance of the Lévy process $U^{\langle n \rangle}$. In the case $n=1$, this is trivial and follows from traciality, see \cite{cebron2017generalized}. However, when $n>1$, this is not sufficient. To be more precise, we would have to prove that, for any polynomial $P$ in two non-commuting variables, pairs $1 \leq i,j \leq n$ and $1 \leq k,l \leq n$ and pair of times $0 < s < t$, that
	\begin{equation}
		\tau(P({\sf U}_t^{\langle n \rangle}{\sf U}^{\langle n \rangle}_{s,t}{\sf U}^{\langle n \rangle}_t{}^{\star}(i,j),{\sf U}^{\langle n \rangle}_t(k,l)) = \tau(P({\sf U}^{\langle n \rangle}_{s,t}(i,j),{\sf U}^{\langle n \rangle}_t(k,l)))
	\end{equation}
	In the case $n=1$, one can simply write:
	\begin{align*}
	\tau(P({\sf U}_t^{\langle 1 \rangle}{\sf U}^{\langle 1 \rangle}_{s,t}{\sf U}^{\langle 1 \rangle}_t{}^{\star},{\sf U}^{\langle 1 \rangle}_t)) &= \tau(P({\sf U}_t^{\langle 1 \rangle}{\sf U}^{\langle 1 \rangle}_{s,t}{\sf U}^{\langle 1 \rangle}_t{}^{\star},{\sf U}^{\langle 1 \rangle}_t{}{\sf U}^{\langle 1 \rangle}_t({\sf U}^{\langle 1 \rangle}_t)^{\star})) \\
	&= \tau({\sf U}_t^{\langle 1 \rangle}P({\sf U}^{\langle 1 \rangle}_{s,t},{\sf U}^{\langle 1 \rangle}_t)({\sf U}^{\langle 1 \rangle}_t)^{\star})\\ 
	&= \tau(P({\sf U}^{\langle 1 \rangle}_{s,t},{\sf U}^{\langle 1 \rangle}_t))
	\end{align*}
	where we have used the fact that $\tau$ is tracial in the last equality. 
		\end{remarque}
We can apply Theorem \ref{maintheoremzhanghol}: there exists a Quantum Holonomy Field over $\mathcal{O}\langle n\rangle$ built from the higher dimensional counterpart of the free unitary Brownian motion $\freeunitary$. We state it more formally in a Proposition.
\begin{proposition}
	\label{prop:higherdimmasterfield}
	Let $n\geq 1$ an integer. There exists a Quantum Holonomy Field (see Definition \ref{definition_master_field_categorical}), which we denote by ${\sf MF}^{\langle n \rangle}$, such that for any one parameter family of growing simple loops $\gamma$ with $|\gamma_{t}|=t$ and $\mathrm{Int}(\gamma_{s}) \subset \mathrm{Int}(\gamma_{t})$ for all times $0 \leq s \leq t$, the process ${\sf MF}^{\langle n \rangle}(\gamma_{t})$ has same distribution as $\freeunitary$, the free unitary Brownian motion of dimension $n$.
\end{proposition}

 Recall that by definition ${\sf MF}^{\langle n \rangle}$ is a group homomorphism from ${\sf RL}_{0}(\mathbb{R}^{2})$ to $\mathrm{Hom}_{\Alg^{\star}}(\mathcal{O}\langle n \rangle, \mathcal{A})$ and $\mathcal{A}$ is endowed with a certain tracial state $\tau_{\mathcal{A}}$. In the remaining part of this paper, we use notation $$\Phi^{\langle n\rangle}: {\sf RL}_0(\mathbb{R}^2)\to\mathbb{C},~\Phi^{\langle n\rangle}(\ell) = \tau \circ {\sf MF}^{\langle n \rangle}(\ell),~\ell \in {\sf RL}_0(\mathbb{R}^2).$$ 

 \begin{remarque}
 It is possible to extend $\Phi^{\langle n \rangle}$ to any rectifiable loop on the plane. One possible way to achieve this is to closely follow the method developed in \cite{cebron2017generalized} and to use again the result proved in \cite{nico2} about convergence in non-commutative distribution of $U^{\mathbb{K}}_{n,d}$ as the dimension $d$ of each block tends to infinity. We leave the details for future work.
  \end{remarque}
  We finish with simple computational examples.
  First, we want to compute $\Phi^{\langle n \rangle}(\ell)$, where $\ell$ is the loop that goes round the area $s$ one time in an anticlockwise manner and then goes round the outer loop enclosing the area $s+t$, also in an anticlockwise manner, see Fig. \ref{fig:exampleMM}. The decomposition into a product of lassos of this loop is given in Fig. \ref{fig:exampleMM2}. By freeness of the increments of $U^{\langle n \rangle}$,
\begin{align*}
\Phi^{\langle n \rangle}(\ell)(u_{11}) &= \tau({\sf U}^{\langle n \rangle}_s{\sf U}^{\langle n \rangle}_{s,s+t}{\sf U}^{\langle n \rangle}_{s}(u_{1,1})) \\
&=\tau({\sf U}^{\langle n \rangle}_s(u_{1,k}){\sf U}^{\langle n \rangle}_s(u_{l,1}))\tau({\sf U}^{\langle n \rangle}_{s,t}(u_{l,k}))
\end{align*}
\begin{figure}[!h]
	\centering
 \scalebox{0.4}{
	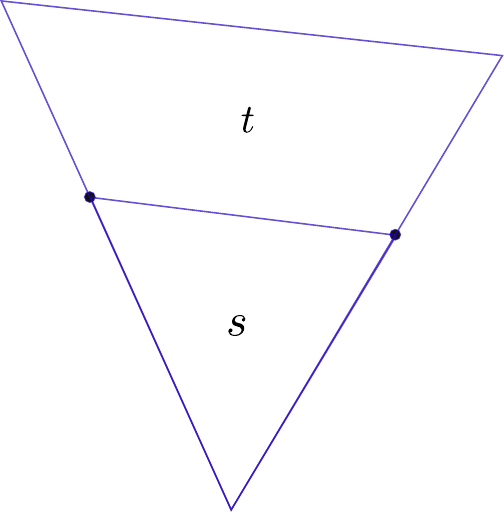
 }
	\caption{\label{fig:exampleMM} A very simple loop.}
\end{figure}
\begin{figure}[!h]
	\centering
 \scalebox{0.4}{
	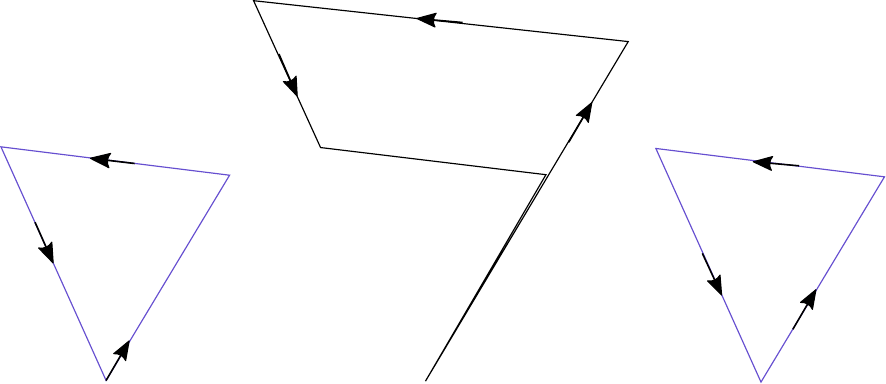
 }
	\caption{\label{fig:exampleMM2} Decomposition into lassos of the loop $\ell$.}
\end{figure}
A simple application of the free Ito formula, together with the free stochastic differential equation satisfied by ${\sf U}^{\langle n \rangle}$ shows that
\begin{equation}
	\label{eqn:expression}
\tau({\sf U}^{\langle n \rangle}_s(u_{1,k}){\sf U}^{\langle n \rangle}_s(u_{l,1})) = \frac{1}{n}(1-s)e^{-s}.
\end{equation}
From this last equation, we conclude that
$$\Phi^{\langle n \rangle}(\ell) = \frac{1}{n}n^{2}(1-s)e^{-s-\frac{t}{2}} = n(1-s)e^{-s-\frac{t}{2}}.$$
The above formula holds more generally for any couples $(a,k)$ and $(b,l)$ not just for $(1,k)$ and $(l,1)$.
We give a second example with a loop with two self-intersections; we compute the value of $\Phi^{\langle n \rangle}$ on the loop drawn on the left half of Fig. \ref{fig:second}. 
\begin{figure}[!h]
\centering
\scalebox{0.5}{
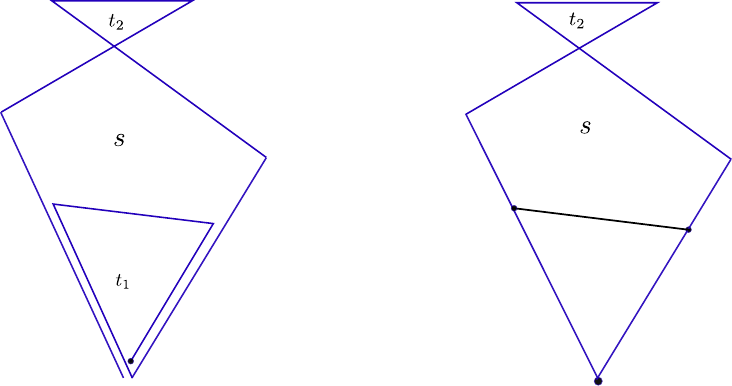
}
\caption{\label{fig:second} A second example. On the left half, the loop. We have slightly shifted the edges of the loop when they are crossed twice by the loop. On the right half is the corresponding graph.}
\end{figure}
We first decompose the loop as a product of (anticlockwise oriented) lassos going round the faces of the graph associated to the loop on the right half of Fig. \ref{fig:second}: it goes round one time in an anticlockwise manner the area $t_1$, move up and enclose the area $t_2$ in a clockwise manner and goes down to the point it started. The decomposition of this loop into a product of lassos is pictured in Fig. \ref{fig:decompoexd}.
\begin{figure}[!h]
\centering
\scalebox{0.6}{
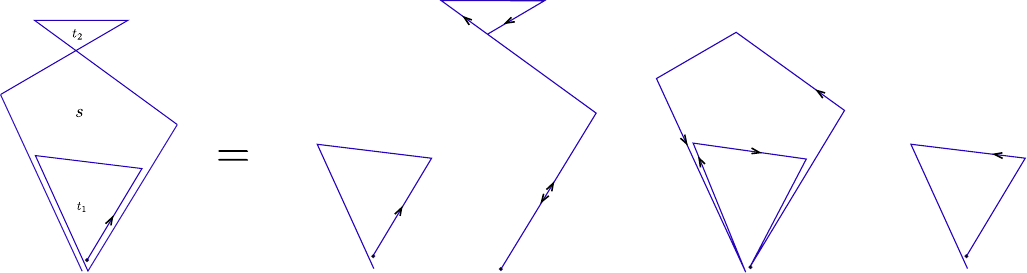
}
\caption{\label{fig:decompoexd} Decomposition of the loop on the left side as a product of lassos with disjoint bulks (on the right side.}
\end{figure}
The computation of $\Phi^{\langle n \rangle}(\ell)$ is then done in a similar way as in the previous example:
\begin{align*}
\Phi^{\langle n \rangle}(\ell)(u_{11}) &= \tau({\sf U}^{\langle n \rangle}_{t_1}(u_{1,k_1}){\sf U}^{\langle n \rangle}_{t_1,t_1+t_2}(u_{k_1,k_2}){\sf U}^{\langle n \rangle}_{t_1+t_2,t_1+t_2+s}(u_{k_2,k_3}){\sf U}^{\langle n \rangle}_{t_1}(u_{k_3,1})) \\
&=\tau({\sf U}^{\langle n \rangle}_{t_1}(u_{1,k_1}){\sf U}^{\langle n \rangle}_{t_1}(u_{k_3,1})) \tau({\sf U}^{\langle n \rangle}_{t_1,t_1+t_2}(u_{k_1,k_2})) \tau({\sf U}^{\langle n \rangle}_{t_1+t_2,t_1+t_2+s}(u_{k_2,k_3})) \\
&=\frac{1}{n}(1-t_1)e^{-t_1}\times n^3 \times e^{-\frac{t_2}{2}}e^{-\frac{s}{2}} = n^{2}(1-t_1)e^{-t_1-\frac{t_2}{2}-\frac{s}{2}}.
\end{align*}
Computing $\Phi^{\langle n \rangle}$ becomes very quickly intricated. For example, to continue making the list of the values of $\Phi^{\langle n \rangle}$ on loops with two self-intersections, we would consider the following loop in Fig \ref{fig:exampleMMt}. We apply the same method as in the two example above to compute the value of $\Phi^{\langle n \rangle}$ over this loop. This yields 
\begin{align*}
\Phi^{\langle n \rangle}(\ell)=n^2(1-t_1)(1-t-2)e^{-(t_1+t_2)-\frac{s}{2}}
\end{align*}
\begin{figure}[!h]
\centering
\scalebox{0.5}{
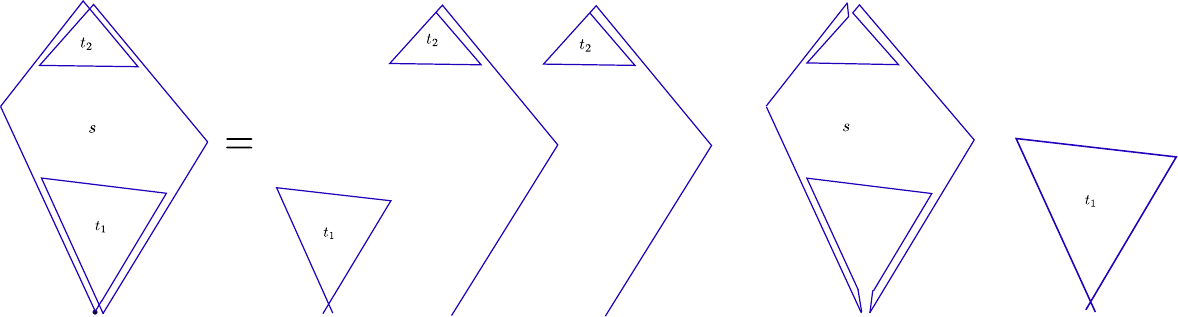
}
\caption{\label{fig:exampleMMt} A third example.}
\end{figure}
\begin{remarque}
From the previous computations, one postulates that $\varphi(\ell)$ is a monomial in $n$ of degree the number of self-intersections of $\ell$.
\end{remarque}
\subsection{Amalgamated Master Fields}
\label{sec:amalgamated}
For this Section, we will be very brief. Let $n\geq 1$ be an integer. In Section \ref{zhangalgebras} Example \ref{ex:halgebra}, we have defined the $H$-algebra $\runitaryalg$, see also the beginning of the current Section.
Let $(r_{1},\ldots,r_{n})$ be a sequence of positive real numbers such that $r_{1}+\cdots+r_{n}=1$. In \cite{nico2}, Section 6, we have introduced the free semi-group ${\sf E}_{(r_1,\ldots,r_n)}$ on the $H$-algebra $\mathcal{R}\mathcal{O}\langle n \rangle$ and prove that ${\sf E}_{(r_1,\ldots,r_n)}$ is a free with amalgamation over $\mathcal{R}$ semi-group. We explain briefly how this free semi-group is obtained.

We pick $d_N=(d_{N}^{1},\ldots,d_{N}^{n})$ an integer partition of $N$ into $n$ parts and we assume:
$$\frac{d_{N}^{i}}{N} \to r_{i},~ \mathrm{as~} N \to +\infty,~ 1 \leq i \leq n.$$ 
At the beginning of this Section, we defined the quantum process $\runitaryfdN$ on the $H$-algebra $\mathcal{R}\mathcal{O}\langle n \rangle$ extracting rectangular blocks from the unitary Brownian motion $\mathbb{U}(N,\mathbb{K})$. The bimodule algebra $\mathcal{M}_{d_N}$ over $\mathcal{R}$ is endowed with the following expectation:
\begin{equation*}
	\mathbb{E}_{d_{N}}(A) = \sum_{i=1}\frac{1}{d_{N}^{i}}\mathbb{E}\left[{\sf Tr}(p^{d_N}_{i}Ap^{d_n}_{i}) \right]p^{d_N}_{i},~ A \in \mathcal{M}_{d_{N}}.
\end{equation*}
The semi-group ${\sf E}_{(r_1,\ldots,r_n)}$ is the limit of the non-commutative distribution of the non-commutative process $\runitaryfdN$:
$${\sf E}_{d_{N}}=\mathbb{E}_{d_{N}} \circ \runitaryfdN(t)\xrightarrow{N\rightarrow +\infty} {\sf E}_{r_{1},\ldots,r_{n}}(t),~\text{ for all time } t\geq 0.$$
\par We apply our main Theorem \ref{maintheoremzhanghol} to obtain a Quantum Holonomy Field  associated with the free with amalgamation semi-group ${\sf E}_{r_1,\ldots,r_n}$ that we name \emph{amalgamated higher dimensional Master Field} with parameter $r_1,\ldots,r_n$. 
\begin{theorem}
	\label{amalgamatedfield}
	Let $r_{1},\ldots,r_{n}\geq 1$ positive real numbers summing to one. There exists a Quantum Holonomy Field which we denote ${\sf MF}^{\langle r_{1},\ldots,r_{n} \rangle}$, such that for all one-parameter growing family of simple loops $\gamma$ with $|\gamma_{t}|=t$ and $\mathrm{Int}(\gamma_{s}) \subset \mathrm{Int}(\gamma_{t})$ for all times $0 \leq s \leq t$, the process $\left({\sf MF}^{\langle r_{1},\ldots,r_{n} \rangle}(\gamma_{t})\right)_{t\geq 1}$ is a free with amalgamation over $\mathcal{R}$ quantum L\'evy process with non-commutative distribution at time $t\geq 0$ equal to ${\sf E}_{r_{1},\ldots,r_{n}}(t)$.
\end{theorem}

\section{Perspectives}
\label{sec:perspectives}
In this section, we would like to point out possible extensions of the work presented in the previous sections. We list \emph{remarks} indicating directions toward which further questions can be asked.

First, we have left open the question of extending to rectifiable loops, in the appropriate sense, the definitions constructions of Quantum and Categorical Holonomy Fields exposed above. In the case of generalized Master Fields of \cite{cebron2017generalized}, that is in our terminology in the case of Quantum Holonomy Field over the $H$-algebras $O\langle 1 \rangle$ or $\mathcal{F}(\mathbb{U}(N,\mathbb{C})$ and associated with free or classical quantum Lévy processes, this question has been answered in \cite{cebron2017generalized} in the positive. The question remains for the higher dimensional Master Field we built. It is however very likely that the method of \cite{cebron2014fluctuations} applies. In the abstract, it remains to build the appropriate framework to extend a Categorical Holonomy Field to rectifiable loops and to find the correct properties on the generator of the associated quantum Lévy process to be able to do so.

 Second, the direction of research we will address in a forthcoming paper deals with a rigorous proof of the Makeenko-Migdal equations we sketch now for the distribution ${\Phi}^{\langle n \rangle}$ of the free Master Field of dimension $n$. Again, we do not intend to give a rigorous derivation of these equations (based, for example, on the free stochastic differential that the higher dimensional free Brownian motion satisfies), but rather to provide candidates equations, supported by simple computations on a very simple configuration of loops. We will be very brief about the nature of these equations and how useful they are for computing the distribution of the Master Field ($\Phi^{\langle 1 \rangle}$ in our notations) see the very good survey \cite{https://doi.org/10.48550/arxiv.1912.06246}. The Makeenko-Migdal equations support a dynamical perspective on the computation of the distribution of the Master Field. Let $\ell$ be a reduced loop on the plane. We consider $\Phi^{\langle 1 \rangle}(\ell)$ as a function of the areas of the domains delimited by the curve $\ell$. The Makeenko-Migdal equations relate a certain combination of derivatives of $\Phi^{\langle 1 \rangle}(\ell)$ to the distribution of the Master Field evaluated on loops obtained by disentangling the initial loop at a simple intersection point, see Fig. 9 in\cite{https://doi.org/10.48550/arxiv.1912.06246}.
For the Master Field, so $n=1$, the equations read
\begin{equation}
\tag{${\sf MM}_{1}$}
\label{eqn:makeenko}
	[\frac{d}{dt_1}-\frac{d}{dt_2} + \frac{d}{dt_3} - \frac{d}{dt_4}]\Phi^{\langle n \rangle}(\ell) = \Phi^{\langle 1 \rangle}(\ell^{\prime})\Phi^{\langle 1 \rangle}(\ell^{\prime\prime})
\end{equation}
The above equations \eqref{eqn:makeenko} as the $U(\infty)$ or $\mathcal{O}\langle 1 \rangle$ Makeenko-Migdal equations. The original Makeenko-Migdal equations for the $U(N)$ Brownian Holonomy field (not just the planar sector), in any dimensions (not just in dimension two) along with a non-rigorous proof can be found in \cite{makeenko1979exact}, see \cite{kazakov1980non} for the formulation of the left-hand side of \eqref{eqn:makeenko} of the equations. The first rigorous proofs of these equations can be found in \cite{levy2008two}. Since the work of Lévy, more concise proofs appeared, see \cite{driver2017three} for the plane and \cite{driver2017makeenko} for any compact surfaces. We would propose the following equation for the distribution of the higher dimensional Master Field $\Phi^{\langle n \rangle}$:
\begin{equation}
\tag{${\sf MM}_{n}$}
	[\frac{d}{dt_1}-\frac{d}{dt_2} + \frac{d}{dt_3} - \frac{d}{dt_4}]\Phi^{\langle n \rangle}(\ell)(u_{i,j}) = \sum_{k=l}^n\Phi^{\langle n \rangle}(\ell^{\prime})(u_{i,l})\Phi^{\langle n \rangle}(\ell^{\prime\prime})(u_{l,j}) = \Phi^{\langle n \rangle}(\ell^{\prime}) \star \Phi^{\langle n \rangle}(\ell^{\prime\prime}),
\end{equation}
where $\star$ denotes the free convolution of functionals over $\mathcal{O}\langle n \rangle$.
Let us check this last equation on a very simple example, drawn in Fig. \ref{fig:exampleMM}. 
Applying the proposed Makeenko-Migdal equation to compute $\Phi^{\langle n \rangle}(\ell)$ yields the formula proved at then of Section \ref{higherdimmasterfield} since in our simple case, one obtains (taking into account that derivative with respect to the face adjacent to the unbounded face is $-\frac{1}{2}$ times $\Phi^{\langle n \rangle}$)
\begin{equation}
	(2\frac{d}{dt} - \frac{d}{ds}) \Phi^{\langle n \rangle}(\ell)(u_{1,1}) = - \Phi^{\langle n \rangle}(\ell)(u_{1,1}) - \frac{d}{ds} \Phi^{\langle n \rangle}(\ell)(u_{1,1}) =  ne^{-s-\frac{t}{2}}
\end{equation}
A strategy which would consist in essentially adapting the already existing proof using loop variables, as in \cite{driver2017makeenko}, section 3, Theorem 3.1. We should add, however, that all proofs the author is aware of, of the Makeenko-Migdal equations for the free Master Field (they are also named the $\mathbb{U}(\infty)$-Makeenko-Migdal equations) imply taking limit in the $\mathbb{U}(N)$-Makeenko--Migdal equations. To the extent of our knowledge, a direct proof of \eqref{eqn:makeenko} based on free stochastic calculus is still missing, even for the case $n=1$.

 A final question is concerned with direct proof of the braid-invariance of the distribution of the free Unitary Brownian motion of dimension $n$, a question that was left open in Section \ref{examples}.

\section{Annexes}
\subsection{Proof of Theorem \ref{thm:takeda}}
We prove Theorem \ref{thm:takeda} for the category of operator-valued probability space and recall the relevant definitions. The reader acquainted with category theory will not find any originality in the presentation of the material.


An upward-directed set is a set $S$ endowed with a preorder $\prec$ such that for any two elements $x,y \in S$ there exists a third element $z \in S$ greater to both $x$ and $y$; $x \prec z$ and $y \prec z$.
As an example, the set of finite sequences of lassos with disjoint bulks drawn on the plane equipped with $\prec$ defined in Section \ref{sec:catholofield} is an upward-directed set.

For the rest of this section, we closely follow the exposition made by Ziro Takeda in its seminal article \emph{direct limit and infinite direct products of $C^{\star}$-algebras} \cite{takeda1955inductive}. In particular, we use the outdated terminology of \emph{direct system} (a diagram in a category) over and \emph{direct limit} (colimit).

\begin{definition}
	Let $\Gamma$ be an upward-directed set. A \emph{direct system} over $\Gamma$ in a category $\mathcal{C}$ is the data of a family of objects $\{O_{\gamma},\gamma \in \Gamma\}$ in $\mathcal{C}$ and morphisms $f_{\alpha,\beta}$ for all couples $(\alpha,\beta)$ with $\alpha \prec \beta$ such that:
	\begin{enumerate}[\indent 1.]
		\item $\forall \alpha \in \Gamma$, $f_{\alpha,\alpha}=\id_{O_{\alpha}}$
		\item $\forall \alpha \prec \beta \leq \gamma$, $f_{\gamma,\beta}\circ f_{\beta,\alpha}=f_{\gamma,\alpha}$
	\end{enumerate}
\end{definition}
A direct system can alternatively be seen as a functor. In fact, the upward-directed set $\Gamma$ has associated a category, also denoted $\Gamma$ which class of objects is the set $\Gamma$. The set of homomorphisms ${\rm Hom}_{\Gamma}(\alpha,\beta)$ between two elements is either empty if $\alpha$ and $\beta$ are not comparable, either equal to the couple $(\alpha, \beta)$ if $\alpha \prec \beta$. With the notations of the last definition, the functor $O$ associated with the direct system is defined as:
\begin{equation*}
	O(\gamma) = O_{\gamma},~ O((\alpha,\beta)) = f_{\beta,\alpha}.
\end{equation*}
For the rest of this section, we fix an upward-directed set $\Gamma$.
\begin{definition}
	\label{definition_inductivelimit}
	Let $\mathcal{C}$ a category and let $O_{\alpha},~f_{\alpha,\beta}, \alpha \leq \beta$
	a direct system of $\mathcal{C}$ (A functor over the small category $\Gamma$). \emph{A direct limit} (a colimit) is the data of an object $O$ of $\mathcal{C}$ and morphisms $\phi_{\gamma}:O_{\gamma}\to O$ such that
	\begin{enumerate}[\indent 1.]
		\item $\phi_{\beta} \circ f_{\beta,\alpha} = \phi_{\alpha}$
		\item The following universal property holds. For all objects $Y \in \mathcal{C}$ and morphisms $g_{\gamma} : O_{\gamma}\to Y$ there exists a morphism $G:O\to Y$ such that the diagram in Fig. \ref{PU} is commutative for all pairs $\alpha \leq \beta$ in $\Gamma$.
		      \begin{figure}[!htb]\centering
			      \begin{tikzcd}
				      O_{\alpha} \arrow["\phi_{\alpha}"]{rd}\arrow["G_{\alpha}",swap]{rdd}\arrow["f_{\beta,\alpha}"]{rr} && O_{\beta} \arrow["\phi_{\beta}",swap]{ld}\arrow["G_{\beta}"]{ldd}\\
				      &O\arrow["G",pos=0.4]{d}& \\
				      &Y&
			      \end{tikzcd}
			      \caption{\label{PU} Universal property of the direct limit.}
		      \end{figure}
	\end{enumerate}
\end{definition}
From the universal property satisfied by direct limits, we see that a direct limit of a family of objects $\{A_{\gamma}, \gamma \in \Gamma\}$ is unique up to isomorphism. If any direct family in a category $\mathcal{C}$ admits a direct limit, we say that $\mathcal{C}$ is closed for taking direct limits, or using the language of category theory that $\mathcal{C}$ is \emph{inductively (co)complete}, see \cite{adamek2004abstract}.
\begin{theorem}
	The categories $\Prob(R)$, $\Alg^{\star}(R)$ and $\biMod(R)$ are inductively (co)complete.
\end{theorem}

\begin{proof}We prove only that $\Prob(R)$ is closed for taking direct limits. We let $\Gamma$ be a directed set and $(A_\alpha, f_{\alpha,\beta}),~\alpha,~\beta \in \Gamma$ be a direct system.
	Let $\mathcal{A}$ be the set of equivalence classes $\{[\left(\gamma,a_{\gamma}\right)],~ a_{\gamma} \in A_{\gamma},~ \gamma \in \Gamma\}$, with
	\begin{equation*}
		a_{\alpha} \sim a_{\beta} \Leftrightarrow \exists \delta \geq \alpha,\beta \textrm{ such that } f_{\delta,\alpha}(a_{\alpha}) = f_{\delta,\beta}(a_{\beta}).
	\end{equation*}
	Since the maps $f_{\gamma}~(\gamma \in \Gamma)$ are trace preserving, one has
	$$
		\tau_{\delta}(f_{\delta,\alpha}(a_{\alpha}))= \tau_{\delta}\left(f_{\delta,\beta}(a_{\beta})\right) = \tau_{F}\left(a_{F}\right) \textrm{ for all pairs } (a_{\alpha},a_{\beta}) \textrm{ with } a_{\alpha} \sim a_{\beta},
	$$
	hence, the function $(\gamma,a_{\gamma}) \rightarrow \phi_{\gamma}(a_{\gamma})$ is constant on the classes for $\sim$ and thus descends to a linear form $\tau$ on the quotient space $\bigsqcup_{F}\mathcal{A}_{F}/\sim$.
	The algebraic operations on $\mathcal{A}$ are defined as follows
	\begin{enumerate}[\indent 1.]
		\item Addition: $[x_{\alpha}] + [x_{\beta}] = [x_{\delta} + y_{\delta}]$, with $\delta \geq \alpha,~\delta \geq \beta$, $x_{\delta} = f_{\delta,\alpha}(x_{\alpha})$ and $y_{\delta} = f_{\delta,\beta}(x_{\beta})$.
		\item Multiplication: $[x_{\alpha}] \cdot [y_{\beta}] = [x_{\delta} \cdot y_{\delta}]$,
		\item Star operation: $[x_{\alpha}]^{\star} = [x_{\alpha}^{\star}]$.
		\item To define the bimodule structure on $\mathcal{A}$ over $R$, we simply set:
		      \begin{equation*}
			      r[x_{\alpha}]r^{\prime} = [rx_{\alpha}r^{\prime}],~r,r^{\prime}\in R.
		      \end{equation*}
		      In fact, if $[x_{\alpha}] = [x_{\beta}]$ with $\beta \geq \alpha$, then $x_{\beta} = f_{\beta,\alpha}(x_{\alpha})$ and $rx_{\beta}r^{\prime} = rf_{\beta,\alpha}(x_{\alpha})r^{\prime} = f(rx_{\alpha}r^{\prime})$ for $r,r^{\prime} \in  R$
	\end{enumerate}

\end{proof}
\subsection{Proof of Lemma \ref{lemma:gaugeinv}}
\begin{lemma}
     Let $t \geq 0$ a time, let $u \in \udualgroup$ and $V$ be an unitary element of $\mathcal{M}_{n}(\mathbb{C})\otimes \mathcal{A}$, with the notations introduced so far,
     $
     W^{V}(t) 
     $
     has same distribution as $W(t)$.
\end{lemma}

\begin{proof}
	
 Let $m\geq 1$ be an integer, we denote by ${\sf NC}_{2m}$ the set of matchings of the interval $\llbracket 1, 2m \rrbracket$. A matching ${\sf m} \in {\sf NC}^{(2)}_{2m}$ is alternatively seen as a non-crossing partition or as an involution of $\llbracket 1,2m \rrbracket$ verifying:
	\begin{equation*}
		\text{for all } k < l \in \llbracket 1,2m \rrbracket,~ {\sf m}(k) < {\sf m}(l).
	\end{equation*}
	We compute the cumulants of the family $\{\mdnoise^{V}(\alpha,\beta),~ 1 \leq \alpha, \beta \leq n \}$ and prove that:
	\begin{equation}
		\label{cumulantsgaugeinvariance}
		\begin{split}
			&k_{2m+1}\left(\mdnoise^{V}(\alpha_{1},\beta_{1}),\ldots,\mdnoise^{V}(\alpha_{2m+1},\beta_{\beta_{2m+1}})\right) = 0,~m\geq 1, \text{ for all } \alpha,\beta \in \{1,\ldots,n\}^{2m+1}, \\
			&k_{2m}\left(\mdnoise^{V}(\alpha_{1},\beta_{1}),\ldots,\mdnoise^{V}(\alpha_{2m+1},\beta_{\beta_{2m}})\right)=k_{2m}\left(\mdnoise(\alpha_{1},\beta_{1}),\ldots,\mdnoise(\alpha_{2m},\beta_{\beta_{2m}})\right) \\
			&\phantom{k_{2m}\left(\mdnoise^{V}(\alpha_{1},\beta_{1}),\ldots,\mdnoise^{V}(\alpha_{2m+1},\beta_{\beta_{2m}})\right)}=\sum_{{\sf m}\in{\sf NC}_{2m}} \prod_{i=1}^{2m}\delta_{\alpha_{i},\beta_{{\sf m}(i)}}\delta_{\alpha_{{\sf m}(i)},\beta_{i}},~\alpha,\beta \in \{1,\ldots,n\}^{2m}.
		\end{split}
	\end{equation}
	To compute these cumulants, we use the following fundamental formula relating cumulants with products of elements of $\mathcal{A}$ as entries to the cumulants of these elements. Let $p\geq 1$ be an integer and $k_{p}\geq 1$ be another one. We define the interval partition $\sigma = \{\{1,\ldots,k_{1}\},\{k_{1}+1,\ldots,k_{2}\},\ldots,\{k_{p-1}+1,\ldots,k_{p}\}\}$ and denote by $1_{k_{p}}$ the partition of $\llbracket 1,k_{p} \rrbracket$ with only one block. Let $a_{1},\ldots,a_{k_{p}} \in \mathcal{A}$, then:
	\begin{equation}
		\label{formulaproductcumulant}
		k_{k_{p}}(a_{1} \cdots a_{k_{1}},a_{k_{1}+1} \cdots a_{k_{2}},\ldots, a_{k_{p-1}+1} \cdots a_{k_{p}}) = \sum_{\substack{\pi \in {\sf NC}_{k_{p}} \\ \pi \vee \sigma=1_{k_{p}}}}k_{\pi}(a_{1},\ldots,a_{k_{p}}).
	\end{equation}
	By using formula $\eqref{formulaproductcumulant}$ and freeness of the matricial entries of $V$ with $\mdnoise_{t}$, it is easy to prove that the odd cumulants in \eqref{cumulantsgaugeinvariance} are equal to zero. Let $p \geq 1$ and $\alpha_{1},\ldots,\alpha_{2p} \in \{1,\ldots,n\}^{2p}$.  Consider a word $u=v_{1}w_{1}\tilde{v}_{1} \cdots v_{1}\tilde{w}\tilde{v}_{1})$ with the $v^{'s}$ and the $\tilde{v}^{s}$ in the algebra generated by the matrix coefficients of $V$ and the $w^{'s}$ in the algebra generated by the matrix coefficients of $W_{t}$. We say that an integer $i\leq 3p$ is white coloured if $u_{i}$ is equal to $v_{i}$ or $v^{\prime}_{i}$ and black coloured if $u_{i}$ is equal to $u_{i}$.

	Let $\sigma$ be the interval partition $\sigma=\{\{1,2,3\},\ldots,\{6p-2,6p-1,6p\}\}$. Let $\pi \in {\sf NC}_{6p}$ such that $\sigma\vee\pi=1_{6p}$. By using nullity of mixed cumulants having components of $V$ and $\mdnoise_{t}$ as entries, we prove that $$k_{\pi}(V(\alpha_{1},k_{1}),\mdnoise(k_{1},q_{1}),V^{\star}(q_{1},\beta_{1}),\ldots, V(\alpha_{2p},k_{2p}),\mdnoise(k_{2p},q_{2p}),V^{\star}(q_{2p},\beta_{2p}))$$
	is equal to zero if a block of $\pi$ is white a black coloured. Also, the trace of $\pi$ on the set of black coloured integers is a matching ${\sf m}$, in that case,
	\begin{equation}
		\label{tosumun}
		\begin{split}
			&k_{\pi}(V(\alpha_{1},k_{1}),\mdnoise(k_{1},q_{1}),V^{\star}(q_{1},\beta_{1}),\ldots, V(\alpha_{2p},k_{2p}),\mdnoise(k_{2p},q_{2p}),V^{\star}(q_{2p},\beta_{2p}))= \\
			&\hspace{2cm}k_{{\sf m}}(W(k_{1},q_{1}),\ldots,W(k_{2p},q_{2p}))k_{{\sf K}({\sf m},\pi)}(V(\alpha_{1},k_{1}),V^{\star}(q_{1},\beta_{1}),\ldots,V^{\star}(q_{2p},\beta_{2p})).
		\end{split}
	\end{equation}
	The non crossing partition denoted ${\sf K}({\sf m},\pi)$ is a partition of the white coloured integers of $\llbracket 1,6p \rrbracket$ and equal to the complement of ${\sf m}$ in the partition $\pi$: ${\sf m} \cup {\sf K}({\sf m},\pi)=\pi$. We sum \eqref{tosumun} over non crossing partitions having the same trace ${\sf m}\in{\sf NC}^{2}_{2p}$ over black coloured integers and over integers $k_{1},\ldots,k_{2p}$,$q_{1},\ldots,q_{2p}$ in $\{1,\ldots,n\}^{2}$. By using the moments-cumulants formula, we obtain:
	\begin{equation*}
		\begin{split}
			&\sum_{\substack{1 \leq k_{1},\ldots,k_{2p}\leq n, \\ 1 \leq q_{1},\ldots,q_{2p} \leq n}}
			\sum_{\substack{\pi\in{\sf NC}_{6p} \\ \pi_{\bullet}={\sf m}}}k_{\pi}(V(\alpha_{1},k_{1}),\mdnoise(k_{1},q_{1}),V^{\star}(q_{1},\beta_{1}),\ldots, V(\alpha_{2p},k_{2p}),\mdnoise(k_{2p},q_{2p}),V^{\star}(q_{2p},\beta_{2p})) \\
			&=\sum_{\substack{1 \leq k_{1},\ldots,k_{2p}\leq n, \\ 1 \leq q_{1},\ldots,q_{2p} \leq n}}k_{{\sf m}}(W(k_{1},q_{1}),\ldots,W(k_{2p},q_{2p}))\tau_{{\sf K}({\sf m},1_{6p})}(V(\alpha_{1},k_{1}),V^{\star}(q_{1},\beta_{1}),\ldots,V^{\star}(q_{2p},\beta_{2p}))\\
			&=\sum_{\substack{1 \leq k_{1},\ldots,k_{2p}\leq n, \\ 1 \leq q_{1},\ldots,q_{2p} \leq n}}\prod_{i=1}^{2p}\delta_{k_{i},q_{{\sf m}(i)}}\delta_{q_{i},k_{{\sf m}(i)}}\tau_{{\sf K}({\sf m},1_{6p})}(V(\alpha_{1},k_{1}),V^{\star}(q_{1},\beta_{1}),\ldots,V^{\star}(q_{2p},\beta_{2p})) \\
		\end{split}
	\end{equation*}
	To compute the right-hand side of the last equation, pick a block $V$ of the partition ${\sf K}({\sf m},1_{6p})$, then by using traciality of $\tau$, we have:
	\begin{equation}
		\sum_{\substack{1 \leq k_{1},\ldots,k_{2p}\leq n, \\ 1 \leq q_{1},\ldots,q_{2p} \leq n}}\tau_{V}(V^{\star}(q_{i_{1}},\beta_{i_{1}})V(\alpha_{i_{2}},k_{i_{2}}) \cdots V(\alpha_{i_{l}},k_{{\sf m}(i_{l})})\prod_{l}\delta_{k_{i_{l}},q_{{\sf m}(i_{l})}}\delta_{q_{i_{l}},k_{{\sf m}(i_{l})}} = \prod_{l}\delta_{\alpha_{l},\beta_{{\sf m}(l)}}\delta_{\beta_{l},\alpha_{{\sf m}(l)}}.
	\end{equation}
	The proof of the formulas \eqref{cumulantsgaugeinvariance} is now complete.
\end{proof}
\subsection{Proof of Lemma \ref{lemma_alge_alg_alg_star}}
\begin{lemma}\label{lemma_alge_alg_alg_star}
	The category $\rCoMod\mathcal C(H)$ is an algebraic category.
\end{lemma}

\label{annex:prooflemmacinq}
 It is a simple verification that $k$ is a right $H$-comodule-algebra and that for every $H$-comodule-algebra $M$, the unique morphism in $\mathcal C$ from $k$ to $\mathcal C$ satisfies \eqref{eq:morcomod}, hence is a morphism in the category of right $H$-comodule-algebras. This shows that $k$ is an initial element of $\rCoMod\mathcal C(H)$.

	We must now define a coproduct in $\rCoMod\mathcal C(H)$. Let $(M,\Omega_{M})$ and $(N,\Omega_{N})$ be two right $H$-comodules. We will endow the object $M\free N$ of $\mathcal C$ with a co-action of $H$. For this, we start from the map $\Omega_{M}\free \Omega_{N}:M\free N \to M\free H \free N\free H$. We will compose this map with a map which, informally, forgets the origin of the factors belonging to $H$. Pictorially, we want a morphism $M\free {\color{cyan} H_{|1}} \free N\free {\color{magenta} H_{|2}}\to M\free N\free H$ which sends, for instance, $n{\color{magenta} h_{|2}}n'm{\color{cyan} h'_{|1}}m'n''{\color{cyan} h''_{|1}}{\color{magenta} h'''_{|2}}$ to $nhn'mh'm'n''(h''h''')$. This map is built from the canonical maps $\iota_{M}:M\to M\free N\free H$, $\iota_{N}:N\to M\free N\free H$, $\iota_{H}:H\to M\free N\free H$ by the formula
	\[\Omega_{M\free N}=(\iota_{M}\freem \iota_{H} \freem \iota_{N} \freem \iota_{H})\circ (\Omega_{M}\free \Omega_{N}).\]
	We claim that $(M\free N, \Omega_{M\free N})$ is a coproduct of $(M,\Omega_{M})$ and $(N,\Omega_{N})$. The fact that $\Omega_{M\free N}$ is a morphism in $\mathcal C$ follows from its very definition. There remains to prove that it satisfies the equalities \eqref{relationcomod}. Let us treat the first equality in detail.

	We begin by drawing (see Fig. \ref{omegamun}) the diagram associated with the universal problem of which the pair $(M\sqcup N,\Omega_{M\sqcup N})$ is the solution. \begin{figure}[!htb]\centering
		\usetikzlibrary{cd}
		\begin{tikzcd}
			N \arrow[blue,"\Omega_{M}"]{r}\arrow{rd}& N \sqcup H \arrow{rd}\\
			& M\sqcup N \arrow[blue,"\Omega_{M\sqcup N}",shift left=0]{r}&M\sqcup N \sqcup H \\
			M \arrow[blue,"\Omega_{N}",swap]{r}\arrow{ru}& M\sqcup H \arrow{ru}
		\end{tikzcd}
		\caption{\label{omegamun}\small The universal problem solved by $(M\sqcup N,\Omega_{M\sqcup N})$.}
	\end{figure}

	Using this universal property and the fact that $\Omega_{M}$ and $\Omega_{N}$ are co-actions, we draw a second diagram (see Fig. \ref{omegadeuxpu}) in which a map  $f:M\free N \to M\free N \free H \free H$ appears, that is the coproduct of the two morphisms obtained by composition of the arrows at the top, respectively at the bottom of the diagram. Choosing the right or orange arrow in each composition does not matter, see Definition \ref{def:comodZhang}. We claim that the two maps of which we want to prove the equality, namely $(\Omega_{M\sqcup N}\sqcup \id_{H}) \circ \Omega_{M\sqcup N}$  and $(\id_{M\free N}\free \Delta) \circ \Omega_{M\sqcup N}$, are equal to this map.

	\begin{figure}[!htb]\centering
		\usetikzlibrary{cd}
		\begin{tikzcd}
			N \arrow[blue]{r}\arrow{rd}& N \sqcup H \arrow[shift right,swap,orange,"\Delta"]{r}\arrow[shift left, blue,"\Omega" ]{r}&N\sqcup \tilde{H} \sqcup H\arrow{rd}\\
			& M\sqcup N \arrow[violet,"f"]{rr}&&M\sqcup N\sqcup\tilde{H}\sqcup H \\
			M \arrow[blue]{r}\arrow{ru}& M\sqcup H\arrow[shift right,orange,swap,"\Delta"]{r}\arrow[shift left,blue]{r}& M\sqcup\tilde{H}\sqcup H \arrow{ru}
		\end{tikzcd}
		\caption{\label{omegadeuxpu}\small In this diagram, blue arrows indicate a co-action of $H$ and orange arrows a coproduct of $H$. We use the notation $\tilde H$ for the sake of clarity.
		}
	\end{figure}

	This is done by two more diagrams. The first (Fig. \ref{omegaomega1}) shows the map $(\Omega_{M\sqcup N}\sqcup \id_{H}) \circ \Omega_{M\sqcup N}$. The commutativity of this diagram and its comparison with Fig. \ref{omegadeuxpu} shows that this map is indeed equal to the map $f$.
	\begin{figure}[!htb]\centering
		\usetikzlibrary{cd}
		\begin{tikzcd}
			N \arrow[blue]{r}\arrow{rd}& N \sqcup H \arrow{rd}\arrow[blue,shift left]{r}&N\sqcup \tilde{H} \sqcup H\arrow{rd}\\
			& M\sqcup N \arrow[blue]{r}&M\sqcup N \sqcup H \arrow[blue,right]{r}&M\sqcup N\sqcup\tilde{H}\sqcup H \\
			M \arrow[blue]{r}\arrow{ru}& M\sqcup H \arrow{ru} \arrow[blue,shift right]{r}& M\sqcup\tilde{H}\sqcup H \arrow{ru}
		\end{tikzcd}
		\caption{\label{omegaomega1}\small This diagram, in which blue arrows correspond to the co-action of $H$, shows the map $(\Omega_{M\sqcup N}\sqcup \id_{H}) \circ \Omega_{M\sqcup N}$.}
	\end{figure}

	For the second map, the diagram that we draw (Fig. \ref{omegaomega2}) has four squares and the commutativity of the rightmost two needs to be checked. This is a simple verification that we leave to the reader. The equality of $(\id_{M\free N}\free \Delta) \circ \Omega_{M\sqcup N}$ with the map $f$, and the fact that $\Omega_{M\free N}$ satisfies the first equality of \eqref{relationcomod}, follows immediately.
	\begin{figure}[!htb]\centering
		\usetikzlibrary{cd}
		\begin{tikzcd}
			N \arrow[blue]{r}\arrow{rd}& N \sqcup H \arrow{rd}\arrow[orange]{r}&N\sqcup \tilde{H} \sqcup H\arrow{rd}\\
			& M\sqcup N \arrow[blue]{r}&M\sqcup N \sqcup H \arrow[orange]{r}&M\sqcup N\sqcup\tilde{H}\sqcup H \\
			M \arrow[blue]{r}\arrow{ru}& M\sqcup H \arrow{ru} \arrow[orange]{r}& M\sqcup\tilde{H}\sqcup H \arrow{ru}
		\end{tikzcd}
		\caption{\label{omegaomega2}\small In this diagram, orange arrows correspond to the coproduct of $H$. We see the map $(\id_{M\free N}\free \Delta) \circ \Omega_{M\sqcup N}$ in the middle line. }
	\end{figure}

	The proof of the second equality of \eqref{relationcomod} is similar and simpler than the proof of the first, and we leave it to the reader.

	We are not only stating that the coproduct of two comodules is a comodule, but also that the category of all comodules over the $H$-algebra $H$ is algebraic. Three points remain to be proved. The first is the equivariance of the coproduct of two equivariant morphisms. The second is the fact that $k$, the initial object of $\mathcal C$, can be endowed with a co-action of $H$. The third is that, with respect to this co-action, the unique morphism from $k$ to any comodule-algebra $M$ is equivariant.

	Let us discuss the first point. Let $(M,\Omega_{M}), (N,\Omega_{N}), (C,\Omega_{C})$ be three right comodules algebras. Let $f: M\to C$ and $g:N\to C$ be two morphisms of the category rCoMod$\mathcal{C}(H)$. We claim that $f\freem g$ is equivariant with respect to co-actions $\Omega_{M\free N}$ and $\Omega_{C}$, which means that the equality $\Omega_{C}\circ (f\freem g) =(\id_{H} \sqcup (f\freem g))\circ \Omega_{M\free N}$ holds. Fig. \ref{equivariancepu} shows a diagram in which, as before, blue arrows indicate the co-action of $H$. The equivariance of $f \freem g$ is equivalent to the commutativity, in this diagram, of the face delimited by the grey bended arrow and the horizontal symmetry axis of the diagram. This commutativity property will be implied by the commutativity of the two outer faces bounded by the violet and grey arrows. This commutativity, in turn, is implied by the associativity of the free product, drawn in Fig. \ref{associativityprod}.
	\begin{figure}[!htb]\centering
		\usetikzlibrary{cd}
		\begin{tikzcd}[every matrix/.append style={name=m},
				execute at end picture={
						\draw [<-,>=latex] ([yshift=7mm,xshift = 6mm]m-3-1.east) arc[start angle=-180,delta angle=270,radius=0.15cm];
						\draw [<-] ([yshift=7mm,xshift = 5mm]m-3-2.east) arc[start angle=-90,delta angle=270,radius=0.15cm];
						\draw [<-] ([yshift=2mm,xshift = 11mm]m-4-2.east) arc[start angle=-90,delta angle=270,radius=0.15cm];
						\draw [<-] ([yshift=-6mm,xshift = 17.4mm]m-3-1.south) arc[start angle=-90,delta angle=270,radius=0.15cm];
					}]
			&\arrow[violet]{ldd}\arrow["f \sqcup\mathrm{id}_{H}",violet]{rrdd}M\sqcup H&&&\\
			&M\arrow{u}\arrow{d}\arrow["f" ]{rd}&&&\\
			M\sqcup N\sqcup H\arrow["f\sqcup g \sqcup \mathrm{id}_{H}",swap,pos=0.3, bend right = 130,gray]{rrr}&\arrow[blue]{l} M \sqcup N \arrow{r}&C\arrow[blue]{r}&C \sqcup H
			\\
			&N\arrow[swap,"g"]{ru}\arrow{d} \arrow{u}\\
			&N\sqcup H\arrow[swap,"g \sqcup\mathrm{id}_{H}",violet]{rruu}\arrow[violet]{luu}
		\end{tikzcd}
		\caption{\label{equivariancepu}\small Equivariance of the map $f \freem g$. The morphisms drawn in violet and grey are equal to the morphisms drawn with the same colour in Figure \ref{associativityprod}.}
	\end{figure}
	\begin{figure}[!htb]\centering
		\usetikzlibrary{cd}
		\begin{tikzcd}[every matrix/.append style={name=m},
				execute at end picture={}]
			&C\sqcup H \\
			\\
			M\sqcup H\arrow[violet]{rd} \arrow[violet,shorten >= 3mm]{ruu}         &H\arrow{d}\arrow{l}\arrow{r} \arrow[]{uu}                &   N\sqcup H\arrow[violet]{ld}\arrow[violet,shorten >= 3mm]{luu}\\
			M \arrow{r} \arrow{u}\arrow[crossing over,shorten >=3mm]{uuur} & M \sqcup N \sqcup H\arrow[gray,xshift=-2.5mm]{uuu}\arrow[gray,xshift=2.5mm]{uuu} & N\arrow{l}\arrow{u}     \arrow[crossing over,shorten >=3.5mm]{uuul}            \\
			& M\sqcup N\arrow{u}\arrow{ur}\arrow{lu} & \\
		\end{tikzcd}
		\caption{\label{associativityprod}\small Associativity of the free product.}
	\end{figure}

	The second and third points concern the initial object $k$. There exists an unique morphism $\eta_{k\free H}:k \to k\free H$ and we prove that $(k,\eta_{k\free H})$  is an object of $\rCoMod\mathcal{C}(H)$. The two morphisms $(\id_{H}\free \eta_{k\free H})\circ\eta_{k\free H}$ and $(\Delta \free \id_{H}) \circ \eta_{k\free H}$ are equal because there is a unique morphism from $k$ to $k\free H\free H$. For the same reason, $(\varepsilon \free \id_{k}) \circ \eta_{k\free H}=\id_{k}$. The third point, that the unique map from $k$ to any comodule-algebra $M$ is equivariant, also follows from the same argument. The proof is complete.

\bibliographystyle{elsarticle-num}

\bibliography{biblio}

\end{document}